\documentclass[11pt,uplatex]{article}
\usepackage{graphicx}
\usepackage{color}
\usepackage{amsmath, amssymb, amsthm, comment, enumerate, ascmac, bm}
\usepackage{latexsym, lscape, mathrsfs, mathtools, psfrag, setspace, subcaption}
\usepackage{newtxtext}
\usepackage{natbib}
\usepackage[stable]{footmisc}
\usepackage{multirow}
\usepackage[colorlinks=true, linkcolor=blue, filecolor=blue, urlcolor=blue, citecolor=blue, driverfallback=dvipdfmx]{hyperref}
\usepackage[margin = 2.2cm]{geometry}
\usepackage{clipboard}
\newclipboard{manuscript}

\DeclareMathOperator*{\argmax}{argmax}

\renewcommand{\hat}{\widehat}
\renewcommand{\tilde}{\widetilde}
\newcommand{\bE}{\mathbb E}

\allowdisplaybreaks[1]

\theoremstyle{definition}
\newtheorem{theorem}{Theorem}[section]
\newtheorem{assumption}{Assumption}[section]
\newtheorem{corollary}{Corollary}[section]

\newtheorem{lemma}{Lemma}[section]

\newtheorem{remark}{Remark}[section]

\numberwithin{equation}{section}

\title{Estimating Marginal Treatment Effects under Unobserved Group Heterogeneity}

\author{Tadao Hoshino\thanks{School of Political Science and Economics, Waseda University, 1-6-1 Nishi-waseda, Shinjuku-ku, Tokyo 169-8050, Japan. Email: \href{mailto:thoshino@waseda.jp}{thoshino@waseda.jp}.} \
	and Takahide Yanagi\thanks{Graduate School of Economics, Kyoto University, Yoshida Honmachi, Sakyo, Kyoto, 606-8501, Japan. Email: \href{mailto:yanagi@econ.kyoto-u.ac.jp}{yanagi@econ.kyoto-u.ac.jp}.}}

\date{This version: May 2022 \hspace{1cm} First version: January 2020}

\begin{document}

\onehalfspacing

\maketitle

\begin{abstract}
	This paper studies treatment effect models in which individuals are classified into unobserved groups based on heterogeneous treatment rules.
	Using a finite mixture approach, we propose a marginal treatment effect (MTE) framework in which the treatment choice and outcome equations can be heterogeneous across groups.
	Under the availability of instrumental variables specific to each group, we show that the MTE for each group can be separately identified.
	Based on our identification result, we propose a two-step semiparametric procedure for estimating the group-wise MTE.
	We illustrate the usefulness of the proposed method with an application to economic returns to college education.
		
	\bigskip
		
\end{abstract}
	
\newpage 
	
%%%%%%%%%%%%%%%%%%%%%%%%%%%%%%%%%%%%%%%%%%%%%%%%%%%%%
	
\section{Introduction}

Assessing heterogeneity in treatment effects is an important issue for precise treatment evaluation.
The marginal treatment effect (MTE) framework (\citealp{heckman1999local, heckman2005structural}) has been increasingly popular in the literature as it provides us with rich information about the treatment heterogeneity not only in terms of observed individual characteristics, but also in terms of unobserved individual cost of the treatment.
Moreover, once the MTE is estimated, it can be used to build other treatment parameters such as the average treatment effect (ATE) and the local average treatment effect (LATE).
For recent developments on the MTE approach, see, for example, \cite{cornelissen2016late}, \citet{lee2018identifying}, \citet{mogstad2018using}, and \cite{mogstad2018identification}.
%The MTE framework is useful in several respects.
%First, the framework can be used in non-experimental applications in which individuals may endogenously determine their own treatment status.
%The endogeneity of the treatment can be dealt with using the method of local \textit{instrumental variables} (IVs).

While the conventional treatment evaluation methods can address heterogeneity across observable groups of individuals, many applications may exhibit ``unobserved'' group-wise heterogeneity in treatment effects for various reasons.
For example, the presence of multiple treatment eligibility criteria may create unobserved groups.
As a typical application, consider evaluating the causal effect of college education.
Since schools typically offer a variety of admissions options such as entrance exams and sports referrals, this process classifies individuals into several groups, and the admission criteria to which each individual has applied is typically unknown to researchers.
Such differences in admissions requirements may result in heterogeneous treatment effects of college education.
Another potential reason for the presence of unobserved group heterogeneity is that the population may be composed of groups with different preference patterns.
For instance, consider estimating the causal effect of foster care for abused children, as in \citet{doyle2007child}.
Here, the treatment variable of interest is whether the child is put into foster care by the child protection investigator.
The author discusses a possibility that the child protection investigators may have different preference patterns that place relative emphasis on child protection.
These examples suggest unobserved group patterns in the treatment choice process that may lead to some heterogeneity in treatment effects.
\bigskip

In this paper, we study endogenous treatment effect models in which individuals are grouped into latent subpopulations, where the presence of the latent groups is accounted for by a finite mixture model.
Finite mixture approaches have been successfully used in various fields to analyze data from heterogeneous subpopulations (\citealp{mclachlan2004finite}).
For example, many empirical studies on economics employ finite mixture approaches to address unobserved group heterogeneity (e.g., \citealp{keane1997career}; \citealp{cameron1998life}).
However, the use of finite mixture models in treatment evaluation has been considered only in a few specific applications (e.g., \citealp{harris2009impact}; \citealp{munkin2010disentangling}; \citealp{deb2018heterogeneous}; \citealp{samoilenko2018using}).
Compared with these studies, our modeling approach is applicable to various contexts and formally builds on \citeauthor{rubin1974estimating}'s (\citeyear{rubin1974estimating}) causal model by directly extending it to finite mixture models.

For this model, we develop identification and estimation procedure for the MTE parameters that can be unique to each latent group.
The proposed group-wise MTE is a novel framework in the literature, which should be informative for understanding the heterogeneous nature of treatment effects by capturing both group-level and individual-level unobserved heterogeneity simultaneously.
Importantly, as we discuss below, the presence of unobserved group heterogeneity threatens the validity of the conventional instrumental variables (IVs)-based causal inference methods, such as the conventional MTE approach and the two-stage least squares approach (2SLS) for estimating the LATE.
Specifically, we demonstrate that the presence of unobserved heterogeneous groups may invalidate the \textit{monotonicity} condition (e.g., \citealp{imbens1994identification,heckman2018unordered}).\footnote{
	Recently, there has been an increasing number of studies that deal with situations where the conventional monotonicity condition would not be satisfied (e.g., \citealp{lee2018identifying, mogstad2020causal, mogstad2020policy, mountjoy2019community, hoshino2021treatment}).
	Among them, our model setup is closely related to the one in \cite{mogstad2020causal, mogstad2020policy} in that these papers explicitly allow for heterogeneity in the treatment choice action associated with multiple IVs.
}
This result implies that the conventional MTE parameter may not even be well-defined under unobserved group-wise heterogeneity.
%, as the monotonicity condition is essentially identical to the latent index model underlying the conventional MTE framework (\citealp{vytlacil2002independence}).

Our identification strategy builds on the method of local IV by \cite{heckman1999local}.
Our main identification result requires three key conditions.
The first condition is that there exists a valid group-specific continuous IV.
As in the standard IV estimation, each group-specific IV must satisfy that the IV is independent of unobserved variables and that it is a determinant of the treatment.
The second condition is the exogeneity of group membership; that is, group membership is conditionally independent of the unobserved variables affecting treatment choices.
Since the second condition may be demanding in some applications, we also provide supplementary identification results when group membership is endogenous.
The third condition is the identifiability of the ``first-stage'' treatment choice equation as in the conventional MTE analysis.
Since nonparametrically identifying finite-mixture binary response models is extremely challenging, restricting our attention to finite-mixture Probit models, we provide sufficient conditions for the identification of the treatment choice model.\footnote{
	Several recent studies consider nonparametric identification of finite mixture models.
	For example, \citet{bonhomme2016non} present an identification result based on repeated measurements data satisfying some independence property, and \citet{kitamura2018nonparametric} develop nonparametric identification for a regression model with additive error.
	However, to the best of our knowledge, no existing results are applicable to our case.
	Nonparametric identification analysis on finite-mixture binary response models would be of great interest, but it still remains an open research question.
}

Based on our constructive identification results, we propose a two-step semiparametric estimator for the group-wise MTEs.
In the first step, we estimate a finite-mixture treatment choice model using a parametric maximum likelihood (ML) method.
In the second step, the MTE parameters are estimated using a series approximation method.
Under certain regularity conditions, we show that the proposed MTE estimator is consistent and asymptotically normally distributed.

As an empirical illustration, we investigate the effects of college education on annual income using the data for labor in Japan.
We focus on heterogeneity caused by the following two latent groups: the first comprises individuals whose college enrollment decisions are mainly affected by regional educational characteristics (group 1), while the second is a group of individuals who are mainly affected by at-home study environment (group 2).
More specifically, we employ regional characteristics, such as the local college enrollment rate and the rate of workforce participation for high school graduates, as the primary IVs for group 1.
For group 2, we create variables that measure the quality of study environment at home, and use these as the IVs.
Our empirical results indicate that for group 1, the treatment effect of college education is significantly positive if the (unobserved) cost of going on to a college is small.
In contrast, we cannot find such heterogeneity for group 2.

\paragraph{Organization of the paper.}

Section \ref{sec:identMTE} introduces our model and presents our main identification result for the group-wise MTE.
In Section \ref{sec:estimation}, we discuss the estimation procedure for the MTE parameters and prove its asymptotic properties.
Section \ref{sec:extension} provides two additional discussions: first, on the identification of MTE when group membership is endogenous and, second, on the identification of LATE when only binary IVs are available.
Section \ref{sec:simulation} gives the results of Monte Carlo experiments.
Section \ref{sec:empirical} presents the empirical illustration, and Section \ref{sec:conclusion} concludes the paper.
%In Appendices \ref{sec:proofs} and \ref{sec:lemma}, we provide technical proofs for our results.
%Appendix \ref{sec:suppident} presents supplementary identification results including the identification of the finite-mixture Probit models.

%%%%%%%%%%%%%%%%%%%%%%%%%%%%%%%%%%%%%%%%%%%%%%%%%%%%%%%%%%%%%%%%%%%%%%%%

\section{Identification of MTE under Unobserved Group Heterogeneity}\label{sec:identMTE}
\subsection{The model}\label{subsec:model}

In this section, we introduce our treatment effect model that allows for the presence of an unknown mixture of multiple subpopulations.
Throughout the paper, we assume that the number of groups is finite and known, which is denoted as $S \in \mathbb{N}$.
Each individual belongs to only one of the $S$ groups, and the group the individual belongs to, which we denote by $s \in \{1, \dots, S\}$, is a latent variable unknown to us.
Our goal is to measure the causal effect of a treatment variable $D \in \{0, 1\}$ on an outcome variable $Y \in \mathbb{R}$ for each group separately.
Let $Y^{(d)}$ be the potential outcome when $D = d$.
Then, the observed outcome can be written as $Y = D Y^{(1)} + (1 - D) Y^{(0)}$.
Suppose that the potential outcome equation is given by
\begin{align}\label{eq:model-Y}
	Y^{(d)} = \mu^{(d)} \left( X, s, \epsilon^{(d)} \right),
\end{align}
where $X \in \mathbb{R}^{\mathrm{dim}(X)}$ is a vector of observed covariates, $\epsilon^{(d)} \in \mathbb{R}$ is an unobserved error term, and $\mu^{(d)}$ is an unknown structural function.
This model specification is fairly general in that the functional form of $\mu^{(d)}$ is fully unrestricted and the distribution of the treatment effect $Y^{(1)} - Y^{(0)}$ can be heterogeneous across different groups.

Based on the latent index framework by \cite{heckman1999local,heckman2005structural}, we characterize our treatment choice model as follows: 
\begin{align}\label{eq:treatment}
	D = 
	\begin{cases}
	\mathbf{1}\left\{\mu^D_1(Z_1) \ge \epsilon_1^D \right\} & \text{if} \;\; s = 1\\ 
	\multicolumn{2}{c}{\vdots} \\
	\mathbf{1}\left\{\mu^D_S(Z_S) \ge \epsilon_S^D \right\} & \text{if} \;\; s = S
	\end{cases}
\end{align}
where $\mathbf{1}\{\cdot\}$ is an indicator function that takes one if the argument inside is true and zero otherwise, for $j \in \{ 1, \ldots , S \}$, $Z_j \in \mathbb{R}^{\mathrm{dim}(Z_j)}$ is a vector of IVs that may contain elements of $X$, $\epsilon_j^D \in \mathbb{R}$ is an unobserved continuous random variable, and $\mu_j^D$ is an unknown function. 
We allow for arbitrary dependence between $\epsilon_j^D$'s.
Assume that for all $j$, the error $\epsilon_j^D$ is independent of $Z_j$'s conditional on $X$. 
Moreover, we require that each $Z_j$ includes at least one group-specific continuous variable to ensure that the function $\mu_j^D(Z_j)$ does not degenerate to a constant after conditioning the values of $(X, Z_1, \ldots, Z_{j-1}, Z_{j+1}, \ldots, Z_S)$.

To proceed, let $F_j(\cdot | X)$ be the conditional cumulative distribution function (CDF) of $\epsilon_j^D$ given $X$.
Further, let $P_j \coloneqq F_j( \mu^D_j(Z_j) | X )$ and $V_j \coloneqq F_j ( \epsilon_j^D | X )$. 
By construction, each $V_j$ is distributed as $\text{Uniform}[0,1]$ conditional on $X$.
Using these definitions, we can rewrite \eqref{eq:treatment} as follows: $D = \mathbf{1}\left\{P_j \ge V_j \right\}$ if $s = j$.

\begin{remark}[Monotonicity]\label{rem:monotonicity}
    The presence of group heterogeneity may lead to the failure of the monotonicity condition in \cite{imbens1994identification}, which requires that shifts in the IVs determine the direction of change in the treatment choices uniformly in all individuals.
    To see this, for simplicity, consider a case with $S = 2$ and suppose that $\mu^D_j(Z_j) = Z_1 \gamma_{1j} + Z_{2,j}\gamma_{2 j}$ for $j \in \{ 1, 2 \}$, where $Z_1$ is a common IV, $Z_{2,j}$ is an IV specific to group $j$, and $Z_j = (Z_1, Z_{2,j})$.
    Suppose that $\gamma_{11} < 0$ and $\gamma_{12} > 0$.
    Then, an increase in $Z_1$ makes the individuals in group 1 (group 2) less (more) likely to take the treatment, implying that the monotonicity condition does not hold.
    As a result, the conventional IV-based methods relying on the monotonicity condition do not exhibit desirable causal interpretations.
    Note, however, that the monotonicity condition is still satisfied in terms of the group-specific IVs $(Z_{2,1}, Z_{2,2})$.
    Thus, if we run a 2SLS method using $(Z_{2,1}, Z_{2,2})$ only, we would obtain some causal effects averaged over the groups, as will be demonstrated in Subsection \ref{subsec:late}.
    However, this approach may overlook the possibility of heterogeneous treatment effects across groups.
\end{remark}

\begin{remark}[Partial monotonicity]\label{rem:pmonotonicity}
    \cite{mogstad2020causal} have proposed a weaker version of \citeauthor{imbens1994identification}'s (\citeyear{imbens1994identification}) monotonicity with multiple IVs, the so-called \textit{partial monotonicity}.
	The partial monotonicity condition only requires component-wise monotonicity with all other IVs being fixed.
	Here, consider the same treatment choice model as given in Remark \ref{rem:monotonicity}, and write the potential treatment as $D(z_1, z_{2,1}, z_{2,2})$. % = \sum_{j \in\{1,2\}} \mathbf{1}\{s = j\}\mathbf{1}\{\mu^D_j(z_1, z_{2,j}) \ge \epsilon_j^D \}$.
	Then, the partial monotonicity requires, for instance, that $D(z_1, z_{2,1}, z_{2,2}) \le D(z_1', z_{2,1}, z_{2,2})$ for $z_1 \neq z_1'$ for all individuals.
	However, this does not hold clearly when the sign of $\gamma_{11}$ is opposite to that of $\gamma_{12}$.
	Thus, although the partial monotonicity permits preference heterogeneity in the relative importance of different IVs (see Section 3 in \cite{mogstad2020causal}), it cannot account for unobserved group-wise heterogeneity as in this case, while ours can (at the cost of assuming the existence of group-specific IVs).
\end{remark}

For the treatment choice model in \eqref{eq:treatment}, we can interpret its meaning in several ways.
The first interpretation is that there are actually multiple different treatment eligibility rules prescribed by policy makers.
In the example of college enrollment, there are typically several different types of admissions processes for each school, for example, paper-based entrance exams, sports referrals, and so on.
Such a situation would correspond to this first type of interpretation.
Another interpretation is that there are several types of treatment preference patterns.
For example, consider again $D = 1$ if an individual goes to college and $D = 0$ otherwise.
Suppose that a common instrumental variable $Z$ is the introduction of a physical education requirement in colleges along with mandatory augmented athletics facilities.\footnote{
		This example scenario is borrowed from \cite{heckman2005structural}, Subsection 6.3.
		}
When we specify the functional form of $\mu^D_j$ as in Remark \ref{rem:monotonicity}, we imagine that some people dislike physical education ($\gamma_{z1} < 0$) while others like it ($\gamma_{z2} > 0$).
In this situation, we can view the treatment choice model \eqref{eq:treatment} as a binary response model with a discrete random coefficient.

\begin{remark}[Statistically choosing an optimal $S$]\label{rem:choiceS}
		It is well known that identifying the true number of latent groups is a non-standard issue.
		Most existing studies addressing this issue rely on some specific distributional assumptions or the availability of multiple outcomes (e.g., \citealp{chen2001modified}; \citealp{zhu2004hypothesis}; \citealp{woo2006robust}; \citealp{kasahara2014non}; \citealp{kitamura2018nonparametric}), and they are not directly applicable to our study.
		Our treatment choice model is essentially different from these studies in that a model with larger than $S$ groups is ``not'' a generalization of an $S$-group model due to the presence of group-specific variables.
		However, the former model always encompasses the latter in the usual mixture framework.
		Although developing a more general testing procedure for $S$ which can be used in our framework is an important open question, since our primary interest is in the identification and estimation of treatment effects, answering it is beyond the scope of this paper.
\end{remark}

\subsection{Identification of MTE}\label{subsec:identification}

Our main identification results are based on the following assumptions.

\begin{assumption}\label{as:IV}\hfil
\begin{enumerate}[(i)]
    \item The IVs $\mathbf Z = (Z_1, \ldots, Z_S)$ are independent of $(\epsilon^{(d)}, \epsilon^D_j, s)$ given $X$ for all $d$ and $j$.
    \item For each $j$, $Z_j$ has at least one group-specific continuous variable that is not included in $X$ and the IVs for the other groups.
\end{enumerate}
\end{assumption}

\begin{assumption}\label{as:membership}\hfil
\begin{enumerate}[(i)]
    \item The membership variable $s$ is conditionally independent of $\epsilon^D_j$ given $X$ for all $j$.
    \item For each $j$, there exists a constant $\pi_j \in (0,1)$ such that $\Pr(s = j | X) = \pi_j$.% and $\sum_{j = 1}^S \pi_j = 1$.
\end{enumerate}
\end{assumption}

Assumption \ref{as:IV}(i) is an exclusion restriction requiring that the IVs are conditionally independent of all unobserved random variables including the latent group membership.
Assumption \ref{as:IV}(ii) is somewhat demanding in that we require prior knowledge as to which variables may be relevant/irrelevant to the membership of each group.
A similar assumption can be found in the econometrics literature on finite mixture models (e.g., \citealp{compiani2016using}).
In practice, which IVs belong to which group needs to be determined on a case-by-case basis, using background knowledge and/or some economic theoretical framework underlying the presumed group structure in the data.\footnote{
	Note that the groups are distinguished only by their susceptibility to different IVs.
	We are generally unable to know if these groups correctly represent the groups originally presumed by researchers.
	Therefore, we need to make a comprehensive interpretation of the estimated groups by combining our prior knowledge of the data and the estimation results for the IVs.
	To obtain a sharper interpretation for the groups, some researchers consider a model in which the group membership probability is a function of individual characteristics, similar to our model in \eqref{eq:member}.
	Such a model is useful, but it does not allow us to know exactly who belongs to which group.
	This is a common limitation in the finite mixture framework.
	}
If researchers do not have such information, simply comparing the information criterion values of the models with different IVs might be practically useful (see also the results in Table \ref{table.AIC.BIC} in Subsection \ref{subsec:misspec}).
Note that in the absence of group-specific covariates, nonparametric identification of the treatment choice model would be in general infeasible, since when conditioned on the covariates, the observable choice probability is just a mixture of Bernoulli distributions, which is another Bernoulli distribution.
The continuity of group-specific IV is necessary to nonparametrically identify the MTE parameter.
As will be shown in Theorem \ref{thm:ident-MTR1}, the MTE parameter for group $j$ can be identified through a partial derivative with respect to $P_j$ conditioned on $(X, \{P_h\}_{h \neq j})$, which cannot be well-defined if all the IVs $Z_j$ are discrete.\footnote{
	Even when all group-specific IVs are discrete, if additional parametric functional form restrictions are imposed, it would be possible to identify MTE, in a similar manner to \cite{brinch2017beyond}.
	Investigation of such an approach is left for future research.
	}
The continuity assumption is also utilized to establish the identification of the finite-mixture Probit model, where the partial derivative of the choice probability with respect to $Z_j$ plays a key role in our identification strategy (see Appendix \ref{subsec:mixture}).

The conditional independence in Assumption \ref{as:membership}(i) excludes some types of endogenous group formation by imposing that group membership  $s$ does not affect the conditional distribution of $\epsilon^D_j$ given $X$.\footnote{
	Assumption \ref{as:membership}(i) differs from the \textit{no-additive-interaction} assumption of unobserved confounders for identifying the ATE in the IV estimation.
	The assumption sates that, in our treatment choice context,
	\begin{align} \label{eq:no-additive-interaction}
		\bE[ D | \bm{Z} = \bm{z}, X, s ] - \bE[ D | \bm{Z} = \bm{z}', X, s ] = \bE[ D | \bm{Z} = \bm{z}, X] - \bE[ D | \bm{Z} = \bm{z}', X]
	\end{align}
	for $\bm{z}, \bm{z}' \in \text{supp}[\bm{Z} | X]$, as in Assumption 5(a) of \citealp{wang2018bounded}.
	Under Assumptions \ref{as:IV}(i) and \ref{as:membership}, we observe that $\bE[ D | \bm{Z} = \bm{z}, X, s = j] = \Pr[ \mu_j^D(z_j) \ge \epsilon_j^D | X ]$ and $\bE[ D | \bm{Z} = \bm{z}, X] = \sum_{h=1}^S \Pr[ \mu_h^D(z_h) \ge \epsilon_h^D | X ] \cdot \pi_h$, which implies that \eqref{eq:no-additive-interaction} does not hold unless $S = 1$.
	As a result, their identification results do not recover the ATE in our situation.
	This is because $s$ interacts nonlinearly with $\bm{Z}$ in the finite mixture treatment choice equation \eqref{eq:treatment}.
	In addition, note that the no-additive-interaction assumption for the outcome equation (Assumption 5(b) of \citealp{wang2018bounded}) is not satisfied in our situation, as $s$ can be arbitrarily correlated with $Y^{(1)} - Y^{(0)}$.
	}
In other words, the expected ``potential'' treatment status when in group $j$ does not vary with the actual group membership $s$.
More specifically, in the above-mentioned example of college enrollment with multiple eligibility rules, this assumption can be interpreted as that, conditional on the individual characteristics $X$, the probability of enrolling in a college for each group is independent of the admission process the individuals actually take.
Thus, if unobserved confounders affect both group membership and treatment choice probability (e.g., unobservable individual talent), Assumption \ref{as:membership}(i) may not hold.
As will be shown in Subsection \ref{subsec:hetero} (see also Appendix \ref{subsubsec:endmember}), the assumption can be dropped at the cost of additional analytical complexity, and we can still establish some identification results of MTE parameters.
In Assumption \ref{as:membership}(ii), we assume homogeneous membership probability for each group, which is commonly used in the literature on finite mixture models.
This assumption is made only for simplicity, and the theorem shown below still holds without modifications even when the membership probability is a function of $X$.
Moreover, we will show  in Subsection \ref{subsec:hetero} that the group-wise MTE can be identified even when the membership probability depends on other covariates besides $X$.

To interpret these assumptions in a practical scenario, consider our empirical application in Section \ref{sec:empirical}, in which $Y$ is annual income, $D$ indicates college enrollment, and $X$ includes demographic variables such as age and sex.
Suppose that we have two latent groups with different treatment preference patterns: the individuals in group 1 determine whether to go to college or not mainly by their neighborhood's educational characteristics and those in group 2 decide it by the quality of own study environment.
Specifically, $Z_1$ includes variables such as local college enrollment rate, and $Z_2$ includes, for example, the number of books in their home when they were 15.
Assumption \ref{as:IV} is satisfied if these IVs are not direct determinants of annual income and if they are also irrelevant to the other unobserved determinants and the preference type conditional on the observed demographic variables.
Assumption \ref{as:membership} requires that the preference type is unrelated to both the demographic variables and the unobserved factors of college enrollment.

A directed acyclic graph (DAG) in Figure \ref{fig:dag1} is helpful for further understanding Assumptions \ref{as:IV} and \ref{as:membership}.
The solid and dashed arrows indicate the effects of the observed and unobserved variables, respectively.
Assumption \ref{as:IV}(i) rules out the direct pathways between $Z_j$ and $(\epsilon^{(d)}, \epsilon_j^D, s)$.
The arrow from $Z_j$ to $D$ is necessary for Assumption \ref{as:IV}(ii).
Under Assumption \ref{as:membership}(i), there is no direct pathway between $s$ and $\epsilon_j^D$.
In this figure, for simplicity, we suppress $X$, which can be associated with all variables except for $s$ (Assumption \ref{as:membership}(ii)).

\begin{figure}[!h]
	\begin{center}
		\includegraphics[width = 6cm, bb = 0 0 115 88]{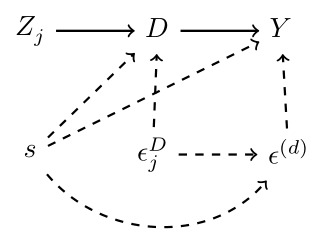}
		\caption{DAG under Assumptions \ref{as:IV} and \ref{as:membership}}
		\label{fig:dag1}
	\end{center}
\end{figure}

To identify the treatment effects of interest, we first need to identify the treatment choice model \eqref{eq:treatment}.
Although there has been a long history of research on the identification of finite mixture models, only few studies address binary outcome regression models (e.g., \citealp{follmann1991identifiability}; \citealp{butler1997consistency}).
Moreover, these previous results are typically based on the availability of repeated measurements data, which cannot be directly applied to our situation.
Therefore, in Appendix \ref{subsubsec:exomember}, we present a new identification result for a finite-mixture treatment choice Probit model.
Since our major focus is on the identification of treatment effect parameters, the details are omitted here, and we hereinafter treat $P_j$'s and $\pi_j$'s as known objects.

The MTE parameter specific to group $j$ is defined as
\begin{align} \label{eq:groupMTE}
    \text{MTE}_j(x,p) \coloneqq m_j^{(1)}(x,p) - m_j^{(0)}(x,p),
\end{align}
where $m_j^{(d)}(x,p) \coloneqq \bE [ Y^{(d)} | X = x, s = j, V_j =p]$ is the marginal treatment response (MTR) function specific to group $j$ and $d \in \{0, 1\}$.\footnote{
	Note that we can consider other types of MTE parameters than \eqref{eq:groupMTE}.
	For example, an MTE that uses $P_j$ instead of conditioning on $X$ has recently gained attention  because of its  good computational properties (cf. \citealp{zhou2019marginal}), but it is not pursued in this study considering the page limitation.
}
This group-wise MTE parameter captures the unobserved heterogeneity in the treatment effects with respect to both $s$ and $V_j$.
Similar to the conventional MTE framework, we can infer whether there is individual unobserved treatment selectivity by examining how $\text{MTE}_j(x,p)$ varies with $p$.
Specifically, if $\text{MTE}_j(x,p)$ is declining in $p$, it indicates that individuals who belong to group $j$ and are more likely to choose $D = 1$ tend to obtain larger gains from the treatment.
The case of no unobserved heterogeneity corresponds to a constant MTE function (cf. \citealp{carneiro2011estimating,mogstad2018identification}).
Importantly, our MTE parameter allows us to examine this group-wisely.

Once the group-wise MTEs are identified for all $x$ and $p$, we can recover many other treatment parameters.
For example, $\text{CATE}_j(x) = \int_0^1 \text{MTE}_j(x,v)\mathrm{d}v$, where $\text{CATE}_j(x) = \bE [ Y^{(1)} - Y^{(0)} | X = x, s = j]$ is the group-wise conditional average treatment effect.
Furthermore, the group-wise ATE: $\text{ATE}_j = \bE [ Y^{(1)} - Y^{(0)} | s = j]$ can be obtained by $\text{ATE}_j = \int \text{CATE}_j(x)f_X(x)\mathrm{d}x$, where $f_X$ is the marginal density function of $X$ (note that Assumption \ref{as:membership}(ii) implies $f_X(\cdot | s = j) = f_X(\cdot)$).
Then, the ATE for the whole population is simply given by $\text{ATE} = \sum_{j=1}^S \pi_j \text{ATE}_j$.
We can also identify the so-called policy relevant treatment effect (PRTE); see Remark \ref{rem:policy} below and Appendix \ref{subsec:prte}.

\begin{remark}[Policy evaluation]\label{rem:policy}
	When policymakers seek a more effective treatment assignment policy, estimating $\text{MTE}_j(x, v)$ should help them understand what type of individuals would benefit from the treatment.
	For instance, if the MTE increases in $v$ for a particular group, the treatment brings a significant benefit especially to individuals in this group who are less likely to be treated.
	In this case, an assignment policy that enlarges the size of the treatment group would be desirable for this group.
	
	When the policymakers cannot directly interfere in the treatment assignment process but can introduce a	policy that affects the probability of treatment, we may consider the PRTE, as in \cite{heckman2005structural}; e.g., in our empirical context, a policy that builds a new university.
	The PRTE parameter is defined as follows:
	\begin{align*}
		\text{PRTE}_j \coloneqq \bE [Y^\star | X = x, s = j] - \bE [Y | X = x, s = j].
	\end{align*}
	Here, $Y^\star$ denotes the counterfactual outcome after the new policy.
	For identification of $\text{PRTE}_j$, see Appendix \ref{subsec:prte}.
	For example, suppose that the PRTEs of two different policies are evaluated and it is revealed that one is actually harmful to most groups and significantly beneficial for only one particular group, while the other is weakly beneficial for all groups.
	Then, we may conclude that the latter policy is ``safer'' than the former in a minimax sense.  
\end{remark}

Below, we show that the MTR functions can be identified through the partial derivatives of the following functions:
\begin{align*}
    \psi_1(x, \mathbf{p})
    \coloneqq \bE [ DY | X = x, \mathbf P = \mathbf p],
    \qquad 
    \psi_0(x, \mathbf{p})
    \coloneqq \bE [ (1- D)Y | X = x, \mathbf P = \mathbf p],
\end{align*}
where $\mathbf P = (P_1, \ldots, P_S)$ and $\mathbf p = (p_1, \ldots, p_S)$.
Note that $\psi_d(x, \mathbf{p})$ can be directly identified from the data on $\text{supp}[X, \mathbf P | D = d]$, where $\text{supp}[X, \mathbf P | D = d]$ denotes the joint support of $(X, \mathbf P)$ given $D = d$.

\begin{theorem}\label{thm:ident-MTR1}
	Suppose that Assumptions \ref{as:IV} and \ref{as:membership} hold.
	If $m_j^{(1)}(x, \cdot)$ and $m_j^{(0)}(x, \cdot)$ are continuous, we have
	\begin{align*}
	    m_j^{(1)}(x, p_j) = \frac{1}{\pi_j}\frac{\partial \psi_1(x, \mathbf p) }{\partial p_j}, 
	    \qquad 
	    m_j^{(0)}(x, p_j) = - \frac{1}{\pi_j}\frac{\partial \psi_0(x, \mathbf p) }{\partial p_j}.
	\end{align*}
\end{theorem}

As shown in the proof of Theorem \ref{thm:ident-MTR1}, we can find the equality $\psi_1(x, \mathbf{p}) = \sum_{j = 1}^S \pi_j \int_0^{p_j}m_j^{(1)}(x,v) \mathrm{d} v$.
Then, the first result is straightforward by partially differentiating both sides of this with respect to $p_j$.
The second result can be obtained similarly.
Here, Assumption \ref{as:IV}(ii) plays an essential role in enabling us to shift $P_j$ only with the other components of $\mathbf{P}$ being fixed.
Notice that the results in Theorem \ref{thm:ident-MTR1} give us a testable implication for the validity of Assumption \ref{as:IV}(ii), as kindly pointed out by a referee.
That is, for $\mathbf{p}$ and $\mathbf{p}'$, where $\mathbf{p}'$ differs from $\mathbf{p}$ only in terms of group $k$'s specific IV, the equality $\partial \psi_d(x, \mathbf{p})/\partial p_j = \partial \psi_d(x, \mathbf{p}')/ \partial p_j$ may not hold for $j \neq k$ without this assumption.
However, to verify this equality from data, we need to estimate the partial derivatives of $\psi_d$ nonparametrically, which should be difficult in practice because of the curse of dimensionality (and even if the equality is confirmed, it only gives a necessary condition for Assumption \ref{as:IV}(ii)).
Note also that since our proposed MTE estimator is defined directly based on Assumption \ref{as:IV}(ii), it does not contain any information to verify $\partial \psi_d(x, \mathbf{p})/\partial p_j = \partial \psi_d(x, \mathbf{p}')/ \partial p_j$. 
To assess the impact of the violation of Assumption \ref{as:IV}(ii), we conducted a small numerical analysis.
This demonstrates that our MTE estimator exhibits poor performance in the absence of group-specific IVs, which is consistent with our theory.
For more details, see Appendix \ref{sec:appendix:simulation}.

\begin{remark}[Necessity of identifying the treatment choice model]
	Similar to the original LIV method in \citet{heckman1999local,heckman2005structural}, identifying the treatment choice probability is essential to our case as well.
	To see this, consider a simple scenario in which $Z_j$ is a scalar continuous IV.
	Further, assume that $p_j(x, z_j) = F_j(\mu_j^D(z_j) | X = x)$ is strictly monotonic in $z_j$.
	It then holds that $\bE [DY|X = x, \mathbf{Z} = \mathbf{z}] = \psi_1(x, \mathbf{p}(x, \mathbf{z}))$, where $\mathbf{p}(x, \mathbf{z}) = (p_1(x, z_1), \dots, p_S(x, z_S))$.
	The equality $\psi_1(x, \mathbf{p}) = \sum_{j = 1}^S \pi_j \int_0^{p_j}m_j^{(1)}(x,v) \mathrm{d} v$ leads to
	\begin{align*}
		\frac{\partial}{\partial z_j} \bE [DY|X = x, \mathbf{Z} = \mathbf{z}] = \pi_j m_j^{(1)}(x, p_j) \left[ \frac{\partial}{\partial z_j} p_j(x, z_j) \right]. 
	\end{align*}
	Hence, as long as $\partial p_j(x, z_j) / \partial z_j$ is unknown, the presence of group-specific IVs alone is not sufficient to infer the MTR function.
	Note that when one knows the sign of $\partial p_j(x, z_j) / \partial z_j$ a priori, since $\pi_j$ is positive, it is possible to know the sign of the MTR without estimating the treatment choice model. 
\end{remark}

%%%%%%%%%%%%%%%%%%%%%%%%%%%%%%%%%%

\section{Estimation and Asymptotics} \label{sec:estimation}

This section considers the estimation of the group-wise MTE parameters using an independent and identically distributed (IID) sample $\{(Y_i, D_i, X_i, \mathbf Z_i): 1 \le i \le n\}$.
Throughout this section, Assumptions \ref{as:IV} and \ref{as:membership} are assumed to hold.

\subsection{Two-step series estimation}\label{subsec:estimation}

\paragraph{First-stage: ML estimation}
We consider the following parametric model specification:
\begin{align}\label{eq:paramodel-D}
	D = \mathbf{1}\left\{Z_j^\top \gamma_j \ge \epsilon_j^D \right\} \; \text{with probability $\pi_j > 0$, for $j \in \{ 1,\dots,S \}$.}
\end{align}
\begin{comment}
\begin{align*}
	D = 
	\begin{cases}
	\mathbf{1}\left\{Z_1^\top \gamma_1 \ge \epsilon_1^D \right\} & \text{with probability} \;\; \pi_1\\ 
	\multicolumn{2}{c}{\vdots} \\
	\mathbf{1}\left\{Z_S^\top \gamma_S \ge \epsilon_S^D \right\} & \text{with probability} \;\; \pi_S
	\end{cases}
\end{align*} 
\end{comment}
We assume that $\epsilon_j^D$ is independent of $(X, \mathbf Z)$, and the CDF $F_j$ of $\epsilon_j^D$ is a known function such as the standard normal CDF.
Define $\gamma = \left(\gamma_1^\top, \ldots, \gamma_S^\top \right)^\top$ and $\pi = (\pi_1, \ldots, \pi_S)^\top$.
Then, the conditional likelihood function for an observation $i$ when $s_i = j$ is given by $L_i(\gamma | s_i = j) \coloneqq F_j( Z_{ji}^\top \gamma_j )^{D_i} [ 1 - F_j( Z_{ji}^\top \gamma_j )]^{1 - D_i}$.
Thus, the ML estimator for $\gamma$ and $\pi$ can be obtained by
\begin{align*}
    (\hat \gamma_n, \hat \pi_n) \coloneqq \argmax_{(\tilde \gamma, \tilde \pi) \in \Gamma \times \mathcal{C}_S} \sum_{i = 1}^n \log\left( \sum_{j = 1}^S \tilde \pi_j L_i(\tilde \gamma | s_i = j)\right),
\end{align*}
where $\Gamma \subset \mathbb{R}^{\sum_{j = 1}^S \dim(Z_j)}$ and $\mathcal{C}_S \coloneqq \{\tilde \pi \in (0, 1)^S: \sum_{j = 1}^S \tilde \pi_j = 1\}$ are the parameter spaces.
We then obtain the estimator of $P_j = F_j( Z_j^\top \gamma_{j} )$ as $\hat P_j = F_j ( Z_j^\top \hat \gamma_{n, j} )$.
In the numerical studies below, we use the expectation-maximization (EM) algorithm to solve this ML problem following the literature on finite mixture models (e.g., \citealp{dempster1977maximum}; \citealp{mclachlan2004finite}; \citealp{train2008algorithms}).

\paragraph{Second-stage: series estimation}

For the potential outcome equation, we assume the following linear model for convenience, which is a popular setup in the literature (see, e.g., \citealp{carneiro2009estimating}):
\begin{align}\label{eq:paramodel-Y}
    Y^{(d)} 
    = \sum_{j = 1}^S \mathbf{1}\{ s = j\} X^\top \beta_j^{(d)} + \epsilon^{(d)},
    \quad
    \text{for $d \in \{0,1\}$.}
\end{align}
Here, the coefficient $\beta_j^{(d)}$ may depend on the group membership $j$ to account for potentially heterogeneous effects of $X$.
%where $\epsilon^{(d)} = \sum_{j=1}^S \mathbf{1}\{s = j\} \epsilon_j^{(d)}$ is an error term.
%Here, the error term $\epsilon^{(d)}$ can generally depend on the group membership $j$, but we suppress the dependence for expositional simplicity.
%\footnote{
%    Or, the error term can be viewed as a composite error $\epsilon^{(d)} = \sum_{j=1}^S \mathbf{1}\{s = j\} \epsilon_j^{(d)}$.
%}
Assume that $X$ is conditional-mean independent of $(\epsilon^{(d)}, \epsilon_j^D, s)$ for all $d$ and $j$.
%Although the independence is stronger than necessary for the identification result above, this type of independence is generally introduced for estimating MTE parameters.
Then, we have
\begin{align*}
    \text{MTE}_j(x,p)
    = m_j^{(1)}(x,p) - m_j^{(0)}(x,p)
    = x^\top (\beta_j^{(1)} - \beta_j^{(0)}) + \bE [\epsilon^{(1)} - \epsilon^{(0)} | s = j, V_j = p],
\end{align*}
with $m_j^{(d)}(x,p) = x^\top \beta_j^{(d)} + \bE [\epsilon^{(d)} | s = j, V_j = p]$.
By the same argument as in the proof of Theorem \ref{thm:ident-MTR1}, we can show that there exist univariate functions $g_j^{(1)}$ and $g_j^{(0)}$ for each $j$ satisfying
\begin{align*}
    \bE [\epsilon^{(1)} | s = j, V_j \le p_j] = \frac{1}{p_j} \int_0^{p_j} \bE [\epsilon^{(1)} | s = j, V_j = v]\mathrm{d}v
    & \eqqcolon \frac{g_j^{(1)}(p_j)}{p_j}, \\
    \bE [\epsilon^{(0)} | s = j, V_j > p_j] = \frac{1}{1 - p_j} \int_{p_j}^1 \bE [\epsilon^{(0)} | s = j, V_j = v]\mathrm{d}v
    & \eqqcolon \frac{g_j^{(0)}(p_j)}{1 - p_j}.
\end{align*}
Then, letting $\nabla g_j^{(d)}(p) \coloneqq \partial g_j^{(d)}(p) / \partial p$, we observe that $\nabla g_j^{(1)}(p) = \bE [\epsilon^{(1)} | s = j, V_j = p]$ and $\nabla g_j^{(0)}(p) = - \bE [\epsilon^{(0)} | s = j, V_j = p]$.
Hence, it follows that
\begin{align*}
    m_j^{(1)}(x, p) = x^{\top} \beta_j^{(1)} + \nabla g_j^{(1)}(p),
    \qquad
    m_j^{(0)}(x, p) = x^{\top} \beta_j^{(0)} - \nabla g_j^{(0)}(p).
\end{align*}
Further, by the law of iterated expectations, 
\begin{align*}
    \psi_1(x, \mathbf p)
    & = \sum_{j = 1}^S \bE [Y^{(1)}|X = x, s = j, V_j \le p_j] \cdot \pi_j p_j
    %& = \sum_{j = 1}^S \left(x^\top \beta_j^{(1)} + \bE [\epsilon^{(1)} | s = j, V_j \le p_j]\right) \cdot \pi_j p_j \\
    = \sum_{j = 1}^S \left( (x\cdot \pi_j p_j )^\top \beta_j^{(1)} + \pi_j g_j^{(1)}(p_j) \right),\\
    \psi_0(x, \mathbf p)
    & = \sum_{j = 1}^S \bE [Y^{(0)}|X = x, s = j, V_j > p_j] \cdot \pi_j (1 - p_j)
    %& = \sum_{j = 1}^S \left(x^\top \beta_j^{(0)} + \bE [\epsilon^{(0)} | s = j, V_j > p_j]\right) \cdot \pi_j (1 - p_j) \\
    = \sum_{j = 1}^S \left( (x\cdot \pi_j (1 - p_j) )^\top \beta_j^{(0)} + \pi_j g_j^{(0)}(p_j) \right).
\end{align*}
Hence, we obtain the following partially linear additive regression models:
\begin{align}\label{eq:additive1}
    DY 
    & = \sum_{j = 1}^S (X\cdot \pi_j P_j)^\top \beta_j^{(1)} + \sum_{j = 1}^S \pi_j g_j^{(1)}(P_j) + \varepsilon^{(1)},\\
    \label{eq:additive0}
    (1 - D)Y 
    & = \sum_{j = 1}^S (X\cdot \pi_j (1 - P_j))^\top \beta_j^{(0)} + \sum_{j = 1}^S \pi_j g_j^{(0)}(P_j) + \varepsilon^{(0)},
\end{align}
where $\bE [\varepsilon^{(d)} | X, \mathbf P] = 0$ by the definition of $\psi_d$ for both $d \in \{0, 1\}$.
To estimate the coefficients $\beta_j^{(d)}$'s and the functions $g_j^{(d)}$'s, we use the series (sieve) approximation method such that $g_j^{(d)}(p) \approx b_K(p)^\top \alpha_j^{(d)}$ with a $K \times 1$ vector of basis functions $b_K(p) = (b_{1K}(p), \dots, b_{KK}(p))^\top$ and corresponding coefficients $\alpha_j^{(d)}$.

To proceed, letting $d_{SXK} \coloneqq S(\text{dim}(X) + K)$, define $\underset{d_{SXK}\times 1}{\theta^{(d)}} \coloneqq (\beta_1^{(d)\top}, \ldots , \beta_S^{(d)\top}, \alpha_1^{(d)\top}, \ldots , \alpha_S^{(d)\top} )^\top$ for $d \in \{0, 1\}$, and
\begin{alignat*}{2}
    \underset{d_{SXK} \times 1}{R_K^{(1)}} 
    & \coloneqq (\pi_1 P_1 X^\top, \ldots ,\pi_S P_S X^\top, 
    & \pi_1 b_K(P_1)^\top, \ldots , \pi_S b_K(P_S)^\top)^\top, \\
    \underset{d_{SXK} \times 1}{R_K^{(0)}}
    & \coloneqq (\pi_1 (1 - P_1) X^\top, \ldots , \pi_S (1 - P_S) X^\top, 
    & \pi_1 b_K(P_1)^\top, \ldots , \pi_S b_K(P_S)^\top)^\top.
\end{alignat*}
Then, we can approximate the regression models \eqref{eq:additive1} and \eqref{eq:additive0}, respectively, by
\begin{align}\label{eq:regression}
    DY \approx R_K^{(1)\top}\theta^{(1)} + \varepsilon^{(1)}, 
    \qquad 
    (1 - D) Y \approx R_K^{(0)\top}\theta^{(0)} + \varepsilon^{(0)},
\end{align}
which implies that we can estimate $\theta^{(d)}$ by 
\begin{align*}
    \tilde \theta_n^{(1)} \coloneqq \left(\sum_{i = 1}^n R_{i,K}^{(1)}R_{i,K}^{(1)\top} \right)^{-}\sum_{i = 1}^n R_{i,K}^{(1)}D_iY_i, \qquad 
    \tilde \theta_n^{(0)} \coloneqq \left(\sum_{i = 1}^n R_{i,K}^{(0)}R_{i,K}^{(0)\top} \right)^{-}\sum_{i = 1}^n R_{i,K}^{(0)} (1 - D_i) Y_i,
\end{align*}
where $A^{-}$ is a generalized inverse of a matrix $A$.
Note, however, that $\tilde \theta_n^{(d)}$'s are infeasible since $P_j$'s and $\pi_j$'s are unknown in practice.
Then,  define $\hat R_K^{(d)}$ analogously as above but replacing $P_j$'s and $\pi_j$'s with their estimators obtained in the first step.
The feasible estimators can be obtained by
\begin{align*}
    \hat \theta_n^{(1)} \coloneqq \left(\sum_{i = 1}^n \hat R_{i,K}^{(1)} \hat R_{i,K}^{(1)\top} \right)^{-}\sum_{i = 1}^n \hat R_{i,K}^{(1)}D_iY_i, \qquad 
    \hat \theta_n^{(0)} \coloneqq \left(\sum_{i = 1}^n \hat R_{i,K}^{(0)} \hat R_{i,K}^{(0)\top} \right)^{-} \sum_{i = 1}^n \hat R_{i,K}^{(0)} (1 - D_i) Y_i.
\end{align*}
The feasible estimator of $g_j^{(d)}(p)$ is given by $\hat g_j^{(d)}(p) = b_K(p)^\top \hat \alpha_{n,j}^{(d)}$, and the infeasible estimator is given by $\tilde g_j^{(d)}(p) = b_K(p)^\top \tilde \alpha_{n,j}^{(d)}$.
Letting $\nabla b_K(p) \coloneqq \partial b_K(p)/\partial p$, we can also estimate $\nabla g_j^{(d)}(p)$ by $\nabla \hat g_j^{(d)}(p) \coloneqq \nabla b_K(p)^\top \hat \alpha_{n,j}^{(d)}$ and $\nabla \tilde g_j^{(d)}(p) \coloneqq \nabla b_K(p)^\top \tilde \alpha_{n,j}^{(d)}$.
Thus, we obtain the estimators of the MTR functions as follows:
\begin{align*}
    &\tilde m_j^{(1)}(x, p) = x^\top \tilde \beta_{n,j}^{(1)} + \nabla \tilde g_j^{(1)}(p),
    \qquad \hat m_j^{(1)}(x, p) = x^\top \hat \beta_{n,j}^{(1)} + \nabla \hat g_j^{(1)}(p), \\
    & \tilde m_j^{(0)}(x, p) = x^\top \tilde \beta_{n,j}^{(0)} - \nabla \tilde g_j^{(0)}(p),
    \qquad \hat m_j^{(0)}(x, p) = x^\top \hat \beta_{n,j}^{(0)} - \nabla \hat g_j^{(0)}(p).
\end{align*}
Finally, the feasible estimator of the MTE is given by $\hat{\text{MTE}}_j(x,p) = \hat m_j^{(1)}(x, p) - \hat m_j^{(0)}(x, p)$, and the infeasible estimator is given by $\tilde{\text{MTE}}_j(x,p) = \tilde m_j^{(1)}(x, p) - \tilde m_j^{(0)}(x, p)$.

\subsection{Asymptotic properties}\label{subsec:asymptotics}

This section presents the asymptotic properties of the proposed MTE estimators for the model given by equations \eqref{eq:paramodel-D} and \eqref{eq:paramodel-Y}.
In the following, for a vector or matrix $A$, we denote its Frobenius norm as $\| A \| = \sqrt{\text{tr}\{ A^\top A \}}$ where $\text{tr}\{ \cdot \}$ is the trace.
For a square matrix $A$, $\lambda_{\text{min}}(A)$ and $\lambda_{\text{max}}(A)$ denote the smallest and the largest eigenvalues of $A$, respectively.

\begin{assumption}\label{as:iid}
	The data $\{ (Y_i, D_i, X_i, \mathbf{Z}_i) : 1 \le i \le n \}$ are IID.
\end{assumption}

\begin{assumption}\label{as:parametricML} \hfil
	\begin{enumerate}[(i)]
		%\item $D_i = \mathbf{1}\{Z_{ji}^\top \gamma_j \ge \epsilon_{ji}^D\}$ if $s_i = j$ for $j \in \{ 1, \ldots, S \}$.
		\item $\text{supp}[Z_j]$ is a compact subset of $\mathbb{R}^{\dim(Z_j)}$ for all $j$.
		\item The random variables $\epsilon^D = (\epsilon_1^D, \dots, \epsilon_S^D)$ are continuously distributed on the whole $\mathbb{R}^S$ and independent of $(X, \mathbf{Z})$.
		Each $\epsilon_j^D$ has a twice continuously differentiable known CDF $F_j$ with bounded derivatives.
		\item $\| \hat \gamma_n - \gamma \| = O_P(n^{-1/2})$ and $\| \hat \pi_n - \pi \| = O_P(n^{-1/2})$.
	\end{enumerate}
\end{assumption}

\begin{assumption}\label{as:outcome}
	\hfil
	\begin{enumerate}[(i)]
		%\item $Y_{ji}^{(d)} = X_i^\top \beta_j^{(d)} + \epsilon_i^{(d)}$ for each $d$ and $j$.
		\item $\text{supp}[X]$ is a compact subset of $\mathbb{R}^{\dim(X)}$. 
		$\bE[\epsilon^{(d)} | X, \epsilon_j^D, s = j] = \bE[\epsilon^{(d)} | \epsilon_j^D, s = j]$ for all $d$ and $j$.
		%$X$ is independent of $(\epsilon^{(d)}, \epsilon_j^D, s)$ for all $d$ and $j$.
		\item The random variable $\epsilon^{(d)}$ is independent of $(X, \mathbf{Z})$ for all $d$.
	\end{enumerate}
\end{assumption}

The IID sampling condition in Assumption \ref{as:iid} is standard in the literature.
For Assumption \ref{as:parametricML}(iii), if the parameters $\gamma$ and $\pi$ are identifiable, the $\sqrt{n}$-consistency of the ML estimators is straightforward; see Appendix \ref{subsec:mixture} for a special case where $\epsilon_j^D$'s are distributed as the standard normal.
Note that combining Assumptions \ref{as:parametricML}(i)-(iii) imply $\max_{1 \le i \le n} |\hat P_{ji} - P_{ji}| = O_P(n^{-1/2})$ for all $j$.
In Assumption \ref{as:outcome}, we assume the (conditional-mean) independence between the observables and the unobservables.
%the conditional-mean independence between $X$ and the unobservables in Assumption \ref{as:outcome}(i).
%, which is weaker than the commonly used full independence condition between them.
%We do not require the latter condition owing to the additive linear structure of the outcome equation.
%For related discussion, see \cite{brinch2017beyond} (Assumption 2 and the arguments therein).

Let $\nabla^a g_j(p) \coloneqq \partial^a g_j(p) / (\partial p)^a$ and $\nabla^a b_K(p) \coloneqq (\partial^a b_{1K}(p) / (\partial p)^a, \dots, \partial^a b_{KK}(p) / (\partial p)^a)^\top$ for a non-negative integer $a$.
Further, define $\Psi_K^{(d)} \coloneqq \bE \left[ R_K^{(d)}R_K^{(d)\top} \right]$ and $\Sigma_K^{(d)} \coloneqq \bE \left[ (\xi_K^{(d)})^2 R_K^{(d)}R_K^{(d)\top} \right]$, where $\xi_K^{(d)} \coloneqq e^{(d)} + B_K^{(d)}$, and $e^{(d)}$ and $B_K^{(d)}$ are unobserved error terms in the series regressions whose definitions are given in \eqref{eq:definitionBe} in Appendix \ref{sec:proofs}.

\begin{assumption}\label{as:series}\hfil
	\begin{enumerate}[(i)]
		\item For each $d$ and $j$, $g^{(d)}_j(p)$ is $r$-times continuously differentiable for some $r \ge 1$, and there exist positive constants $\mu_0$ and $\mu_1$ and some $\alpha_j^{(d)} \in \mathbb{R}^K$ satisfying $\sup_{p \in [0,1]} | b_K(p)^\top \alpha_j^{(d)} - g^{(d)}_j(p) | = O(K^{-\mu_0})$ and $\sup_{p \in [0,1]} | \nabla b_K(p)^\top \alpha_j^{(d)} - \nabla g^{(d)}_j(p) | = O(K^{-\mu_1})$.
		\item $b_K(p)$ is twice continuously differentiable and satisfies $\zeta_0(K) \sqrt{(K \log K)/ n} \to 0$, $\zeta_1(K) \sqrt{K/n} \to 0$ and $\zeta_2(K) / \sqrt{n} \to 0$, where $\zeta_{l}(K) \coloneqq \max_{0 \le a \le l} \sup_{p \in [0,1]} \| \nabla^a b_K(p) \|$.
	\end{enumerate}
\end{assumption}

\begin{assumption}\label{as:eigen}\hfil
	\begin{enumerate}[(i)]
		\item For $d \in \{ 0, 1 \}$, there exist constants $\underbar{c}_{\Psi}$ and $\bar{c}_{\Psi}$ such that $0 < \underbar{c}_{\Psi} \le \lambda_{\text{min}} \left( \Psi_K^{(d)} \right) \le \lambda_{\text{max}} \left( \Psi_K^{(d)} \right) \le \bar{c}_{\Psi} < \infty$, uniformly in $K$.
		\item For $d \in \{ 0, 1 \}$, there exist constants $\underbar{c}_{\Sigma}$ and $\bar{c}_{\Sigma}$ such that $0 < \underbar{c}_\Sigma \le \lambda_{\text{min}}\left( \Sigma_K^{(d)} \right) \le \lambda_{\text{max}}\left( \Sigma_K^{(d)} \right) \le \bar{c}_\Sigma < \infty$, uniformly in $K$.
	\end{enumerate}	
\end{assumption}

\begin{assumption}\label{as:error}
	$\bE [(e^{(d)})^4 | D, X, \mathbf Z, s = j] < \infty$ for all $d$ and $j$.
\end{assumption}

Assumption \ref{as:series}(i) imposes smoothness conditions on the function $g_j^{(d)}$.
These conditions are standard in the literature on series approximation methods.
For example, Lemma 2 in \cite{holland2017penalized} shows that Assumption \ref{as:series}(i) is satisfied by B-splines of order $k$ for $k - 2 \ge r$ such that $\mu_0 = r$ and $\mu_1 = r - 1$.
This assumption requires $K$ to increase to infinity for unbiased estimation while Assumption \ref{as:series}(ii) requires that $K$ should not diverge too quickly.
It is well known that $\zeta_l(K) = O(K^{(1/2) + l})$ for B-splines (e.g., \citealp{newey1997convergence}).
Thus, when B-spline basis functions are employed, Assumption \ref{as:series}(ii) is satisfied if $K^5 / n \to 0$.
Assumption \ref{as:eigen} ensures that the matrices $\Psi_K^{(d)}$ and $\Sigma_K^{(d)}$ are positive definite for all $K$ so that their inverses exist.
Assumption \ref{as:error} is used to derive the asymptotic normality of our estimator in a convenient way.

Finally, we introduce the selection matrices $\underset{\dim(X) \times d_{SXK}}{\mathbb{S}_{X, j}}$ and $\underset{K \times d_{SXK}}{\mathbb{S}_{K, j}}$ such that $\mathbb{S}_{X, j} \theta^{(d)} = \beta_j^{(d)}$ and $\mathbb{S}_{K, j} \theta^{(d)} = \alpha_j^{(d)}$ for each $d$ and $j$.
The next theorem gives the asymptotic normality for the infeasible estimator.

\begin{theorem}\label{thm:normality}
    Suppose that Assumptions \ref{as:iid} -- \ref{as:error} hold.
	For a given $p \in \text{supp}[P_j | D = 1] \cap \text{supp}[P_j | D = 0]$, if $\| \nabla b_K(p) \| \to \infty$, $\sqrt{n} K^{-\mu_0} \to 0$, and $\sqrt{n} K^{-\mu_1} / \| \nabla b_K(p) \| \to 0$ hold, then
	\begin{align*}
    	\frac{\sqrt{n}\left( \tilde{\text{MTE}}_j(x,p) - \text{MTE}_j(x,p) \right)}{\sqrt{\left[ \sigma_{K,j}^{(1)}(p) \right]^2 + 2cov_{K}(p) + \left[ \sigma_{K,j}^{(0)}(p) \right]^2}} \overset{d}{\to} N(0,1),
	\end{align*}
	where
	\begin{align*}
	\sigma_{K,j}^{(d)}(p) 
	& \coloneqq \sqrt{
		\nabla b_K(p)^\top  \mathbb{S}_{K,j} \left[ \Psi_K^{(d)} \right]^{-1} \Sigma_K^{(d)} \left[ \Psi_K^{(d)} \right]^{-1} \mathbb{S}_{K,j}^\top \nabla b_K(p)
	}, \;\; \text{ for $d \in \{0,1\}$},\\
	cov_{K}(p)
	& \coloneqq \nabla b_K(p)^\top  \mathbb{S}_{K,j} \left[ \Psi_K^{(0)} \right]^{-1} C_K \left[ \Psi_K^{(1)} \right]^{-1} \mathbb{S}_{K,j}^\top \nabla b_K(p), \;\; \text{ and } \;\; C_K \coloneqq \bE \left[ \xi_K^{(0)}\xi_K^{(1)} R_K^{(0)}R_K^{(1)\top} \right].
	\end{align*}
\end{theorem}

As shown in Lemma \ref{lem:MTRnormal}, $\sigma_{K,j}^{(d)}(p)$ corresponds to the asymptotic standard deviation of the MTR estimator for $D = d$.
Further, $cov_{K}(p)$ is the asymptotic covariance between the MTR estimators for $D = 1$ and $D = 0$, which is supposed to be non-zero in our case.
This non-zero covariance originates from replacing the unobserved membership indicator $\mathbf{1} \{ s = j \}$ with the membership probability $\pi_j$.

As a corollary of Theorem \ref{thm:normality}, the asymptotic properties of the feasible MTE estimator can be derived relatively easily with the following additional assumption.

\begin{assumption}\label{as:deriv}
	$\sup_{p \in [0, 1]} | \nabla b_K(p)^\top \alpha | = O(\sqrt{K}) \cdot \sup_{p \in [0, 1]} | b_K(p)^\top \alpha |$ for any $\alpha \in \mathbb{R}^K$.
\end{assumption}

\citet{chen2018optimal} show that this assumption holds true for B-splines and wavelets.

To proceed, let $\mathcal{C}(\mathcal{D})$ be the set of uniformly bounded continuous functions defined on $\mathcal{D}$.
Further, let $T$ be a generic random vector where $\text{supp}[T]$ is compact and $\dim(T)$ is finite.
We define the linear operator $\mathcal{P}_{n,K,j}^{(d)}$ that maps a given function $q \in \mathcal{C}(\text{supp}[T])$ to the sieve space defined by $b_K$ as follows:
\begin{align*}
    \mathcal{P}_{n,K,j}^{(d)}q &\coloneqq b_K(\cdot)^\top \mathbb{S}_{K, j} \left[ \Psi_{nK}^{(d)} \right]^{-1}  \frac{1}{n} \sum_{i = 1}^n R_{i, K}^{(d)} q(T_i),
\end{align*}
where $\Psi_{nK}^{(d)} \coloneqq n^{-1} \sum_{i=1}^n R_{i, K}^{(d)} R_{i, K}^{(d) \top}$.
% its inverse exists with probability approaching one by Lemma \ref{lem:matLLN}(i).
The functional form of $q$ may implicitly depend on $n$.
The operator norm of $\mathcal{P}_{n,K,j}^{(d)}$ is defined as
\begin{align*}
	\| \mathcal{P}_{n,K,j}^{(d)} \|_{\infty} \coloneqq \sup \left\{ \sup_{p \in [0, 1]} \left| \left( \mathcal{P}_{n,K,j}^{(d)}q\right)(p) \right| : q \in \mathcal{C}(\text{supp}[T]), \sup_{t \in \text{supp}[T]} |q(t)| = 1 \right\},
\end{align*}
which is typically of order $O_P(1)$ for splines and wavelets under some regularity conditions.\footnote{
    A set of easy-to-check conditions ensuring this is to verify the conditions in Lemma 7.1 in \citet{chen2015optimal} and show that $\mathbb{S}_{K, j} \left[ \Psi_{nK}^{(d)} \right]^{-1}$ is stochastically bounded in the $\ell$-infinity norm.
    See also Appendix B.5 of \cite{hoshino2021treatment}.
}

\begin{corollary} \label{cor:normality}
    Suppose that the assumptions in Theorem \ref{thm:normality} hold.
    If Assumption \ref{as:deriv}, $\zeta_0(K) \zeta_1(K) / \sqrt{n} = O(1)$, and $(\| \mathcal{P}_{n,K,j}^{(d)}\|_{\infty} + 1) \sqrt{K} / \| \nabla b_K(p) \| \to 0$ for both $d \in \{ 0, 1 \}$ hold additionally, then
	\begin{align*}
    	\frac{\sqrt{n}\left( \hat{\text{MTE}}_j(x,p) - \text{MTE}_j(x,p) \right)}{\sqrt{\left[ \sigma_{K,j}^{(1)}(p) \right]^2 + 2cov_{K}(p) + \left[ \sigma_{K,j}^{(0)}(p) \right]^2}} \overset{d}{\to} N(0,1).
	\end{align*}
\end{corollary}

As shown above, the feasible MTE estimator has the same asymptotic distribution as the infeasible estimator.
This is because the asymptotic distributions of the MTE estimators are determined only by the second-stage series estimator of the nonparametric component $\nabla g_j^{(d)}(p)$ whose convergence rate is slower than the first-stage ML estimator and the series estimator of the parametric component $\beta_j^{(d)}$.

The standard errors of the MTE estimators can be straightforwardly computed by replacing unknown terms in $\sigma_{K,j}^{(d)}(p)$ and $cov_K(p)$ with their empirical counterparts: $\hat \Psi_{nK}^{(d)} \coloneqq \frac{1}{n} \sum_{i=1}^n \hat R_{i, K}^{(d)} \hat R_{i, K}^{(d) \top}$, $\hat \Sigma_{nK}^{(d)} \coloneqq \frac{1}{n} \sum_{i=1}^n (\hat \xi_{i, K}^{(d)})^2 \hat R_{i, K}^{(d)} \hat R_{i, K}^{(d) \top}$, and $\hat C_{nK} \coloneqq \frac{1}{n} \sum_{i=1}^n \hat \xi_{i, K}^{(0)} \hat \xi_{i, K}^{(1)} \hat R_{i, K}^{(0)} \hat R_{i, K}^{(1) \top}$, where $\hat \xi_{i, K}^{(1)} \coloneqq D_i Y_i - \hat R_{i, K}^{(1)\top} \hat \theta_n^{(1)}$ and $\hat \xi_{i, K}^{(0)} \coloneqq (1 - D_i) Y_i - \hat R_{i, K}^{(0)\top} \hat \theta_n^{(0)}$.

%%%%%%%%%%%%%%%%%%%%%%%%%%%%%%%%%%
\section{Additional Discussion and Extensions} \label{sec:extension}

\subsection{Endogenous group membership with covariate-dependent membership probability}\label{subsec:hetero}

We consider relaxing Assumption \ref{as:membership} by allowing the membership variable $s$ to be dependent of both $\epsilon_j^D$ and a vector of covariates $W$.
%In the literature, similar types of finite mixture models have been considered in \cite{huang2012mixture} and \cite{huang2013nonparametric}.
For simplicity, we focus on the case of $S = 2$ only.
Suppose that $s$ and $D$ are determined by the following model:
\begin{align}\label{eq:member}
	\begin{split}
		s
		& = 2 - \mathbf{1}\{\pi(W) \ge U\}, \\
		D 
		& = \mathbf{1}\{P_j \ge V_j\} \qquad \text{if $s = j$ for $j \in \{1,2\}$},
	\end{split}
\end{align} 
where $U$ is an unobserved continuous random variable distributed as $\text{Uniform}[0,1]$ conditional on $X$, and $\pi$ is an unknown function that takes values on $[0,1]$.
The other components are the same as those in \eqref{eq:treatment}.
We assume that $W$ is conditionally independent of $U$ given $X$.
Then, $\Pr(s = 1|X, W) = \pi(W)$.
In addition, we assume that $W$ contains at least one continuous variable that is not included in $(X, Z_1, Z_2)$.

In this setup, we re-define the MTE parameter and the MTR function, respectively, as follows:
\begin{align} \label{eq:covariate-dependent-MTE}
\begin{split}
	\text{MTE}_j(x,q,p )
	& \coloneqq m_j^{(1)}(x,q,p) - m_j^{(0)}(x,q,p),\\
	m_j^{(d)}(x, q, p)
	& \coloneqq \bE [ \mu^{(d)} ( x, j, \epsilon^{(d)} ) | X = x, U = q, V_j = p].
\end{split}
\end{align}
Further, letting $Q \coloneqq \pi(W)$, we define 
\begin{align*}
	\psi_1(x, q, p_1, p_2) &
	\coloneqq \bE [ DY | X = x, Q = q, P_1 = p_1, P_2 = p_2],\\
	\psi_0(x, q, p_1, p_2) &
	\coloneqq \bE [ (1- D)Y | X = x, Q = q, P_1 = p_1, P_2 = p_2],\\
	\rho_{dj}(x, q, p_1, p_2) &
	\coloneqq \Pr( D = d, s = j | X = x, Q = q, P_1 = p_1, P_2 = p_2 ) \quad \text{for $d \in \{ 0, 1 \}$ and $j \in \{ 1, 2 \}$}.
\end{align*}
For identification of these functions, we first need to establish the identification of all model parameters in \eqref{eq:member}.
In Appendix \ref{subsubsec:endmember}, we present a supplementary identification result for an endogenous finite-mixture treatment choice model where the error terms are assumed to be jointly normal.
For notational simplicity, we denote the cross-partial derivatives with respect to $(q, p_j)$ by $\nabla_{qp_j}$; for instance, $\nabla_{qp_1} \psi_1(x, q, p_1, p_2) = \partial^2 \psi_1(x, q, p_1, p_2) / (\partial q \partial p_1)$.

We introduce the following identification conditions that can be visually understood using the DAG in Figure \ref{fig:dag2}.

\begin{assumption}\label{as:W}\hfil
	\begin{enumerate}[(i)]
		\item $(W, Z_1, Z_2)$ are independent of $(\epsilon^{(d)}, \epsilon^D_j, U)$ given $X$ for all $d$ and $j$.
		\item Each of $W$, $Z_1$, and $Z_2$ contains at least one continuous variable that is not included in $X$ and the rest.
		\item For each $j$, $(U, V_j)$ is continuously distributed on $[0, 1]^2$ with  conditional density $f_{UV_j}(\cdot, \cdot|X)$ given $X$.
	\end{enumerate}
\end{assumption}

\begin{figure}[!h]
	\begin{center}
		\includegraphics[width = 6cm, bb = 0 0 163 107]{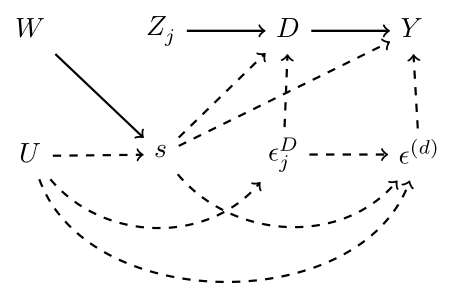}
		\caption{DAG under Assumption \ref{as:W}}
		\label{fig:dag2}
	\end{center}
\end{figure}

\begin{theorem}\label{thm:ident-MTR2}
	Suppose that Assumption \ref{as:W} holds.
	If $m_j^{(d)}(x, \cdot, \cdot)$ and $f_{UV_j}(\cdot, \cdot|X=x)$ are continuous, we have
	\begin{align*}
		m_j^{(d)}(x, q, p_j) = \frac{\nabla_{qp_j} \psi_d(x, q, p_1, p_2)}{\nabla_{qp_j} \rho_{dj}(x, q, p_1, p_2)},
		\quad \text{and} \quad
		f_{UV_j}(q, p_j | X = x) = (-1)^{d + j} \cdot \nabla_{qp_j} \rho_{dj}(x, q, p_1, p_2).
	\end{align*}
\end{theorem}

The following theorem shows that the group-wise MTE parameter $\text{MTE}_j(x, p)$ defined in \eqref{eq:groupMTE} can be recovered by the weighted average of $\text{MTE}_j(x, q, p)$ in \eqref{eq:covariate-dependent-MTE}.

\begin{theorem}\label{thm:ident-g-MTE}
	Suppose that Assumption \ref{as:W} holds.
	Then, we have
	\begin{align*}
		\text{MTE}_j(x, p) = \int_0^1 \text{MTE}_j(x, u, p) \omega_j(x, u, p) \mathrm{d}u,
		\quad 
		\text{for $j \in \{ 1, 2 \}$}, 
	\end{align*}
	where 
	\begin{align*}
		\omega_1(x, u, p) & \coloneqq \frac{ \Pr(u \le Q | X = x)   f_{UV_1}(u,p|X = x) }{\int_0^1 \Pr(u' \le Q | X = x)  f_{UV_1}(u', p| X = x) \mathrm{d}u'}, \\
		\omega_2(x, u, p) & \coloneqq \frac{\Pr(u > Q|X = x) f_{UV_2}(u, p|X = x)  }{\int_0^1 \Pr(u' > Q | X = x) f_{UV_2}(u',p|X = x)  \mathrm{d}u'}.
	\end{align*}
\end{theorem}

\subsection{Absence of continuous IVs}\label{subsec:late}

In some empirical situations, only discrete IVs are available, and many are binary.
In this subsection, we focus on a situation where only binary IVs are available and show that some LATE parameters are still identifiable by the Wald estimand.
%As is well-known, the LATE parameters can be nonparametrically identified without continuous IVs.
%In contrast, identification of MTE parameters generally requires either continuous IVs or parametric functional-form restrictions (e.g., \citealp{brinch2017beyond}).
For expositional simplicity, the condition $X = x$ is suppressed throughout this subsection.

Let $Z$ be an $S \times 1$ vector of binary instruments $Z = (Z_1, \dots, Z_S) \in \mathcal{Z}$, where $\mathcal{Z} \coloneqq \text{supp}[Z]$.
In this subsection, we allow $Z$ to contain some overlapping elements.
In an extreme case, when there is only one instrument common for all groups, we have $\mathcal{Z} = \{ \mathbf{0}_S, \mathbf{1}_S \}$, where $\mathbf{0}_S$ and $\mathbf{1}_S$ are $S \times 1$ vectors of zeros and ones, respectively.
Let $Y^{(d, z)}$ be the potential outcome when $D = d$ and $Z = z$.
Similarly, denote the potential treatment status when $s = j$ and $Z = z$ as $D_j^{(z)}$.

\begin{assumption}\label{as:late}\hfil
\begin{enumerate}[(i)]
    \item For each $d$ and $j$, $Y^{(d, z)} = Y^{(d)}$ and $D_j^{(z)} = D_j^{(z_j)}$ for any $z = (z_1, \dots, z_S) \in \mathcal{Z}$.
    \item $Z$ is independent of $(D_j^{(0)}, D_j^{(1)}, Y^{(0)}, Y^{(1)}, s)$ for all $j$.
    \item $\Pr(D_j^{(1)} = 1, D_j^{(0)} = 0 | s = j) > 0$ and $\Pr(D_j^{(1)} = 0, D_j^{(0)} = 1 | s = j) = 0$ for all $j$.
\end{enumerate}
\end{assumption}

Assumption \ref{as:late}(i) can be violated when the instruments affect the outcome directly or when the group-$j'$-specific instrument $Z_{j'}$ affects the treatment status of the individuals in group $j$ ($j \neq j'$).
Note that, under this condition, it holds that $D = D_j^{(z_j)}$ conditional on $s = j$ and $Z = z$.
Assumption \ref{as:late}(ii) is essentially the same as Assumption \ref{as:IV}(i).
Assumption \ref{as:late}(iii) is similar to the monotonicity condition in \citet{imbens1994identification}.
Following the literature, the individuals with $D_j^{(1)} = D_j^{(0)} = 1$ can be referred to as \textit{$Z_j$-always-takers}; those with $D_j^{(1)} = D_j^{(0)} = 0$ are \textit{$Z_j$-never-takers}; those with $D_j^{(1)} > D_j^{(0)}$ are \textit{$Z_j$-compliers}; and those with $D_j^{(1)} < D_j^{(0)}$ are \textit{$Z_j$-defiers}.
Hence, the condition (iii) ensures that for all $j$, there are $Z_j$-compliers but no $Z_j$-defiers in group $j$.
%When some IVs are common across groups, this monotonicity condition might be violated, as mentioned in Remark \ref{rem:monotonicity}.

\begin{theorem}\label{thm:late}
    Suppose that Assumption \ref{as:late} holds.
    Then, the weighted average of the group-wise LATEs can be identified as below:
    \begin{align*}
        \frac{\bE [Y|Z = \mathbf{1}_S] - \bE [Y|Z = \mathbf{0}_S]}{\bE [D|Z = \mathbf{1}_S] - \bE [D|Z = \mathbf{0}_S]}
        %= \sum_{j = 1}^S \bE [Y^{(1)} - Y^{(0)} | D_j^{(1)} > D_j^{(0)}, s = j] \frac{\Pr(D_j^{(1)} > D_j^{(0)}, s = j)}{\sum_{h=1}^S  \Pr(D_h^{(1)} > D_h^{(0)}, s = h) },
        = \sum_{j = 1}^S \bE [Y^{(1)} - Y^{(0)} | \text{$Z_j$-compliers}, s = j] \frac{\Pr(\text{$Z_j$-compliers}, s = j)}{\sum_{h=1}^S  \Pr(\text{$Z_h$-compliers}, s = h) }.
    \end{align*}
    Moreover, if $\mathbf{e}_j \in \mathcal{Z}$, where $\mathbf{e}_j$ is an $S \times 1$ unit vector with its $j$-th element equal to one, we can identify the LATE specific to group $j$:
    \begin{align*}
        \frac{\bE [Y|Z = \mathbf{e}_j] - \bE [Y|Z = \mathbf{0}_S]}{\bE [D|Z = \mathbf{e}_j] - \bE [D|Z = \mathbf{0}_S]}
        = \bE [Y^{(1)} - Y^{(0)} | \text{$Z_j$-compliers}, s = j].
    \end{align*}
\end{theorem}

%%%%%%%%%%%%%%%%%%%%%%%%%%%%%%%%%%

\section{Monte Carlo Experiments} \label{sec:simulation}

We conduct a set of Monte Carlo experiments to evaluate the finite sample performance of our estimators.
Here, we provide the main simulation results only, and some supplementary experiments are provided in Appendix \ref{sec:appendix:simulation}.

Setting $S = 2$ with the membership probabilities $(\pi_1, \pi_2) = (0.6, 0.4)$, the treatment variable is generated by $D = \mathbf{1}\{Z_j^\top \gamma_j \ge \epsilon_j^D\}$ for $s = j$, where $Z_j = (1, X_1, \zeta_j)^\top$, $X_1 \sim N(0, 1)$, $\zeta_j \sim N(0, 1)$, and $\epsilon_j^D \sim N(0, 1)$ for both $j \in \{ 1, 2 \}$.
We set $\gamma_1 = (\gamma_{11}, \gamma_{12}, \gamma_{13})^\top = (0, -0.5, 0.5)^\top$ and $\gamma_2 = (\gamma_{21}, \gamma_{22}, \gamma_{23})^\top = (0, 0.5, -0.5)^\top$.
The potential outcomes are generated by $Y^{(d)} = \sum_{j \in \{1,2\}} \mathbf{1}\{s = j\} X^\top \beta_j^{(d)} + \epsilon^{(d)}$, where $X = (1, X_1)^\top$, $\epsilon^{(0)} = \sum_{j \in \{1,2\}}\mathbf{1}\{s = j\}V_j + \eta^{(0)}$, $\epsilon^{(1)} = \sum_{j \in \{1,2\}}\mathbf{1}\{s = j\} V_j^2 + \eta^{(1)}$, and $\eta^{(d)} \sim N(0, 0.5^2)$ for both $d \in \{ 0, 1 \}$.
Here, $V_j = \Phi(\epsilon_j^D)$, and $\Phi$ denotes the standard normal CDF.
The coefficients are set to $\beta_1^{(0)} = (-1, 1)^\top$, $\beta_2^{(0)} = (1, 2)^\top$, $\beta_1^{(1)} = (1, -1)^\top$, and $\beta_2^{(1)} = (2, 1)^\top$.
For each setup, we considered two sample sizes, $n \in \{ 1000, 4000 \}$.

For the first-stage estimation of the finite mixture Probit model, we use the EM algorithm.
The second-stage MTE estimation is carried out using both infeasible and feasible estimators.
We employ B-splines of order $3$ for the basis functions.
The number of inner knots of B-splines, say $\tilde K$, is set to $\tilde K = 1$ when $n = 1000$ and $\tilde K \in \{1,2\}$ when $n = 4000$.
To stabilize the series regression, ridge regression with penalty $10 n^{-1}$ is also considered for comparison.
The simulation results reported below are based on 1,000 Monte Carlo replications.

Table \ref{table:mc1}(a) shows the bias and root mean squared error (RMSE) of estimating the group-wise MTE for both groups with $x = 0.5$ and $v \in \{0.2, 0.4, 0.6, 0.8\}$ (labelled respectively as MTE1.1, and MTE1.2, and so on for group 1, and similarly labelled for group 2).
Overall, the performance of the feasible estimator is almost the same as that of the infeasible estimator.
This is consistent with our asymptotic theory.
The bias of our estimator is satisfactorily small, except for cases close to the boundary.
The RMSE quickly decreases as the sample size increases, as long as the number of basis terms remain unchanged, as expected.
Theoretically, we need to employ a larger number of basis terms as the sample size increases, however using $\tilde K = 2$ seems too flexible and the increase in the variance is rather problematic even for $n=4000$ (of course, this result is more or less specific to our choice of the functional form for MTE).
The RMSE for group 1 tends to be smaller than that for group 2, probably because of the difference in their group sizes.
Introducing a penalty term can improve the overall RMSE; hence, we recommend employing ridge regression in practice with a moderate sample size.

Table \ref{table:mc1}(b) presents the simulation results for the ML estimation of the finite mixture Probit model.
For all estimation parameters, the bias is satisfactorily small, especially when $n = 4000$.
The RMSE is approximately halved when the sample size increases from $1000$ to $4000$, implying $\sqrt{n}$-consistency of our ML estimator.

\begin{table}[!p]
	\caption{Simulation results} \label{table:mc1}
	\begin{subtable}{\textwidth}
		\caption{MTE estimation}
		{\footnotesize
			\begin{center}
				\begin{tabular}{rrrrcrrrrcrrrr}
					\hline\hline
					& \multicolumn{3}{c}{\bfseries }&\multicolumn{1}{c}{\bfseries }&\multicolumn{4}{c}{\bfseries Group 1}&\multicolumn{1}{c}{\bfseries }&\multicolumn{4}{c}{\bfseries Group 2}\tabularnewline
					\cline{6-9} \cline{11-14}
					& \multicolumn{1}{c}{$n$}&\multicolumn{1}{c}{$\tilde K$}&\multicolumn{1}{c}{ridge}&\multicolumn{1}{c}{}&\multicolumn{1}{c}{MTE1.1}&\multicolumn{1}{c}{MTE1.2}&\multicolumn{1}{c}{MTE1.3}&\multicolumn{1}{c}{MTE1.4}&\multicolumn{1}{c}{}&\multicolumn{1}{c}{MTE2.1}&\multicolumn{1}{c}{MTE2.2}&\multicolumn{1}{c}{MTE2.3}&\multicolumn{1}{c}{MTE2.4}\tabularnewline
					\hline
					\multicolumn{14}{l}{\bfseries Bias for the feasible estimator}\tabularnewline
					&$1000$&$1$&$0$&&$ 0.270$&$ 0.052$&$-0.142$&$ 0.118$&&$-0.208$&$ 0.049$&$ 0.211$&$ 0.323$\tabularnewline
					&$1000$&$1$&$1$&&$ 0.036$&$ 0.043$&$-0.068$&$-0.084$&&$-0.166$&$-0.006$&$ 0.070$&$ 0.055$\tabularnewline
					&$4000$&$1$&$0$&&$ 0.044$&$ 0.020$&$-0.055$&$ 0.011$&&$ 0.045$&$ 0.072$&$ 0.091$&$ 0.076$\tabularnewline
					&$4000$&$1$&$1$&&$-0.055$&$ 0.014$&$-0.037$&$-0.054$&&$-0.121$&$ 0.022$&$ 0.071$&$-0.011$\tabularnewline
					&$4000$&$2$&$0$&&$ 0.064$&$ 0.020$&$-0.063$&$-0.015$&&$ 0.024$&$ 0.113$&$ 0.040$&$ 0.224$\tabularnewline
					&$4000$&$2$&$1$&&$-0.067$&$ 0.031$&$-0.057$&$-0.080$&&$-0.124$&$ 0.069$&$ 0.030$&$ 0.002$\tabularnewline
					\hline
					\multicolumn{14}{l}{\bfseries Bias for the infeasible estimator}\tabularnewline
					&$1000$&$1$&$0$&&$ 0.049$&$ 0.005$&$-0.014$&$ 0.005$&&$-0.064$&$ 0.010$&$ 0.071$&$ 0.109$\tabularnewline
					&$1000$&$1$&$1$&&$-0.225$&$-0.053$&$-0.036$&$-0.169$&&$-0.281$&$-0.085$&$-0.046$&$-0.166$\tabularnewline
					&$4000$&$1$&$0$&&$ 0.019$&$ 0.012$&$-0.002$&$-0.023$&&$ 0.007$&$ 0.003$&$ 0.000$&$-0.001$\tabularnewline
					&$4000$&$1$&$1$&&$-0.103$&$ 0.002$&$ 0.010$&$-0.080$&&$-0.161$&$-0.027$&$-0.008$&$-0.102$\tabularnewline
					&$4000$&$2$&$0$&&$ 0.013$&$ 0.016$&$-0.027$&$-0.005$&&$-0.019$&$ 0.025$&$-0.040$&$ 0.020$\tabularnewline
					&$4000$&$2$&$1$&&$-0.145$&$ 0.053$&$-0.032$&$-0.075$&&$-0.204$&$ 0.024$&$-0.041$&$-0.114$\tabularnewline
					\hline
					\multicolumn{14}{l}{\bfseries RMSE for the feasible estimator}\tabularnewline
					&$1000$&$1$&$0$&&$ 1.552$&$ 0.897$&$ 0.969$&$ 2.060$&&$ 2.861$&$ 1.575$&$ 1.672$&$ 3.472$\tabularnewline
					&$1000$&$1$&$1$&&$ 0.707$&$ 0.602$&$ 0.606$&$ 0.837$&&$ 0.844$&$ 0.780$&$ 0.785$&$ 0.886$\tabularnewline
					&$4000$&$1$&$0$&&$ 0.716$&$ 0.392$&$ 0.435$&$ 0.923$&&$ 1.123$&$ 0.699$&$ 0.650$&$ 1.445$\tabularnewline
					&$4000$&$1$&$1$&&$ 0.492$&$ 0.322$&$ 0.347$&$ 0.670$&&$ 0.632$&$ 0.492$&$ 0.445$&$ 0.747$\tabularnewline
					&$4000$&$2$&$0$&&$ 0.872$&$ 0.638$&$ 0.730$&$ 1.136$&&$ 1.444$&$ 1.001$&$ 0.947$&$ 1.887$\tabularnewline
					&$4000$&$2$&$1$&&$ 0.508$&$ 0.483$&$ 0.588$&$ 0.722$&&$ 0.674$&$ 0.688$&$ 0.615$&$ 0.780$\tabularnewline
					\hline
					\multicolumn{14}{l}{\bfseries RMSE for the infeasible estimator}\tabularnewline
					&$1000$&$1$&$0$&&$ 1.446$&$ 0.695$&$ 0.775$&$ 1.894$&&$ 2.230$&$ 1.134$&$ 0.984$&$ 2.528$\tabularnewline
					&$1000$&$1$&$1$&&$ 0.651$&$ 0.483$&$ 0.514$&$ 0.938$&&$ 0.863$&$ 0.658$&$ 0.629$&$ 0.896$\tabularnewline
					&$4000$&$1$&$0$&&$ 0.685$&$ 0.339$&$ 0.380$&$ 0.930$&&$ 1.140$&$ 0.574$&$ 0.491$&$ 1.249$\tabularnewline
					&$4000$&$1$&$1$&&$ 0.492$&$ 0.292$&$ 0.318$&$ 0.719$&&$ 0.713$&$ 0.429$&$ 0.383$&$ 0.790$\tabularnewline
					&$4000$&$2$&$0$&&$ 0.788$&$ 0.754$&$ 0.861$&$ 1.152$&&$ 1.408$&$ 1.213$&$ 1.116$&$ 1.402$\tabularnewline
					&$4000$&$2$&$1$&&$ 0.507$&$ 0.563$&$ 0.670$&$ 0.819$&&$ 0.751$&$ 0.786$&$ 0.723$&$ 0.835$\tabularnewline
					\hline
				\end{tabular}
			\end{center}
			Note: The column labeled ``ridge'' indicates whether the ridge regression is used (1 for ``yes'' and 0 for ``no'').}
	\end{subtable}
	
	\bigskip \bigskip
	
	\begin{subtable}{\textwidth}
		\caption{ML estimation of the finite mixture Probit model}
		{\footnotesize
			\begin{center}
				\begin{tabular}{rrcrrrcrrrcrr}
					\hline\hline
					\multicolumn{2}{c}{\bfseries }&&\multicolumn{3}{c}{\bfseries Group 1}&&\multicolumn{3}{c}{\bfseries Group 2}&&\multicolumn{2}{c}{\bfseries Membership}\tabularnewline
					\cline{4-6} \cline{8-10} \cline{12-13}
					&\multicolumn{1}{c}{$n$}&&\multicolumn{1}{c}{$\gamma_{11}$}&\multicolumn{1}{c}{$\gamma_{12}$}&\multicolumn{1}{c}{$\gamma_{13}$}&&\multicolumn{1}{c}{$\gamma_{21}$}&\multicolumn{1}{c}{$\gamma_{22}$}&\multicolumn{1}{c}{$\gamma_{23}$}&&\multicolumn{1}{c}{$\pi_1$}&\multicolumn{1}{c}{$\pi_2$}\tabularnewline
					\hline
					\multicolumn{13}{l}{\textbf{Bias}}\tabularnewline
					&$1000$&&$ 0.003$&$0.004$&$0.168$&&$-0.025$&$0.045$&$-0.273$&&$-0.017$&$0.017$\tabularnewline
					&$4000$&&$-0.003$&$0.008$&$0.040$&&$ 0.002$&$0.000$&$-0.065$&&$-0.008$&$0.008$\tabularnewline
					\hline
					\multicolumn{13}{l}{\textbf{RMSE}}\tabularnewline
					&$1000$&&$ 0.373$&$0.420$&$0.408$&&$ 0.581$&$0.697$&$ 0.600$&&$ 0.140$&$0.140$\tabularnewline
					&$4000$&&$ 0.155$&$0.179$&$0.152$&&$ 0.231$&$0.283$&$ 0.233$&&$ 0.094$&$0.094$\tabularnewline
					\hline
				\end{tabular}
			\end{center}
		}
	\end{subtable}
\end{table}

\section{An Empirical Illustration: Economic Returns to College Education} \label{sec:empirical}

In this empirical analysis, we investigate the effects of college education on income in the Japanese labor market.
There are two sources used for data collection.
The primary data source is the Japanese Life Course Panel Survey 2007 (wave 1),\footnote{
    Acknowledgement: The data for this secondary analysis, "Japanese Life Course Panel Surveys, wave 1, 2007, of the Institute of Social Science, The University of Tokyo," was provided by the Social Science Japan Data Archive, Center for Social Research and Data Archives, Institute of Social Science, The University of Tokyo.
    }
which includes detailed information on Japanese workers aged 20 to 40, including their working condition, annual income, education level, and family member characteristics.
The outcome variable $Y$ of interest is the respondent's annual income, and the treatment variable is defined as follows: $D = 1$ if the respondent has a college degree (including junior colleges and technical colleges) or higher, $D = 0$ otherwise.
The second data source is the School Basic Survey conducted by the Ministry of Education, Japan.
Using this dataset, we collect information on regional university enrollment statistics for each prefecture where the respondents were living at the age of 15 to create IVs for the treatment. 
Table \ref{table:variables} below shows the list of variables used in this analysis.
After excluding observations with missing data and those with zero income, the analysis is performed on 3,318 individuals.

\begin{table}[!h]
\footnotesize{
\caption{Variables used}
\begin{center}
\begin{tabular}{lll}
\hline\hline
 & \multicolumn{1}{c}{Variable} & \\
 \hline
$Y$ & Annual income (in million JPY). & \\
$D$ & Dummy variable: the respondent has a college degree or higher. & \\
$X$ & Dummy variable: the respondent is currently living in an urban area. & \\
 & Dummy variable: the respondent is male. & \\
 & Age.* & \\
 & Working experience in years for the current job.* & \\
 & Dummy variable: the respondent has a permanent job. & \\
 & Log of average working hours per day. & \\
 & Dummy variable: the respondent works overtime almost everyday. & \\
 & Dummy variable: the respondent has professional skills. & \\
 & Dummy variable: the respondent is a public worker. & \\
 & Dummy variable: the respondent is in a managerial position. & \\
 & Dummy variable: the respondent is working in a large company. & \\
 & Dummy variable: the respondent is married. & \\
 & Dummy variable: the respondent's partner (if any) is a part-time worker. & \\
 & Self-reported health status (in five scales). & \\
 & The respondent's father's education in years. & \\
 & The respondent's mother's education in years. & (Shorthand)\\
$\mathbf{Z}$ 
 & Dummy variable: the respondent is male. & \textit{male}\\
 & Age. & \textit{age}\\
 & The respondent's father's education in years. & \textit{feduc}\\
 & The respondent's mother's education in years. & \textit{meduc}\\
 &&\\
 & \multicolumn{2}{c}{(Variables exculded from the outcome equation)} \\ 
 & Capacity: 10000 $\times$ (the number of universities)/(the number of high-school graduates).** & \textit{cap}\\
 & The rate of university enrollment for high school graduates.** & \textit{runiv}\\
 & The rate of workforce participation for high school graduates.** & \textit{rwork}\\
 & Sum of the dummies for the possession of (own room/student desk/encyclopedias/visual dictionaries/PC) & \\
 & \qquad when he/she was 15 years old. & \textit{homeenv}\\
 & Log of the number of books in the respondent's home when he/she was 15 years old. & \textit{books}\\
 & Economic condition of the respondent's household when he/she was 15 years old (in five scales). & \textit{econom}\\
 & The number of elder siblings. & \textit{nesibs}\\
 & The number of younger siblings. & \textit{nysibs}\\
 \hline
 & * The square of these variables are also included in the regressors. & \\
 & ** These variables are all at prefecture level, and they are as of when the respondent was 15 years old. & \\
\end{tabular}
\label{table:variables}
\end{center}}
\end{table}

In this analysis, we employ a finite mixture Probit model with two latent groups in which group 1 is composed of individuals whose college enrollment decisions are mainly affected by the regional educational characteristics, and group 2 comprises individuals who are mainly affected by their home educational environment.
For the IVs that should be included in $Z_1$ and $Z_2$, we identify them by running a stepwise AIC-based variable selection.
In doing this, we assume that $Z_1$ contains at least \textit{cap}, \textit{runiv}, and \textit{rwork} and that $Z_2$ contains at least \textit{homeenv} and \textit{books}.
Note that we do not rule out the possibility that these variables are common IVs (in fact, \textit{cap} is identified as a common IV).\footnote{
	We estimated the models with $S = 3$ or larger using the same AIC-based model selection procedure.
	However, the EM algorithm for such models did not converge within tolerable parameter values (the results are not reported here due to word limit constraints, but are available on request).
	This would indicate that the two-group model is appropriate for our dataset.
}

In Table \ref{table:treatsel}, we present the result of our mixture model and that of the standard Probit model without mixture for comparison.
First of all, our treatment choice model has a significantly smaller AIC value than the model with no mixture.
The estimation results indicate that approximately 80\% of our observations are classified into group 1, which is composed of individuals who are influenced not only by regional characteristics of university enrollment, but also by their family characteristics such as their parents' education level and the number of siblings.
The remaining 20\% belong to group 2, which is composed of individuals whose home study environment is major determinant of the college entrance decisions.

\begin{table}[!h]
\footnotesize{
\caption{Estimation results: treatment choice models}
\begin{center}
\begin{tabular}{lcccccc}
\hline\hline
 & \multicolumn{4}{c}{Mixture Model ($S=2$)}  & \multicolumn{2}{c}{Standard Probit Model} \\
 & \multicolumn{2}{c}{\textbf{Group 1}} &  \multicolumn{2}{c}{\textbf{Group 2}} &  & \\
 \cline{2-3}\cline{4-5}
 & Estimate & $t$-value & Estimate & $t$-value & Estimate & $t$-value\\
 \hline
$\pi$ & 0.799  & 39.332  & 0.201  &  &  &  \\
Intercept & -3.045  & -7.357  & -18.444  & -2.007  & -3.112  & -8.960  \\
\textit{cap} & -0.016  & -0.740  & -0.519  & -1.654  & -0.031  & -1.899  \\
\textit{runiv} & 0.989  & 1.877  &  &  & 0.731  & 1.701  \\
\textit{rwork} & -0.926  & -2.088  &  &  & -0.909  & -2.525  \\
\textit{homeenv} &  &  & 0.725  & 1.585  & 0.043  & 2.067  \\
\textit{books} &  &  & 4.082  & 2.028  & 0.122  & 7.721  \\
\textit{econom} &  &  & 2.064  & 1.830  & 0.083  & 2.576  \\
\textit{nesibs} & -0.167  & -4.169  &  &  & -0.133  & -4.080  \\
\textit{nysibs} & -0.091  & -2.040  & -1.201  & -1.559  & -0.099  & -2.939  \\
\textit{male} &  &  &  &  & -0.010  & -0.223  \\
\textit{age} & 0.013  & 1.867  &  &  & 0.013  & 2.272  \\
\textit{feduc} & 0.140  & 10.438  &  &  & 0.115  & 10.350  \\
\textit{meduc} & 0.094  & 5.612  &  &  & 0.073  & 5.307  \\
 &  &  &  &  &  & \\
Log Likelihood & \multicolumn{4}{c}{-1978.771}  & \multicolumn{2}{c}{-1991.632} \\
AIC & \multicolumn{4}{c}{3989.542}  & \multicolumn{2}{c}{4009.300} \\
Sample Size & \multicolumn{6}{c}{3318}\\
\hline
\end{tabular}
\label{table:treatsel}
\end{center}}
\end{table}

Following the suggestions from the Monte Carlo results in Appendix \ref{sec:simulation}, we employ the third-order B-splines with $\tilde K = 1$ and use the ridge regression with the penalty equal to $10n^{-1}$ for the estimation of the MTEs.\footnote{
	Note that introducing a sufficiently small penalty term does not alter our asymptotic results.
}

Figure \ref{fig:mte} shows our main results.
We find that for those who belong to group 1, the effect of a college degree on income is not significant if they have higher unobserved costs of pursuing higher education, while it is significantly positive if the cost is moderate.
This result indicates the presence of a certain treatment selectivity for returns in group 1.
Recall that regional educational characteristics and family characteristics are the main factors affecting the college enrollment status of the members in this group.
Thus, we can expect that their personal willingness or resistance toward higher education is somewhat heterogeneous within the group, depending on their surrounding environment.
Such heterogeneity may also affect the magnitude of the educational returns, leading to the nonlinear shape of the MTE curve for group 1, as depicted in the figure.
In contrast, for the members in group 2, those who are influenced by their own study environment, the MTE curve is relatively flat and moderately positively significant in almost the entire region of $p$.
This implies that their unobserved costs are less dramatically related to educational returns than those in group 1.
	
\begin{figure}[!h]
    \begin{center}
    \includegraphics[width=15cm, bb = 0 0 1232 541]{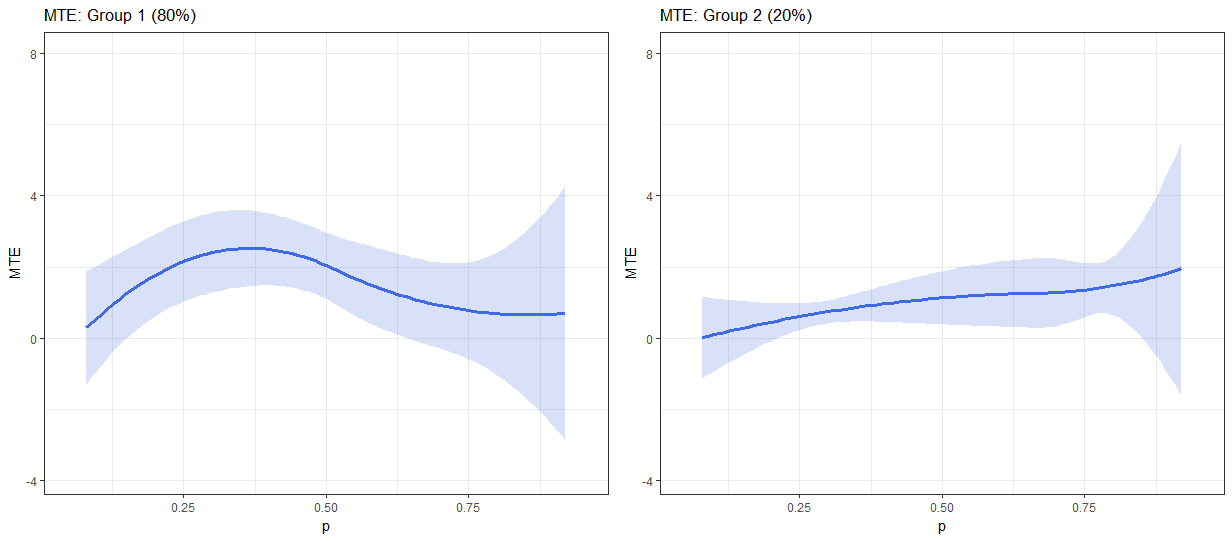}
    \caption{Estimated MTEs}
    \label{fig:mte}
    \end{center}
    \small{Note: In each panel, the solid line indicates the point estimate of the MTE evaluated at the sample mean of $X$, and the grayed area corresponds to the 95\% confidence interval.}
\end{figure}
%%%%%%%%%%%%%%%%%%%%%%%%%%%%%%%%%%

\section{Conclusion} \label{sec:conclusion}

This paper considered identification and estimation of MTE when the data is composed of a mixture of latent groups.
We developed a general treatment effect model with unobserved group heterogeneity by extending the Rubin's causal model to finite mixture models.
We proved that the MTE for each latent group can be separately identified under the availability of group-specific continuous IVs.
Based on our constructive identification result, we proposed the two-step semiparametric procedure for estimating the group-wise MTEs and established its asymptotic properties.
%The results of the Monte Carlo simulations show that our estimators perform well in finite samples. 
An empirical application to the estimation of economic returns to college education indicates the usefulness of the proposed model.

\section*{Acknowledgments}

The authors thank the editor, two anonymous referees, Toru Kitagawa, Yasushi Kondo, Ryo Okui, Myung Hwan Seo, Katsumi Shimotsu, Hisatoshi Tanaka, Yuta Toyama, and seminar participants at Hitotsubashi University, Seoul National University, and Waseda University for their valuable comments.
This work was supported by JSPS KAKENHI Grant Numbers 15K17039, 19H01473, and 20K01597.

\small
\clearpage
\bibliography{ref}

\begin{thebibliography}{42}
\providecommand{\natexlab}[1]{#1}
\providecommand{\selectlanguage}[1]{\relax}

\bibitem[{Bonhomme \textit{et~al.}(2016)Bonhomme, Jochmans, and
  Robin}]{bonhomme2016non}
Bonhomme, S., Jochmans, K., and Robin, J.M., 2016. Non-parametric estimation of
  finite mixtures from repeated measurements, \textit{Journal of the Royal
  Statistical Society: Series B (Statistical Methodology)}, 78~(1), 211--229.

\bibitem[{Brinch \textit{et~al.}(2017)Brinch, Mogstad, and
  Wiswall}]{brinch2017beyond}
Brinch, C.N., Mogstad, M., and Wiswall, M., 2017. Beyond {LATE} with a discrete
  instrument, \textit{Journal of Political Economy}, 125~(4), 985--1039.

\bibitem[{Butler and Louis(1997)}]{butler1997consistency}
Butler, S.M. and Louis, T.A., 1997. Consistency of maximum likelihood
  estimators in general random effects models for binary data, \textit{The
  Annals of Statistics}, 25~(1), 351--377.

\bibitem[{Cameron and Heckman(1998)}]{cameron1998life}
Cameron, S.V. and Heckman, J.J., 1998. Life cycle schooling and dynamic
  selection bias: Models and evidence for five cohorts of american males,
  \textit{Journal of Political Economy}, 106~(2), 262--333.

\bibitem[{Carneiro \textit{et~al.}(2011)Carneiro, Heckman, and
  Vytlacil}]{carneiro2011estimating}
Carneiro, P., Heckman, J.J., and Vytlacil, E.J., 2011. Estimating marginal
  returns to education, \textit{American Economic Review}, 101~(6), 2754--81.

\bibitem[{Carneiro and Lee(2009)}]{carneiro2009estimating}
Carneiro, P. and Lee, S., 2009. Estimating distributions of potential outcomes
  using local instrumental variables with an application to changes in college
  enrollment and wage inequality, \textit{Journal of Econometrics}, 149~(2),
  191--208.

\bibitem[{Chen \textit{et~al.}(2001)Chen, Chen, and
  Kalbfleisch}]{chen2001modified}
Chen, H., Chen, J., and Kalbfleisch, J.D., 2001. A modified likelihood ratio
  test for homogeneity in finite mixture models, \textit{Journal of the Royal
  Statistical Society: Series B (Statistical Methodology)}, 63~(1), 19--29.

\bibitem[{Chen and Christensen(2015)}]{chen2015optimal}
Chen, X. and Christensen, T.M., 2015. Optimal uniform convergence rates and
  asymptotic normality for series estimators under weak dependence and weak
  conditions, \textit{Journal of Econometrics}, 188~(2), 447--465.

\bibitem[{Chen and Christensen(2018)}]{chen2018optimal}
Chen, X. and Christensen, T.M., 2018. Optimal sup-norm rates and uniform
  inference on nonlinear functionals of nonparametric iv regression,
  \textit{Quantitative Economics}, 9~(1), 39--84.

\bibitem[{Compiani and Kitamura(2016)}]{compiani2016using}
Compiani, G. and Kitamura, Y., 2016. Using mixtures in econometric models: a
  brief review and some new results, \textit{The Econometrics Journal}, 19,
  C95--C127.

\bibitem[{Cornelissen \textit{et~al.}(2016)Cornelissen, Dustmann, Raute, and
  Sch{\"o}nberg}]{cornelissen2016late}
Cornelissen, T., Dustmann, C., Raute, A., and Sch{\"o}nberg, U., 2016. From
  late to mte: Alternative methods for the evaluation of policy interventions,
  \textit{Labour Economics}, 41, 47--60.

\bibitem[{Deb and Gregory(2018)}]{deb2018heterogeneous}
Deb, P. and Gregory, C.A., 2018. Heterogeneous impacts of the supplemental
  nutrition assistance program on food insecurity, \textit{Economics Letters},
  173, 55--60.

\bibitem[{Dempster \textit{et~al.}(1977)Dempster, Laird, and
  Rubin}]{dempster1977maximum}
Dempster, A.P., Laird, N.M., and Rubin, D.B., 1977. Maximum likelihood from
  incomplete data via the {EM} algorithm, \textit{Journal of the Royal
  Statistical Society: Series B (Statistical Methodology)}, 39~(1), 1--22.

\bibitem[{Doyle~Jr(2007)}]{doyle2007child}
Doyle~Jr, J.J., 2007. Child protection and child outcomes: Measuring the
  effects of foster care, \textit{American Economic Review}, 97~(5),
  1583--1610.

\bibitem[{Follmann and Lambert(1991)}]{follmann1991identifiability}
Follmann, D.A. and Lambert, D., 1991. Identifiability of finite mixtures of
  logistic regression models, \textit{Journal of Statistical Planning and
  Inference}, 27~(3), 375--381.

\bibitem[{Harris and Sosa-Rubi(2009)}]{harris2009impact}
Harris, J.E. and Sosa-Rubi, S.G., 2009. Impact of "{S}eguro {P}opular" on
  prenatal visits in mexico, 2002-2005: Latent class model of count data with a
  discrete endogenous variable, {N}BER Working Paper 14995.

\bibitem[{Heckman and Pinto(2018)}]{heckman2018unordered}
Heckman, J.J. and Pinto, R., 2018. Unordered monotonicity,
  \textit{Econometrica}, 86~(1), 1--35.

\bibitem[{Heckman and Vytlacil(1999)}]{heckman1999local}
Heckman, J.J. and Vytlacil, E.J., 1999. Local instrumental variables and latent
  variable models for identifying and bounding treatment effects,
  \textit{Proceedings of the National Academy of Sciences}, 96~(8), 4730--4734.

\bibitem[{Heckman and Vytlacil(2005)}]{heckman2005structural}
Heckman, J.J. and Vytlacil, E.J., 2005. Structural equations, treatment
  effects, and econometric policy evaluation, \textit{Econometrica}, 73~(3),
  669--738.

\bibitem[{Holland(2017)}]{holland2017penalized}
Holland, A.D., 2017. Penalized spline estimation in the partially linear model,
  \textit{Journal of Multivariate Analysis}, 153, 211--235.

\bibitem[{Hoshino and Yanagi(2021)}]{hoshino2021treatment}
Hoshino, T. and Yanagi, T., 2021. Treatment effect models with strategic
  interaction in treatment decisions, \textit{arXiv preprint arXiv:1810.08350}.

\bibitem[{Imbens and Angrist(1994)}]{imbens1994identification}
Imbens, G.W. and Angrist, J.D., 1994. Identification and estimation of local
  average treatment effects, \textit{Econometrica}, 62~(2), 467--475.

\bibitem[{Kasahara and Shimotsu(2014)}]{kasahara2014non}
Kasahara, H. and Shimotsu, K., 2014. Non-parametric identification and
  estimation of the number of components in multivariate mixtures,
  \textit{Journal of the Royal Statistical Society: Series B (Statistical
  Methodology)}, 97--111.

\bibitem[{Keane and Wolpin(1997)}]{keane1997career}
Keane, M.P. and Wolpin, K.I., 1997. The career decisions of young men,
  \textit{Journal of Political Economy}, 105~(3), 473--522.

\bibitem[{Kitamura and Laage(2018)}]{kitamura2018nonparametric}
Kitamura, Y. and Laage, L., 2018. Nonparametric analysis of finite mixtures,
  \textit{arXiv preprint arXiv:1811.02727}.

\bibitem[{Lee and Salani{\'e}(2018)}]{lee2018identifying}
Lee, S. and Salani{\'e}, B., 2018. Identifying effects of multivalued
  treatments, \textit{Econometrica}, 86~(6), 1939--1963.

\bibitem[{McLachlan and Peel(2004)}]{mclachlan2004finite}
McLachlan, G. and Peel, D., 2004. \textit{Finite Mixture Models}, Wiley, New
  York.

\bibitem[{Mogstad \textit{et~al.}(2018)Mogstad, Santos, and
  Torgovitsky}]{mogstad2018using}
Mogstad, M., Santos, A., and Torgovitsky, A., 2018. Using instrumental
  variables for inference about policy relevant treatment parameters,
  \textit{Econometrica}, 86~(5), 1589--1619.

\bibitem[{Mogstad and Torgovitsky(2018)}]{mogstad2018identification}
Mogstad, M. and Torgovitsky, A., 2018. Identification and extrapolation of
  causal effects with instrumental variables, \textit{Annual Review of
  Economics}, 10, 577--613.

\bibitem[{Mogstad \textit{et~al.}(2020{\natexlab{a}})Mogstad, Torgovitsky, and
  Walters}]{mogstad2020causal}
Mogstad, M., Torgovitsky, A., and Walters, C.R., 2020{\natexlab{a}}. The causal
  interpretation of two-stage least squares with multiple instrumental
  variables, \textit{NBER Working Paper Series}, 25691.

\bibitem[{Mogstad \textit{et~al.}(2020{\natexlab{b}})Mogstad, Torgovitsky, and
  Walters}]{mogstad2020policy}
Mogstad, M., Torgovitsky, A., and Walters, C.R., 2020{\natexlab{b}}. Policy
  evaluation with multiple instrumental variables, \textit{NBER Working Paper
  Series}, 27546.

\bibitem[{Mountjoy(2019)}]{mountjoy2019community}
Mountjoy, J., 2019. Community colleges and upward mobility, \textit{SSRN
  Working Paper}, 3373801.

\bibitem[{Munkin and Trivedi(2010)}]{munkin2010disentangling}
Munkin, M.K. and Trivedi, P.K., 2010. Disentangling incentives effects of
  insurance coverage from adverse selection in the case of drug expenditure: a
  finite mixture approach, \textit{Health Economics}, 19~(9), 1093--1108.

\bibitem[{Newey(1997)}]{newey1997convergence}
Newey, W.K., 1997. Convergence rates and asymptotic normality for series
  estimators, \textit{Journal of Econometrics}, 79~(1), 147--168.

\bibitem[{Rubin(1974)}]{rubin1974estimating}
Rubin, D.B., 1974. Estimating causal effects of treatments in randomized and
  nonrandomized studies, \textit{Journal of Educational Psychology}, 66~(5),
  688.

\bibitem[{Samoilenko \textit{et~al.}(2018)Samoilenko, Blais, Boucoiran, and
  Lefebvre}]{samoilenko2018using}
Samoilenko, M., Blais, L., Boucoiran, I., and Lefebvre, G., 2018. Using a
  mixture-of-bivariate-regressions model to explore heterogeneity of effects of
  the use of inhaled corticosteroids on gestational age and birth weight among
  pregnant women with asthma, \textit{American Journal of Epidemiology},
  187~(9), 2046--2059.

\bibitem[{Tamer(2003)}]{tamer2003incomplete}
Tamer, E., 2003. Incomplete simultaneous discrete response model with multiple
  equilibria, \textit{The Review of Economic Studies}, 70~(1), 147--165.

\bibitem[{Train(2008)}]{train2008algorithms}
Train, K.E., 2008. E{M} algorithms for nonparametric estimation of mixing
  distributions, \textit{Journal of Choice Modelling}, 1~(1), 40--69.

\bibitem[{Wang and Tchetgen~Tchetgen(2018)}]{wang2018bounded}
Wang, L. and Tchetgen~Tchetgen, E., 2018. Bounded, efficient and multiply
  robust estimation of average treatment effects using instrumental variables,
  \textit{Journal of the Royal Statistical Society. Series B, Statistical
  methodology}, 80~(3), 531.

\bibitem[{Woo and Sriram(2006)}]{woo2006robust}
Woo, M.J. and Sriram, T., 2006. Robust estimation of mixture complexity,
  \textit{Journal of the American Statistical Association}, 101~(476),
  1475--1486.

\bibitem[{Zhou and Xie(2019)}]{zhou2019marginal}
Zhou, X. and Xie, Y., 2019. Marginal treatment effects from a propensity score
  perspective, \textit{Journal of Political Economy}, 127~(6), 3070--3084.

\bibitem[{Zhu and Zhang(2004)}]{zhu2004hypothesis}
Zhu, H.T. and Zhang, H., 2004. Hypothesis testing in mixture regression models,
  \textit{Journal of the Royal Statistical Society: Series B (Statistical
  Methodology)}, 66~(1), 3--16.

\end{thebibliography}

\clearpage

\appendix
\begin{center}
	\Large{ \textbf{Supplementary Appendix (not for publication)} }
\end{center}
\section{Appendix: Proofs of Theorems} \label{sec:proofs}

This appendix collects the proofs of the theorems.
The lemmas used to prove Theorem \ref{thm:normality} and Corollary \ref{cor:normality} are relegated to Appendix \ref{sec:lemma}.
Below, we denote the conditional density of a random variable $T$ as $f_T(\cdot | \cdot)$.

\subsection{Proof of Theorem \ref{thm:ident-MTR1}}

%%%%%%%%%%%%%% Proof: \Pr(s = j, D = 1 | X = x, \mathbf{P} = \mathbf{p}) = \pi_j p_j %%%%%%%%%%%%%%
\begin{comment}
\begin{itemize}
    \item $\Pr(s = j, D = 1 | X = x, \mathbf{P} = \mathbf{p}) = \Pr(s = j, p_j \ge V_j | X = x, \mathbf P = \mathbf p )$
    
    $\because (s = j) \land (D = 1) = (s = j) \land ((s = 1 \land P_1 \ge V_1) \lor (s = 2 \land P_2 \ge V_2)) = (s = j) \land (P_j \ge V_j)$
    \item Assumption \ref{as:IV}(i) $\Rightarrow$ $\Pr(s = j, p_j \ge V_j | X = x, \mathbf P = \mathbf p ) = \Pr(s = j, p_j \ge V_j | X = x )$
    \item Assumption \ref{as:membership}(i) $\Rightarrow$ $\Pr(s = j, p_j \ge V_j | X = x ) = \Pr(s = j | X = x )\Pr( p_j \ge V_j | X = x )$
    \item Assumption \ref{as:membership}(ii) $\Rightarrow$ $\Pr(s = j | X = x )\Pr( p_j \ge V_j | X = x ) = \pi_j p_j$
    \item Hence, $\Pr(s = j, D = 1 | X = x, \mathbf{P} = \mathbf{p}) = \pi_j p_j$
    \end{itemize}
\end{comment}
%%%%%%%%%%%%%%% Proof of Theorem 2.1: D = 1 %%%%%%%%%%%%%%%

We provide the proof for $m^{(1)}_j(x, p)$ only, as the proof for $m^{(0)}_j(x, p)$ is analogous.
First, observe that
\begin{align*}
    \psi_1(x, \mathbf p)
    & = \bE \left[ D Y^{(1)} \middle| X = x, \mathbf{P} = \mathbf{p} \right] \\
    & = \sum_{j = 1}^S \bE [Y^{(1)} |X = x, \mathbf{P} = \mathbf{p}, s =j, D = 1] \cdot \Pr(s = j, D = 1 | X = x, \mathbf{P} = \mathbf{p})\\
    & = \sum_{j = 1}^S \bE [Y^{(1)} |X = x, s = j, V_j \le p_j] \cdot \pi_j p_j,
\end{align*}
under Assumptions \ref{as:IV}(i) and \ref{as:membership}.
Further, it holds that
\begin{align*}
    \bE [Y^{(1)} |X = x, s = j, V_j \le p_j]  
    & = \int_0^{p_j} \bE [Y^{(1)} |X = x, s = j, V_j = v] \frac{f_{V_j}(v | X = x, s = j)}{\Pr(V_j \le p_j | X = x, s = j)} \mathrm{d}v\\ 
    & = \frac{1}{p_j} \int_0^{p_j} m_j^{(1)}(x,v) \mathrm{d}v,
\end{align*}
by Assumption \ref{as:membership}(i).
Therefore, $\psi_1(x, \mathbf p) = \sum_{j = 1}^S \pi_j \int_0^{p_j} m_j^{(1)}(x,v) \mathrm{d}v$, and the Leibniz integral rule leads to $\partial \psi_1(x, \mathbf p) / \partial p_j = \pi_j \cdot m_j^{(1)}(x,p_j)$.
This completes the proof. \qed

%%%%%%%%%%%%%%% Proof of Theorem 2.1: D = 0 %%%%%%%%%%%%%%%
\begin{comment}
\begin{align*}
    \psi_0(x, \mathbf p)
    & = \bE \left[ (1 - D) Y^{(0)} |X = x, \mathbf{P} = \mathbf{p} \right] \\
    & = \sum_{j = 1}^S \bE [Y^{(0)} |X = x, \mathbf{P} = \mathbf{p}, s =j, D = 0] \cdot \Pr(s = j, D = 0 | X = x, \mathbf{P} = \mathbf{p})\\
    & = \sum_{j = 1}^S \bE [Y^{(0)} |X = x, s = j, V_j > p_j] \cdot \pi_j (1 - p_j) 
\end{align*}
by Assumptions \ref{as:IV}(i) and \ref{as:membership}.
Further,
\begin{align*}
    \bE [Y^{(0)} |X = x, s = j, V_j > P_j]  
    & = \frac{\int_{p_j}^1 \bE [Y^{(0)} |X = x, s = j, V_j = v] f_{V_j}(v | X = x, s = j) \mathrm{d}v}{\Pr(V_j > p_j | X = x, s = j)}\\ 
    & = \frac{\int_{p_j}^1 m_j^{(0)}(x,v) \mathrm{d}v}{1 - p_j} 
\end{align*}
by Assumption \ref{as:membership}(i).
Therefore,
\begin{align*}
    \psi_0(x, \mathbf p) =\sum_{j = 1}^S \pi_j \int_{p_j}^1 m_j^{(0)}(x,v) \mathrm{d}v 
    \Longrightarrow \frac{\partial \psi_0(x, \mathbf p) }{\partial p_j} = -\pi_j m_j^{(0)}(x,p_j).
\end{align*}
\end{comment}

%%%%%%%%%%%%%%%%%%%%%%%%%%%%%%%%%%%%%%%%%%%%%%%%%%%%%
\bigskip

Here, we introduce additional notations for the subsequent discussions.
Let $\delta_j^{(1)} \coloneqq \mathbf{1}\{ s = j, V_j \le P_j \} = D \cdot \mathbf{1}\{s = j\}$ and $\delta_j^{(0)} \coloneqq\mathbf{1}\{ s = j, V_j > P_j \} =  (1 - D) \cdot \mathbf{1}\{s = j\}$.
Note that $D = \sum_{j = 1}^S \delta_j^{(1)}$ and $1 - D = \sum_{j = 1}^S \delta_j^{(0)}$.
By \eqref{eq:paramodel-Y}, we can write
\begin{align*}
	DY 
	= \sum_{j = 1}^{S} \delta_j^{(1)} X^\top \beta_j^{(1)} + D \epsilon^{(1)}
	&= \sum_{j = 1}^{S} \delta_j^{(1)} ( X^\top \beta_j^{(1)} + g_j^{(1)}(P_j) / P_j) + e^{(1)}\\
	&= \underbrace{\sum_{j = 1}^{S} \delta_j^{(1)} ( X^\top \beta_j^{(1)} + b_K(P_j)^\top \alpha_j^{(1)} / P_j)}_{\eqqcolon \: T_K^{(1)\top}\theta^{(1)}}  + r_K^{(1)} + e^{(1)}\\
	&= R_K^{(1)\top} \theta^{(1)} + r_K^{(1)} + \underbrace{B_K^{(1)}  + e^{(1)}}_{\eqqcolon \: \xi_K^{(1)}},
\end{align*}
where $r_K^{(1)} \coloneqq \sum_{j=1}^{S} \delta_j^{(1)} [g_j^{(1)}(P_j) - b_K(P_j)^\top \alpha_j^{(1)}] / P_j$,
\begin{align}\label{eq:definitionBe}
\begin{split}
    & B_K^{(1)} \coloneqq \left( T_K^{(1)} - R_K^{(1)} \right)^\top \theta^{(1)} = \sum_{j = 1}^S (\delta_j^{(1)} - \pi_j P_j) \left( X^\top \beta_j^{(1)} + b_K(P_j)^\top \alpha_j^{(1)} / P_j \right), \\
	& e^{(1)} \coloneqq D \epsilon^{(1)} - \sum_{j = 1}^S \delta_j^{(1)} g_j^{(1)}(P_j) / P_j = \sum_{j=1}^S \delta_j^{(1)} \left[ \epsilon^{(1)} - g_j^{(1)}(P_j) / P_j \right].
\end{split}
\end{align}
Let $\mathbf{Y}^{(1)} = (D_1 Y_1, \dots, D_n Y_n)^\top$, $\mathbf{R}_K^{(1)} = (R_{1, K}^{(1)}, \dots, R_{n, K}^{(1)})^\top$, $\mathbf{r}_K^{(1)} = (r_{1, K}^{(1)}, \dots, r_{n, K}^{(1)})^\top$, $\mathbf{B}_K^{(1)} = (B_{1, K}^{(1)}, \dots, B_{n, K}^{(1)})^\top$, $\mathbf{e}^{(1)} = (e_1^{(1)}, \dots, e_n^{(1)})^\top$, and $\bm{\xi}_K^{(1)} = (\xi_{1, K}^{(1)}, \dots, \xi_{n, K}^{(1)})^\top$.
The infeasible estimator for $\theta^{(1)}$ can be written as
\begin{align*}
	\tilde \theta_n^{(1)}
	&= \left( \mathbf{R}_K^{(1)\top} \mathbf{R}_K^{(1)} \right)^{-} \mathbf{R}_K^{(1)\top} \mathbf{Y}^{(1)}\\
	&= \theta^{(1)} + \left[ \Psi_{nK}^{(1)} \right]^{-1} \mathbf{R}_K^{(1)\top} \mathbf{r}_K^{(1)} / n + \left[ \Psi_{nK}^{(1)} \right]^{-1} \mathbf{R}_K^{(1)\top} \bm{\xi}_K^{(1)} / n.
\end{align*}
Similarly, noting that $DY = \hat R_K^{(1)\top}\theta^{(1)} + \hat \Delta_K^{(1)} + r_K^{(1)} + \xi_K^{(1)}$ with $\hat \Delta_K^{(1)} \coloneqq ( R_K^{(1)} - \hat R_K^{(1)} )^\top \theta^{(1)}$, the feasible estimator $\hat \theta_n^{(1)}$ can be written as
\begin{align*}
	\hat \theta_n^{(1)}
	%& = \left( \hat{\mathbf{R}}_K^{(1)\top} \hat{\mathbf{R}}_K^{(1)} \right)^{-} \hat{\mathbf{R}}_K^{(1)\top} \mathbf{Y}^{(1)}\\
	= \theta^{(1)} + \left[ \hat \Psi_{nK}^{(1)} \right]^{-1} \hat{\mathbf{R}}_K^{(1)\top} \hat{\bm{\Delta}}_K^{(1)} / n + \left[ \hat \Psi_{nK}^{(1)} \right]^{-1} \hat{\mathbf{R}}_K^{(1)\top} \mathbf{r}_K^{(1)} / n + \left[ \hat \Psi_{nK}^{(1)} \right]^{-1} \hat{\mathbf{R}}_K^{(1)\top} \bm{\xi}_K^{(1)} / n,
\end{align*}
where $\hat{\mathbf{R}}_K^{(1)} = (\hat R_{1, K}^{(1)}, \dots, \hat R_{n, K}^{(1)})^\top$, and $\hat{\bm{\Delta}}_K^{(1)} = (\hat{\Delta}_{1, K}^{(1)}, \dots, \hat{\Delta}_{n, K}^{(1)})^\top$.
In the same manner, for the estimators of $\theta^{(0)}$, we have
\begin{align*}
    & \tilde \theta_n^{(0)} - \theta^{(0)}=  \left[ \Psi_{nK}^{(0)} \right]^{-1} \mathbf{R}_K^{(0)\top} \mathbf{r}_K^{(0)} / n + \left[ \Psi_{nK}^{(0)} \right]^{-1} \mathbf{R}_K^{(0)\top} \bm{\xi}_K^{(0)} / n, \\
	& \hat \theta_n^{(0)} - \theta^{(0)} = \left[ \hat \Psi_{nK}^{(0)} \right]^{-1} \hat{\mathbf{R}}_K^{(0)\top} \hat{\bm{\Delta}}_K^{(0)} / n + \left[ \hat \Psi_{nK}^{(0)} \right]^{-1} \hat{\mathbf{R}}_K^{(0)\top} \mathbf{r}_K^{(0)} / n + \left[ \hat \Psi_{nK}^{(0)} \right]^{-1} \hat{\mathbf{R}}_K^{(0)\top} \bm{\xi}_K^{(0)} / n,
\end{align*}
where the definitions of the newly introduced variables should be clear from the context.

\subsection{Proof of Theorem \ref{thm:normality}}

By \eqref{eq:sigmalbound} and \eqref{eq:MTRlinear}, we observe that
\begin{align*}
	& \sqrt{n}\left( \tilde{\text{MTE}}_j(x,p) - \text{MTE}_j(x,p) \right) \\
	& = \sqrt{n} \left( \tilde m^{(1)}_j(x, p) - m^{(1)}_j (x, p) \right) - \sqrt{n} \left( \tilde m^{(0)}_j(x, p) - m^{(0)}_j (x, p) \right) \\
	& = \nabla b_K(p)^\top  \mathbb{S}_{K,j} \left[ \Psi_K^{(1)} \right]^{-1} \mathbf{R}_K^{(1)\top} \bm{\xi}_K^{(1)} / \sqrt{n} + \nabla b_K(p)^\top  \mathbb{S}_{K,j} \left[ \Psi_K^{(0)} \right]^{-1} \mathbf{R}_K^{(0)\top} \bm{\xi}_K^{(0)} / \sqrt{n} + o_P(\| \nabla b_K(p) \|).
\end{align*}
Thus, as shown in Lemma \ref{lem:MTRnormal}(i), the term on the left-hand side is approximated by the sum of two asymptotically normal random variables with mean zero.
Note that unlike standard treatment effect estimators, the covariance of these two terms is not zero:
\begin{align*}
	\frac{1}{n} \sum_{l = 1}^n \sum_{m = 1}^n  \bE \left[ R_{l,K}^{(0)} R_{m,K}^{(1)\top} \xi_{l,K}^{(0)} \xi_{m,K}^{(1)} \bigg| \{X_i, \mathbf Z_i\}_{i = 1}^n \right] 
	& = \frac{1}{n} \sum_{l = 1}^n R_{l,K}^{(0)} R_{l,K}^{(1)\top} \bE \left[ \xi_{l,K}^{(0)} \xi_{l,K}^{(1)} \bigg| \{X_i, \mathbf Z_i\}_{i = 1}^n \right],
\end{align*}
by $\bE [\xi_K^{(d)} | X, \mathbf Z] = 0$ for both $d \in \{ 0, 1 \}$ and Assumption \ref{as:iid}, and
\begin{align*}
	\bE \left[ \xi_K^{(0)} \xi_K^{(1)} \bigr| X, \mathbf Z \right] 
	& = \bE \left[ (T_K^{(0)\top}\theta^{(0)} - R_K^{(0)\top}\theta^{(0)} + e^{(0)}) (T_K^{(1)\top}\theta^{(1)} - R_K^{(1)\top}\theta^{(1)} + e^{(1)}) \bigr| X, \mathbf Z \right] \\
	& = - R_K^{(0)\top}\theta^{(0)} \cdot R_K^{(1)\top}\theta^{(1)} \neq 0 \;\; \text{in general.} 
\end{align*}
The remainder of the proof is straightforward. 
\qed 

\subsection{Proof of Corollary \ref{cor:normality}}

By \eqref{eq:sigmalbound} and \eqref{eq:MTRlinear2}, we have
\begin{align*}
	& \sqrt{n}\left( \hat{\text{MTE}}_j(x,p) - \text{MTE}_j(x,p) \right) \\
	& = \nabla b_K(p)^\top  \mathbb{S}_{K,j} \left[ \Psi_K^{(1)} \right]^{-1} \mathbf{R}_K^{(1)\top} \bm{\xi}_K^{(1)} / \sqrt{n} + \nabla b_K(p)^\top  \mathbb{S}_{K,j} \left[ \Psi_K^{(0)} \right]^{-1} \mathbf{R}_K^{(0)\top} \bm{\xi}_K^{(0)} / \sqrt{n} + o_P(\| \nabla b_K(p) \|).
\end{align*}
The same argument as in the proof of Theorem \ref{thm:normality} leads to the desired result.
\qed

%%%%%%%%%%%%%%%%%%%%%%%%%%%%%%%%%%%%%%%%%%%%%%%%%%%%%
\subsection{Proof of Theorem \ref{thm:ident-MTR2}}

To save the space, we provide the proof for the case of $d = 1$ and $j = 1$ only.
Under Assumption \ref{as:W}, we have
\begin{align*}
	\rho_{11}(x, q, p_1, p_2) 
	= \Pr(U \le Q, V_1 \le P_1 | X = x, Q = q, P_1 = p_1, P_2 = p_2)
	= \int_0^{p_1} \int_0^q f_{U V_1}(u, v_1 | X = x) \mathrm{d}u \mathrm{d}v_1.
\end{align*}
On the other hand, by the law of iterated expectations, it holds that
\begin{align*}
	\psi_1(x, q, p_1, p_2) 
	& = \bE [Y^{(1)} | D = 1, s = 1, X = x, Q = q, P_1 = p_1, P_2 = p_2] \rho_{11}(x, q, p_1, p_2) \\
	& \qquad + \bE [Y^{(1)} | D = 1, s = 2, X = x, Q = q, P_1 = p_1, P_2 = p_2] \rho_{12}(x, q, p_1, p_2) \\
	& = \bE [ \mu^{(1)} ( x, 1, \epsilon^{(1)} )  | U \le q, V_1 \le p_1, X = x] \Pr(U \le q, V_1 \le p_1 | X = x) \\
	& \qquad + \bE [\mu^{(1)} ( x, 2, \epsilon^{(1)} )  | U > q, V_2 \le p_2, X = x] \Pr(U > q, V_2 \le p_2 | X = x) \\
	& = \int_0^{p_1} \int_0^q m_1^{(1)}(x, u, v_1) f_{U V_1}(u, v_1 | X = x) \mathrm{d}u \mathrm{d}v_1 + \int_0^{p_2} \int_q^1 m_2^{(1)}(x, u, v_2) f_{U V_2}(u, v_2 | X = x) \mathrm{d}u \mathrm{d}v_2.
\end{align*}
Thus, the Leibniz integral rule completes the proof.

\subsection{Proof of Theorem \ref{thm:ident-g-MTE}}

%%%%%%%%%%%%%%% Proof of Theorem 4.2: s = 1 %%%%%%%%%%%%%%%
We only prove for the case of $s = 1$.
By the law of iterated expectations, we observe that 
\begin{align*}
    \text{MTE}_1(x, p) 
    &= \bE \left[ \bE [Y^{(1)} - Y^{(0)} | X = x, V_1 = p, Q, s = 1] \middle| X = x, V_1 = p, s = 1 \right]\\
    &= \int_0^1 \bE [Y^{(1)} - Y^{(0)} | X = x, V_1 = p, Q = q, s = 1] f_Q(q | X = x, V_1 = p,  s = 1) \mathrm{d}q.
\end{align*}
Further, by Assumption \ref{as:W}(i),
\begin{align*}
 	\bE [Y^{(1)} - Y^{(0)} | X = x, V_1 = p, Q = q, s = 1]
    & = \bE [\mu^{(1)} ( x, 1, \epsilon^{(1)} ) - \mu^{(0)} ( x, 1, \epsilon^{(0)} ) | X = x, V_1 = p, Q = q, U \le q]\\
    & = \frac{\int_0^1 \mathbf{1}\{u \le q\} \text{MTE}_1(x, u, p) f_U(u|X = x, V_1 = p) \mathrm{d}u}{\Pr(U \le q|X = x, V_1 = p)} \\
    & = \frac{\int_0^1 \mathbf{1}\{u \le q\} \text{MTE}_1(x, u, p) f_{UV_1}(u, p |X = x) \mathrm{d}u}{\Pr(U \le q|X = x, V_1 = p)} ,
\end{align*}
where the last equality follows from the fact that $V_1 \sim \text{Uniform}[0, 1]$ conditional on $X = x$.
On the other hand, Bayes' theorem implies that
\begin{align*}
    f_Q(q | X = x, V_1 = p,  s = 1)
    %& = \frac{\Pr(s = 1 | X = x, V_1 = p, Q = q) f_Q(q | X = x, V_1 = p)}{\Pr(s = 1|X = x, V_1 = p)} \\
    & = \frac{\Pr(U \le q | X = x, V_1 = p) f_Q(q | X = x)}{\Pr(s = 1|X = x, V_1 = p)},
\end{align*}
under Assumption \ref{as:W}(i).
By the law of iterated expectations and Assumption \ref{as:W}(i), we have
\begin{align*}
    \Pr(s = 1 | X = x, V_1 = p ) 
    &= \bE [ \Pr(s = 1 | X = x,  V_1 = p, Q) | X = x, V_1 = p ]\\
    &= \int_0^1 \Pr(U \le q | X = x, V_1 = p, Q = q) f_Q(q|X = x, V_1 = p) \mathrm{d}q\\
    &= \int_0^1 \int_0^1 \mathbf{1}\{u \le q\} f_U(u|X = x, V_1 = p) f_Q(q|X = x) \mathrm{d}q \mathrm{d}u \\
    & = \int_0^1 \Pr(u \le Q | X = x)  f_{UV_1}(u, p| X = x) \mathrm{d}u.
\end{align*}
Combining these results yields
\begin{align*}
	\text{MTE}_1(x, p)
	& = \frac{\int_0^1 \int_0^1 \text{MTE}_1(x, u, p) \mathbf{1}\{u \le q\}f_{UV_1}(u,p|X = x) f_Q(q|X = x)   \mathrm{d}q \mathrm{d}u}{\int_0^1 \Pr(u' \le Q | X = x)  f_{UV_1}(u', p| X = x) \mathrm{d}u'} \\
	& = \int_0^1 \text{MTE}_1(x, u, p) \frac{ \Pr(u \le Q | X = x)   f_{UV_1}(u,p|X = x) }{\int_0^1 \Pr(u' \le Q | X = x)  f_{UV_1}(u', p| X = x) \mathrm{d}u'}\mathrm{d}u .
\end{align*}
This completes the proof. \qed

\subsection{Proof of Theorem \ref{thm:late}}

We provide the proof of the first result only since the second result can be shown analogously.
First, observe that
\begin{align*}
    \bE [D|Z=\mathbf{1}_S]
    %&= \sum_{j=1}^S \bE \left[ \mathbf{1}\{s = j\} D \middle| Z = \mathbf{1}_S \right]\\
    &= \sum_{j = 1}^S \bE [\mathbf{1}\{s = j\} D_j^{(1)} | Z = \mathbf{1}_S] \\
    %&= \sum_{j = 1}^S \Big\{ \bE [\mathbf{1}\{s = j\} D_j^{(1)} ] \Big\}\\
    &= \sum_{j = 1}^S \Big\{ \Pr(D_j^{(1)} = 1, D_j^{(0)} = 0, s = j) + \Pr(D_j^{(1)} = 1, D_j^{(0)} = 1, s = j) \Big\},
\end{align*}
where the first equality follows from Assumption \ref{as:late}(i), and the second follows from Assumption \ref{as:late}(ii).
Similarly, we can show that $\bE [D|Z=\mathbf{0}_S] = \sum_{j = 1}^S \Pr(D_j^{(1)} = 1, D_j^{(0)} = 1, s = j)$ under Assumption \ref{as:late}(iii).
%\begin{align*}
%    \bE [D|Z=\mathbf{0}_S]
%    &= \sum_{j = 1}^S \Big\{ \Pr(D_j^{(1)} = 0, D_j^{(0)} = 1, s = j) + \Pr(D_j^{(1)} = 1, D_j^{(0)} = 1, s = j) \Big\}\\
%    &= \sum_{j = 1}^S \Pr(D_j^{(1)} = 1, D_j^{(0)} = 1, s = j),
%\end{align*}
%where the last follows from Assumption \ref{as:late}(iii).
Thus, we have
\[
    \bE [D|Z=\mathbf{1}_S] - \bE [D|Z=\mathbf{0}_S] = \sum_{j = 1}^S \Pr(\text{$Z_j$-compliers}, s = j).
\]

Next, observe that $\bE [Y|Z=\mathbf{1}_S] = \sum_{j=1}^S \bE [ \mathbf{1}\{s = j\} Y | Z = \mathbf{1}_S ]$, and that by the law of iterated expectations,
\begin{align*}
    \bE [\mathbf{1}\{s = j\} Y | Z = \mathbf{1}_S] 
    &= \bE [Y^{(1)}| D_j^{(1)} = 1, D_j^{(0)} = 1, s = j] \Pr(D_j^{(1)} =1, D_j^{(0)} = 1, s = j)\\
    & \qquad + \bE [Y^{(0)}|D_j^{(1)} = 0, D_j^{(0)} = 0, s = j] \Pr(D_j^{(1)} = 0, D_j^{(0)} = 0, s = j)\\
    & \qquad + \bE [Y^{(1)}|D_j^{(1)} = 1, D_j^{(0)} = 0, s = j] \Pr(D_j^{(1)} = 1, D_j^{(0)} = 0, s = j),
\end{align*}
by Assumption \ref{as:late}.
In the same manner, it holds that
\begin{align*}
    \bE [\mathbf{1}\{s = j\} Y | Z = \mathbf{0}_S] 
    &= \bE [Y^{(1)}| D_j^{(1)} = 1, D_j^{(0)} = 1, s = j] \Pr(D_j^{(1)} =1, D_j^{(0)} = 1, s = j)\\
    & \qquad + \bE [Y^{(0)}|D_j^{(1)} = 0, D_j^{(0)} = 0, s = j] \Pr(D_j^{(1)} = 0, D_j^{(0)} = 0, s = j)\\
    & \qquad + \bE [Y^{(0)}|D_j^{(1)} = 1, D_j^{(0)} = 0, s = j] \Pr(D_j^{(1)} = 1, D_j^{(0)} = 0, s = j).
\end{align*}
Thus, 
\[
    \bE [Y|Z=\mathbf{1}_S] - \bE [Y|Z=\mathbf{0}] = \sum_{j=1}^S \bE [Y^{(1)} - Y^{(0)}|\text{$Z_j$-compliers}, s = j] \Pr(\text{$Z_j$-compliers}, s = j).
\]
This completes the proof. \qed

\section{Appendix: Lemmas}\label{sec:lemma}
This appendix collects the lemmas that are used in the proof of Theorem \ref{thm:normality} and Corollary \ref{cor:normality}.
In the following, we present the results for $D = 1$ only, and those for $D = 0$ are similar and thus omitted to save space.
Below, we often use the notation $\mathbb{S}$ to denote either of $\mathbb{S}_{X, j}$ and $\mathbb{S}_{K, j}$, and we suppress the superscript $^{(1)}$ if there is no confusion.
For a matrix $A$, we denote $||A||_2 = \sqrt{\lambda_\text{max}(A^\top A)}$ as its spectral norm.

\begin{comment}
\begin{lemma}\label{lem:P}
	Under Assumption \ref{as:parametricML}, $\max_{1 \le i \le n} |\hat P_{ji} - P_{ji}| = O_P(n^{-1/2})$ for all $j = 1, \dots, S$.
\end{lemma}

\begin{proof}
	The mean value theorem leads to $\hat P_{ji} = P_{ji} + f_j(Z_{ji}^\top \bar \gamma_{n, j}) \cdot Z_{ji}^\top (\hat \gamma_{n, j} - \gamma_j)$ where $\bar \gamma_{n, j}$ is between $\hat \gamma_{n, j}$ and $\gamma_{j}$.
	Thus, $\max_{1 \le i \le n} |\hat P_{ji} - P_{ji}| \le \max_{1 \le i \le n} ( |f_j(Z_{ji}^\top \bar \gamma_{n, j})| \cdot \| Z_{ji} \|) \cdot \| \hat \gamma_{n, j} - \gamma_j \| = O_P(n^{-1/2})$.
\end{proof}
\end{comment}

\begin{lemma}\label{lem:matLLN}
	Suppose that Assumptions \ref{as:iid} -- \ref{as:error} hold.
	\begin{enumerate}[(i)]
		\item $\left\| \Psi_{nK}^{(1)} - \Psi_K^{(1)} \right\|_2 = O_P (\zeta_0(K)\sqrt{(\log K)/ n})$.
		\item $\left\| \left[ \Psi_{nK}^{(1)} \right]^{-1} - \left[ \Psi_K^{(1)} \right]^{-1} \right\|_2 = O_P (\zeta_0(K)\sqrt{(\log K)/ n})$.
		\item $\left\| \mathbb{S} \left[ \Psi_{nK}^{(1)} \right]^{-1} \mathbf{R}_K^{(1)\top}\mathbf{B}_K^{(1)} / n \right\| = O_P\left(\sqrt{\text{tr}\left\{\mathbb{S}\mathbb{S}^\top \right\}/n}\right)$.
		\item $\left\| \mathbb{S} \left[ \Psi_{nK}^{(1)} \right]^{-1} \mathbf{R}_K^{(1)\top}\mathbf{r}_K^{(1)} / n \right\| = O_P(K^{-\mu_0})$.
		\item $\left\| \mathbb{S} \left[ \Psi_{nK}^{(1)} \right]^{-1} \mathbf{R}_K^{(1)\top}\mathbf{e}^{(1)} / n \right\| = O_P\left(\sqrt{\text{tr}\left\{\mathbb{S}\mathbb{S}^\top \right\}/n}\right)$.
	\end{enumerate}
\end{lemma}

\begin{proof}
	(i), (ii) The proofs are the same as those of Lemma A.1(i) and (iii) in \citet{hoshino2021treatment}.
	\bigskip
	
	%(ii) This is straightforward from $\| \Psi_{nK}^{-1} - \Psi_K^{-1} \| = \| \Psi_{nK}^{-1} ( \Psi_{nK} - \Psi_K ) \Psi_K^{-1} \|$, Assumption \ref{as:eigen}(i), and result (i).
	%\bigskip
	
	(iii) Since $\bE [\delta_j^{(1)} | X, \mathbf Z] = \pi_j P_j$, we have $\bE [T_K^{(1)} |X, \mathbf Z] = R_K^{(1)}$ and $\bE [B_K^{(1)} | X, \mathbf Z] = 0$.
	Then, under Assumption \ref{as:iid}, $\bE [ B_{l,K}^{(1)} B_{k,K}^{(1)} | \{X_i, \mathbf Z_i \}_{i = 1}^n ]  = 0$ for any $l, k \in \{ 1,\dots, n \}$ such that $l \neq k$.
	Additionally, observe that $\max_{1 \le i \le n} | B_{i,K}^{(1)}| = O(1)$ holds from Assumptions \ref{as:parametricML}(i),(ii), and \ref{as:outcome}(i) and
	\begin{align}\label{eq:uniformbounded}
	\sup_{p \in [0,1]} \left| b_K(p)^\top \alpha_j^{(1)} \right|
	& \le \sup_{p \in [0,1]} \left| b_K(p)^\top \alpha_j^{(1)} - g^{(1)}_j(p)\right| + \sup_{p \in [0,1]} \left| g^{(1)}_j(p)\right| = O(K^{-\mu_0}) + O(1),
	\end{align}
	by Assumption \ref{as:series}(i).
	Thus, noting that $\bE [\mathbf{B}_K \mathbf{B}_K^\top | \{X_i, \mathbf Z_i \}_{i = 1}^n ]$ is a diagonal matrix whose diagonal elements are $O(1)$, we have 
	\begin{align*}
	    \bE \left[ \| \mathbb{S} \Psi_{nK}^{-1} \mathbf{R}_K^\top \mathbf{B}_K / n \|^2 | \{X_i, \mathbf Z_i \}_{i = 1}^n \right]
	    & = \text{tr} \{ \mathbb{S} \Psi_{nK}^{-1} \mathbf{R}_K^\top \bE [ \mathbf{B}_K \mathbf{B}_K^\top | \{X_i, \mathbf Z_i \}_{i = 1}^n ] \mathbf{R}_K \Psi_{nK}^{-1} \mathbb{S}^\top \} / n^2 \\
	    & \le O(1/n) \cdot \text{tr} \{\mathbb{S} \Psi_{nK}^{-1} \Psi_{nK} \Psi_{nK}^{-1} \mathbb{S}^\top \} = O_P( \text{tr} \{ \mathbb{S} \mathbb{S}^\top \} / n),
	\end{align*}
	where the last equality follows from Assumption \ref{as:eigen}(i) and result (ii).
	Then, the result follows from Markov's inequality.
	
	\bigskip
	
	(vi), (v) For (iv), since $\min_{1 \le i \le n} P_{ji} > 0$ for any $j$ under Assumptions \ref{as:parametricML}(i) and (ii), we have
	\begin{align*}
	    \max_{1 \le i \le n} |r_{i, K}^{(1)}|
	    \le O(1) \cdot \sum_{j=1}^S  \max_{1 \le i \le n} \left| g_j^{(1)}(P_{ji}) - b_K(P_{ji})^\top \alpha_j^{(1)} \right| 
	    = O(K^{-\mu_0})
	\end{align*}
    by Assumption \ref{as:series}(i).
	 For (v), note that 
	\begin{align}\label{eq:errormean1}
	\begin{split}
    	\bE [e^{(1)} | X, \mathbf Z]
    	& = \sum_{j = 1}^S \left[ \bE [\delta_j \epsilon^{(1)} | X, \mathbf Z] -  \bE [ \delta_j | X, \mathbf Z]g_j^{(1)}(P_j)/P_j\right] \\
    	& =  \sum_{j = 1}^S \left[ \bE [\epsilon^{(1)} | X, \mathbf Z,  s = j, V_j \le P_j]\cdot \pi_j P_j -  \pi_j g_j^{(1)}(P_j) \right] \\
    	& =  \sum_{j = 1}^S \left[ \pi_j \int_0^{P_j} \bE [\epsilon^{(1)} | s = j, V_j = v] \mathrm{d}v -  \pi_j g_j^{(1)}(P_j) \right] = 0,
    \end{split}
	\end{align}
    by Assumptions \ref{as:IV}(i), \ref{as:membership}, and \ref{as:outcome}(i), so that $\bE [\mathbf{e}^{(1)} \mathbf{e}^{(1)^\top} | \{ X_i, \mathbf{Z}_i \}_{i=1}^n]$ is a diagonal matrix whose diagonal elements are $O(1)$ by Assumptions \ref{as:iid} and \ref{as:error}.
    Then, the rest of the proofs are similar to the proof of Lemma A.2 in \citet{hoshino2021treatment}.
\end{proof}

\bigskip

To prove the next lemma, define
\begin{align*}
\begin{array}{lll}
    \underset{S\dim(X) \times 1}{\mathbf{X}_i} \coloneqq (X_i^\top, \dots, X_i^\top)^\top,
    &
    \multicolumn{2}{l}{\underset{SK \times 1}{\mathbf{b}_{i, K}} \coloneqq (b_K(P_{1i})^\top, \dots, b_K(P_{Si})^\top)^\top,}\\
    \underset{n \times S\dim(X)}{\mathbf{X}} = (\mathbf{X}_1, \dots, \mathbf{X}_n)^\top,
    &
    \underset{n \times SK}{\mathbf{b}_K} = (\mathbf{b}_{1,K}, \dots, \mathbf{b}_{n,K})^\top,
    &
    \underset{n \times d_{SXK}}{\mathbf{W}_K} \coloneqq (\mathbf{X}, \mathbf{b}_K),\\
    \multicolumn{2}{l}{\underset{d_{SXK} \times 1}{\varUpsilon_{i, K}^{(1)}} \coloneqq (\pi_1 P_{1i} \mathbf{1}_{\dim(X)}^\top, \dots, \pi_S P_{Si} \mathbf{1}_{\dim(X)}^\top, \pi_1 \mathbf{1}_K^\top, \dots, \pi_S \mathbf{1}_K^\top)^\top,}
    &
    \underset{n \times d_{SXK}}{\bm{\varUpsilon}_K^{(1)}} \coloneqq (\varUpsilon_{1, K}^{(1)}, \dots, \varUpsilon_{n, K}^{(1)})^\top,
\end{array}
\end{align*}
and define analogously $\hat{\mathbf{b}}_K$, $\hat{\mathbf{W}}_K \coloneqq (\mathbf{X}, \hat{\mathbf{b}}_K)$, and $\hat{\bm{\varUpsilon}}_K^{(1)}$.
Then, we can write $\mathbf{R}_K^{(1)} = \bm{\varUpsilon}_K^{(1)} \circ \mathbf{W}_K$ and $\hat{\mathbf{R}}_K^{(1)} = \hat{\bm{\varUpsilon}}_K^{(1)} \circ \hat{\mathbf{W}}_K$, where $\circ$ denotes the Hadamard product.

\bigskip

\begin{lemma}\label{lem:estmatLLN}
	Suppose that Assumptions \ref{as:iid} -- \ref{as:error} hold.
	\begin{enumerate}[(i)]
		\item $\left\| \hat{\Psi}_{nK}^{(1)} - \Psi_{nK}^{(1)} \right\|_2 = O_P(\zeta_1(K) / \sqrt{n})$.
		\item $\left\| \left[ \hat{\Psi}_{nK}^{(1)} \right]^{-1} - \left[ \Psi_{nK}^{(1)} \right]^{-1} \right\|_2 = O_P(\zeta_1(K) / \sqrt{n})$.
		\item $\left\| \mathbb{S} \left[ \hat{\Psi}_{nK}^{(1)} \right]^{-1} \hat{\mathbf{R}}_K^{(1)\top} \hat{\bm{\Delta}}_K^{(1)} / n \right\| = O_P(n^{-1/2})$.
		\item $\left\| \mathbb{S} \left[ \hat{\Psi}_{nK}^{(1)} \right]^{-1} \hat{\mathbf{R}}_K^{(1)\top} \mathbf{B}_K^{(1)} / n \right\| = O_P \left( \sqrt{\text{tr}\{\mathbb{S} \mathbb{S}^\top\} / n} \right) + O_P(\zeta_1(K) \sqrt{K} / n) + O_P(\zeta_2(K) / n)$.
		\item $\left\| \mathbb{S} \left[ \hat{\Psi}_{nK}^{(1)} \right]^{-1} \hat{\mathbf{R}}_K^{(1)\top} \mathbf{r}_K^{(1)} / n \right\| = O_P \left( K^{-\mu_0} \right)$.
		\item $\left\| \mathbb{S} \left[ \hat{\Psi}_{nK}^{(1)} \right]^{-1} \hat{\mathbf{R}}_K^{(1)\top} \mathbf{e}^{(1)} / n \right\| = O_P \left( \sqrt{\text{tr}\{\mathbb{S} \mathbb{S}^\top\} / n} \right)$.
	\end{enumerate}
\end{lemma}

\begin{proof}
	(i) By the triangle inequality, we have
	\begin{align}\label{eq:differencePsi}
	\begin{split}
		\left\| \hat{\Psi}_{nK} - \Psi_{nK} \right\|_2
		& \le \left\| \hat{\Psi}_{nK} - \Psi_{nK} \right\| \\
		& \le \left\| \left( \hat{\mathbf{R}}_K - \mathbf{R}_K \right)^\top \left( \hat{\mathbf{R}}_K - \mathbf{R}_K \right) / n \right\| + 2\left\| \mathbf{R}_K^\top \left( \hat{\mathbf{R}}_K - \mathbf{R}_K \right) / n \right\|.
	\end{split}
	\end{align}
	For the first term of \eqref{eq:differencePsi}, the mean value theorem and Assumption \ref{as:parametricML} lead to $b_K(\hat P_{ji}) - b_K(P_{ji}) = \nabla b_K(\bar P_{ji}) \cdot O_P(n^{-1/2})$ where $\bar P_{ji}$ is between $\hat P_{ji}$ and $P_{ji}$.
	Let $\nabla \bar{\mathbf{b}}_{i, K} = (\nabla b_K(\bar P_{1i})^\top, \dots, \nabla b_K(\bar P_{Si})^\top)^\top$ and $\nabla \bar{\mathbf{b}}_K = (\nabla \bar{\mathbf{b}}_{1, K}, \dots, \nabla \bar{\mathbf{b}}_{n, K})^\top$.
	By the triangle inequality and Assumptions \ref{as:parametricML}(iii) and \ref{as:outcome}(i), we have
	\begin{align} \label{eq:differenceR}
	\begin{split}
		\| \hat{\mathbf{R}}_K - \mathbf{R}_K \|
		& \le \left\| (\hat{\bm{\varUpsilon}}_K - \bm{\varUpsilon}_K) \circ \hat{\mathbf{W}}_K \right\| + \left\| \bm{\varUpsilon}_K \circ (\hat{\mathbf{W}}_K - \mathbf{W}_K) \right\|\\
		& \le O_P(n^{-1/2}) \cdot \left\{ \left\| \hat{\mathbf{b}}_K \right\| + \left\| \nabla \bar{\mathbf{b}}_K \right\| \right\}
		= O_P(\zeta_0(K)) + O_P(\zeta_1(K)) 
		= O_P(\zeta_1(K)).
	\end{split}
	\end{align}
	Thus, $\| ( \hat{\mathbf{R}}_K - \mathbf{R}_K )^\top ( \hat{\mathbf{R}}_K - \mathbf{R}_K ) / n \| \le \| \hat{\mathbf{R}}_K - \mathbf{R}_K \|^2 / n = O_P(\zeta_1^2(K) / n)$.
	For the second term of \eqref{eq:differencePsi}, we have $\| \mathbf{R}_K^\top ( \hat{\mathbf{R}}_K - \mathbf{R}_K ) / n \|^2 = \text{tr} \{ ( \hat{\mathbf{R}}_K - \mathbf{R}_K )^\top \mathbf{R}_K \mathbf{R}_K^\top ( \hat{\mathbf{R}}_K - \mathbf{R}_K ) \} / n^2 \le O_P(1/n) \cdot \| \hat{\mathbf{R}}_K - \mathbf{R}_K \|^2 = O_P(\zeta_1^2(K) / n)$ by Lemma \ref{lem:matLLN}(i) and \eqref{eq:differenceR}.
	Thus, the second term is of order $O_P(\zeta_1(K) / \sqrt{n})$, and we obtain the desired result.
		
	\bigskip
	
	(ii) The proof is the same as that of Lemma A.1(iii) in \citet{hoshino2021treatment}.
	
	\bigskip
	
	(iii) By Lemma \ref{lem:matLLN}(ii), result (ii), and Assumption \ref{as:eigen}(i), we have
	\begin{align*}
	    \left\| \mathbb{S} \hat{\Psi}_{nK}^{-1} \hat{\mathbf{R}}_K^\top \hat{\bm{\Delta}}_K / n \right\|^2
	    = \text{tr} \{ \hat{\bm{\Delta}}_K^\top \hat{\mathbf{R}}_K \hat{\Psi}_{nK}^{-1} \mathbb{S}^\top \mathbb{S} \hat{\Psi}_{nK}^{-1} \hat{\mathbf{R}}_K^\top \hat{\bm{\Delta}}_K \} / n^2
	    \le O_P(n^{-1}) \cdot \| \hat{\bm{\Delta}}_K \|^2.
	\end{align*}
	Note that $\hat{\bm{\Delta}}_K = [\hat{\bm{\varUpsilon}}_K \circ (\mathbf{W}_K - \hat{\mathbf{W}}_K)]\theta^{(1)} + [(\bm{\varUpsilon}_K - \hat{\bm{\varUpsilon}}_K) \circ \mathbf{W}_K]\theta^{(1)}$.
	By the mean value theorem, it is easy to see that 
	\begin{align*}
	    \left\| [\hat{\bm{\varUpsilon}}_K \circ (\mathbf{W}_K - \hat{\mathbf{W}}_K)]\theta^{(1)} \right\|^2 
	    %& = \sum_{i = 1}^n \left( \sum_{j = 1}^S \hat \pi_{n, j} [ b_K(P_{ji}) - b_K(\hat P_{ji}) ]^\top \alpha_j^{(1)} \right)^2 \\
	    = \sum_{i = 1}^n \left( \sum_{j = 1}^S \hat \pi_{n, j} (P_{ji} - \hat P_{ji}) \cdot \nabla b_K(\bar P_{ji})^\top \alpha_j^{(1)} \right)^2 = O_P(1)
	\end{align*}
	from
	\begin{align*}
		\sup_{p \in [0, 1]} \left| \nabla b_K(p)^\top \alpha_j^{(1)} \right|
		\le \sup_{p \in [0, 1]} \left| \nabla b_K(p)^\top \alpha_j^{(1)} - \nabla g_j^{(1)}(p) \right| + \sup_{p \in [0, 1]} \left| \nabla g_j^{(1)}(p) \right|
		= O(K^{-\mu_1}) + O(1).
	\end{align*}
	The same argument shows that $\| [(\bm{\varUpsilon}_K - \hat{\bm{\varUpsilon}}_K) \circ \mathbf{W}_K]\theta^{(1)} \| = O_P(1)$.
	Thus, $\| \hat{\mathbf{\Delta}}_K \| = O_P(1)$, and we have the desired result.
	
	\bigskip
	
	(iv) By Lemma \ref{lem:matLLN}(iii), we have
	\begin{equation*}\label{eq:lemmaBK}
	\begin{split}
		\mathbb{S} \hat{\Psi}_{nK}^{-1} \hat{\mathbf{R}}_K^\top \mathbf{B}_K / n
		&= \mathbb{S} \Psi_{nK}^{-1} \mathbf{R}_K^\top \mathbf{B}_K / n + \mathbb{S} (\hat{\Psi}_{nK}^{-1} - \Psi_{nK}^{-1}) \mathbf{R}_K^\top \mathbf{B}_K / n + \mathbb{S} \hat \Psi_{nK}^{-1} (\hat{\mathbf{R}}_K - \mathbf{R}_K)^\top \mathbf{B}_K / n\\
		&= \mathbb{S} (\hat{\Psi}_{nK}^{-1} - \Psi_{nK}^{-1}) \mathbf{R}_K^\top \mathbf{B}_K / n + \mathbb{S} \hat \Psi_{nK}^{-1} (\hat{\mathbf{R}}_K - \mathbf{R}_K)^\top \mathbf{B}_K / n + O_P\left(\sqrt{\text{tr}\{\mathbb{S} \mathbb{S}^\top \}/n}\right).
	\end{split}
	\end{equation*}
	
	For the first term, observe that 
	\begin{align} \label{eq:differenceRB1}
	    \begin{split}
            \| \mathbb{S} (\hat{\Psi}_{nK}^{-1} - \Psi_{nK}^{-1}) \mathbf{R}_K^\top \mathbf{B}_K / n \| 
	        & \le \| \hat{\Psi}_{nK}^{-1} - \Psi_{nK}^{-1} \|_2 \| \mathbf{R}_K^\top \mathbf{B}_K / n \|  = O_P(\zeta_1(K) \sqrt{K} / n),
	    \end{split}
	\end{align} 
	by result (ii) and Markov's inequality.
	
	For the second term, observe that $\| \mathbb{S} \hat \Psi_{nK}^{-1} (\hat{\mathbf{R}}_K - \mathbf{R}_K)^\top \mathbf{B}_K / n \| \le O_P(1/n) \cdot \| (\hat{\mathbf{R}}_K - \mathbf{R}_K)^\top \mathbf{B}_K \|$ and
	\begin{align*}
		\underset{d_{SXK}\times 1}{(\hat{\mathbf{R}}_K - \mathbf{R}_K)^\top \mathbf{B}_K} =
		\begin{pmatrix}
			\sum_{i=1}^{n} (\hat \pi_{n, 1} \hat P_{1i} - \pi_1 P_{1i}) X_i B_{i, K} \\
			\vdots \\
			\sum_{i=1}^{n} (\hat \pi_{n, S} \hat P_{Si} - \pi_S P_{Si}) X_i B_{i, K} \\
			\sum_{i=1}^{n} [\hat \pi_{n, 1} b_K(\hat P_{1i}) - \pi_1 b_K(P_{1i})] B_{i, K} \\
			\vdots\\
			\sum_{i=1}^{n} [\hat \pi_{n, S} b_K(\hat P_{Si}) - \pi_S b_K(P_{Si})] B_{i, K} \\
		\end{pmatrix}
		.
	\end{align*}
	For each element of the right-hand side, applying the Taylor expansion to $\hat P_{ji} = F_j(Z_{ji}^\top \hat \gamma_{n,j})$ around $\gamma_j$ yields
	\begin{align*}
		& \sum_{i = 1}^n (\hat \pi_{n, j} \hat P_{ji} - \pi_j P_{ji})X_i B_{i, K}\\
		& \qquad = \sum_{i = 1}^n (\hat \pi_{n, j} - \pi_j) \hat P_{ji}X_i B_{i, K} + \sum_{i = 1}^n  \pi_j ( \hat P_{ji} - P_{ji})X_i B_{i, K}\\
		& \qquad = (\hat \pi_{n, j} - \pi_j) \sum_{i = 1}^n P_{ji} X_i B_{i, K} + \sum_{h = 1}^{\dim(Z_j)} (\hat \gamma_{n, jh} - \gamma_{jh}) \sum_{i = 1}^n  Z_{jhi} \pi_j f_j(Z_{ji}^\top \gamma_j) X_i B_{i, K} + O_P(1),
	\end{align*}
	and similarly
	\begin{align*}
		&\sum_{i = 1}^n [\hat \pi_{n, j} b_K(\hat P_{ji}) - \pi_j b_K(P_{ji})] B_{i, K} \\
		& \qquad = \sum_{i = 1}^n (\hat \pi_{n, j} - \pi_j) b_K(\hat P_{ji}) B_{i, K} + \sum_{i = 1}^n  \pi_j ( b_K(\hat P_{ji}) - b_K(P_{ji})) B_{i, K}\\
		&\qquad = (\hat \pi_{n, j} - \pi_j) \sum_{i = 1}^n b_K(P_{ji}) B_{i, K} + \sum_{h = 1}^{\dim(Z_j)} (\hat \gamma_{n, jh} - \gamma_{jh}) \sum_{i=1}^{n} Z_{jhi} \pi_j f_j(Z_{ji}^\top \gamma_j) \nabla b_K(P_{ji}) B_{i, K} + O_P(\zeta_2(K)).
	\end{align*}
	For expositional simplicity, assume that $\dim(Z_j) = 1$ for all $j$.
	Let $\hat{\bm{\Pi}}_{nK}$ and $\hat{\bm{\Gamma}}_{nK}$ be appropriate $d_{SXK} \times d_{SXK}$ diagonal matrices with diagonal elements $(\hat \pi_{n, j} - \pi_j)$ and $(\hat \gamma_{n, j} - \gamma_j)$, respectively, so that we can write
	\begin{align}\label{eq:taylor1}
		(\hat{\mathbf{R}}_K - \mathbf{R}_K)^\top \mathbf{B}_K 
		 = \sqrt{n}\hat{\bm{\Pi}}_{nK} \sum_{i=1}^{n} M_{i,K} + \sqrt{n} \hat{\bm{\Gamma}}_{nK} \sum_{i = 1}^n N_{i, K}
		+ O_P( \zeta_2(K) ),
	\end{align}
	where
	\begin{align*}
		M_{i, K} \coloneqq n^{-1/2}
		\begin{pmatrix}
			P_{1i} X_i B_{i, K}\\
			\vdots \\
			P_{Si} X_i B_{i, K}\\
			b_K(P_{1i}) B_{i, K}\\
			\vdots \\
			b_K(P_{Si}) B_{i, K}
		\end{pmatrix},
		\quad \text{and} \;\;
		N_{i, K} \coloneqq n^{-1/2}
		\begin{pmatrix}
			Z_{1i} \pi_1 f_1(Z_{1i}\gamma_1) X_i B_{i,K}\\
			\vdots \\
			Z_{Si} \pi_S f_S(Z_{Si} \gamma_S) X_i B_{i,K}\\
			Z_{1i} \pi_1 f_1(Z_{1i} \gamma_1) \nabla b_K(P_{1i}) B_{i,K}\\
			\vdots \\
			Z_{Si} \pi_S f_S(Z_{Si} \gamma_S) \nabla b_K(P_{Si}) B_{i,K}
		\end{pmatrix}.
	\end{align*}
	Note that $\bE [N_{i, K}] = \mathbf{0}_{d_{SXK}}$ by $\bE [B_{i, K} | X_i, \mathbf{Z}_i] = 0$ and $\bar N_K \coloneqq \max_{1 \le i \le n} \| N_{i, K} \| = O(\zeta_1(K) / \sqrt{n})$ under Assumptions \ref{as:parametricML}(i),(ii), \ref{as:outcome}(i), and \ref{as:series}.
	\begin{comment}
	Then, by the matrix Bernstein inequality of Theorem 1.6 in \citet{tropp2012user}, it holds that for all $t \ge 0$
	\begin{align*}
		\Pr\left( \left\| \sum_{i=1}^n N_{i, K} \right\|_2 \ge t \right) \le \exp\left( \log(d_{SXK} + 1) + \frac{-t^2/2}{\sigma_{nK}^2 + \bar N_K  t  / 3} \right),
	\end{align*}
	\end{comment}
	Further, let $\sigma_{nK}^2 \coloneqq \max \{ \| \sum_{i=1}^{n} \bE(N_{i, K} N_{i, K}^\top) \|_2, \| \sum_{i=1}^{n} \bE(N_{i, K}^\top N_{i, K}) \|_2 \}$.
	It is easy to see that $\sigma_{nK}^2 = O(\zeta_1^2(K))$.
	Observe that $\bar N_K \sqrt{\log(d_{SXK} + 1)} = O(\zeta_1(K) \sqrt{(\log K) / n}) = o(\sigma_{nK})$.
	Then, by Corollary 4.1 in \citet{chen2015optimal}, we obtain $\| \sum_{i=1}^{n} N_{i, K} \|_2 = O_P( \zeta_1(K) \sqrt{\log K} )$.
	In a similar manner, we can show that $\| \sum_{i=1}^{n} M_{i, K} \|_2 = O_P( \zeta_0(K) \sqrt{\log K} )$.
	Noting that $||\hat{\bm{\Pi}}_{nK}||_2 = O_P(n^{-1/2})$ and $|| \hat{\bm{\Gamma}}_{nK} ||_2 = O_P(n^{-1/2})$, we can show that the first and second terms on the right-hand side of \eqref{eq:taylor1} are $O_P( \zeta_0(K) \sqrt{ \log K} )$ and $O_P( \zeta_1(K) \sqrt{\log K} )$, respectively.
	Thus, 
	\begin{align} \label{eq:differenceRB2}
	    \mathbb{S} \hat \Psi_{nK}^{-1} (\hat{\mathbf{R}}_K - \mathbf{R}_K)^\top \mathbf{B}_K / n = O_P( \zeta_1(K) \sqrt{\log K} / n ) + O_P( \zeta_2(K) / n).
	\end{align}
	
	Summarizing these results, we obtain 
	\begin{equation*}
	\begin{split}
		\mathbb{S} \hat{\Psi}_{nK}^{-1} \hat{\mathbf{R}}_K^\top \mathbf{B}_K / n
		&= O_P(\zeta_1(K) \sqrt{K} / n) +  O_P( \zeta_1(K) \sqrt{\log K} / n ) + O_P( \zeta_2(K) / n) + O_P\left(\sqrt{\text{tr}\{\mathbb{S} \mathbb{S}^\top \}/n}\right).
	\end{split}
	\end{equation*}
	\bigskip

	(v), (vi) The proofs are similar to the proof of Lemma A.2 in \citet{hoshino2021treatment}.
	For (vi), note that $\bE [e^{(1)} | D, X, \mathbf{Z}] = \bE [ \bE [e^{(1)} | D, X, \mathbf{Z}, s] | D, X, \mathbf{Z} ] = 0$ since
	\begin{align*}
    	\bE [e^{(1)} | D, X, \mathbf{Z}, s = j] 
    	&= \bE \left[ \sum_{h = 1}^S \delta_h^{(1)} \left( \epsilon^{(1)} - g_h(P_h) / P_h  \right) \biggr| D, X, \mathbf{Z}, s = j \right]\\
    	&= \bE \left[ D \left( \epsilon^{(1)} - g_j(P_j) / P_j  \right) \biggr| D, X, \mathbf{Z}, s = j \right] = 0 \quad \text{for all $j$}.
	\end{align*}
	
	\begin{comment}
	\begin{align*}
		& D = 0 \Longrightarrow \bE \left[ D \left( \epsilon^{(1)} - g_j(P_j) / P_j  \right) \biggr| D = 0, X, \mathbf{Z}, s = j \right] = 0 \\
		& D = 1 \Longrightarrow \bE \left[ D \left( \epsilon^{(1)} - g_j(P_j) / P_j  \right) \biggr| D = 1, X, \mathbf{Z}, s = j \right]  = \bE \left[\epsilon^{(1)} \biggr| X, \mathbf{Z}, s = j, V_j \le P_{ij} \right]  - g_j(P_j) / P_j  = 0.
	\end{align*}
	\end{comment}
\end{proof}

Here, for a generic random variable $T$ and $q \in \mathcal{C}(\text{supp}[T])$, we define
\begin{align*}
    \hat{\mathcal{P}}_{n,K,j}^{(d)}q &\coloneqq b_K(\cdot)^\top \mathbb{S}_{K, j} \left[ \hat{\Psi}_{nK}^{(d)} \right]^{-1} \frac{1}{n} \sum_{i = 1}^n \hat R_{i, K}^{(d)}q(T_i).
\end{align*}

\begin{lemma} \label{lem:opnorm}
    Suppose that Assumptions \ref{as:iid} -- \ref{as:error} hold.
    If $\zeta_0(K) \zeta_1(K) / \sqrt{n} = O(1)$ holds, then
    \begin{align*}
        \| \hat{\mathcal{P}}_{n,K,j}^{(1)} \|_{\infty} = \| \mathcal{P}_{n,K,j}^{(1)} \|_{\infty} + O_P(1),
    \end{align*}
    where
    \begin{align*}
    \| \hat{\mathcal{P}}_{n,K,j}^{(d)} \|_{\infty} \coloneqq \sup \left\{ \sup_{p \in [0, 1]} \left| \left( \hat{\mathcal{P}}_{n,K,j}^{(d)}q\right)(p) \right| : q \in \mathcal{C}(\text{supp}[T]), \sup_{t \in \text{supp}[T]} |q(t)| = 1 \right\}.
    \end{align*}
\end{lemma}

\begin{proof}
	The proof is the same as that of Lemma A.3 in \cite{hoshino2021treatment}.
\end{proof}

\begin{lemma}\label{lem:MTRnormal}
	Suppose that Assumptions \ref{as:iid} -- \ref{as:error} hold.
	For a given $p \in \text{supp}[P_j | D = 1]$, if $\| \nabla b_K(p) \| \to \infty$, $\sqrt{n} K^{-\mu_0} \to 0$, and $\sqrt{n} K^{-\mu_1} / \| \nabla b_K(p) \| \to 0$ hold, then
	\begin{align*}
		\text{(i)} &\quad \frac{\sqrt{n} \left( \tilde m^{(1)}_j(x, p) - m^{(1)}_j (x, p) \right)}{\sigma_{K,j}^{(1)}(p)} \overset{d}{\to} N(0, 1).
	\end{align*}
	If Assumption \ref{as:deriv}, $\zeta_0(K) \zeta_1(K) / \sqrt{n} = O(1)$, and $(\| \mathcal{P}_{n,K,j}^{(1)}\|_{\infty} + 1) \sqrt{K} / \| \nabla b_K(p) \| \to 0$ hold additionally, then
	\begin{align*}
		\text{(ii)} &\quad \frac{\sqrt{n} \left( \hat m^{(1)}_j(x, p) - m^{(1)}_j (x, p) \right)}{\sigma_{K,j}^{(1)}(p)} \overset{d}{\to} N(0, 1).
	\end{align*}
\end{lemma}

\begin{proof}
	(i) First, by Assumption \ref{as:eigen}, we have
	\begin{align}\label{eq:sigmalbound}
	\begin{split}
	\sigma_{K,j}^2(p)
	= \nabla b_K(p)^\top  \mathbb{S}_{K,j} \Psi_K^{-1} \Sigma_K \Psi_K^{-1} \mathbb{S}_{K,j}^\top \nabla b_K(p)
	\ge \frac{\underbar{c}_\Sigma}{\bar{c}_\Psi^2} \cdot \| \nabla b_K(p) \|^2 > 0.
	\end{split}
	\end{align}
	
	Next, by Lemmas \ref{lem:matLLN}(iii)-(v), we have
	\begin{align*}
	    \left\| \tilde \beta_{n,j}^{(1)} - \beta_j^{(1)} \right\|
	   & \le \left\| \mathbb{S}_{X,j} \left[ \Psi_{nK}^{(1)} \right]^{-1} \mathbf{R}_K^{(1)\top}\mathbf{B}_K^{(1)} / n \right\| + \left\| \mathbb{S}_{X,j} \left[ \Psi_{nK}^{(1)} \right]^{-1} \mathbf{R}_K^{(1)\top}\mathbf{r}_K^{(1)} / n \right\| + \left\| \mathbb{S}_{X,j} \left[ \Psi_{nK}^{(1)} \right]^{-1} \mathbf{R}_K^{(1)\top}\mathbf{e}^{(1)} / n \right\|  \\
	    & = O_P(n^{-1/2}) + O_P(K^{-\mu_0}).
	\end{align*}
	Thus, by the definition of the infeasible estimator $\tilde m^{(1)}_j(x, p)$ and Assumption \ref{as:series}(i),
	\begin{align*}
	\tilde m^{(1)}_j(x, p)  -  m^{(1)}_j(x, p)
	& = x^\top \left( \tilde \beta_{n,j}^{(1)} - \beta_j^{(1)}\right) + \nabla b_K(p)^\top \tilde \alpha_{n,j}^{(1)} - \nabla g_j^{(1)}(p)\\
	& = \nabla b_K(p)^\top \left( \tilde \alpha_{n,j}^{(1)} - \alpha_j^{(1)} \right) + O_P(n^{-1/2}) + O_P(K^{-\mu_0}) + O(K^{-\mu_1}) \\
	& = A_{1n,j} + A_{2n,j} + O_P(n^{-1/2}) + O_P(K^{-\mu_0}) + O(K^{-\mu_1}),
	\end{align*}
	where $A_{1n,j} \coloneqq \nabla b_K(p)^\top \mathbb{S}_{K,j} \Psi_{nK}^{-1} \mathbf{R}_K^\top \bm{\xi}_K / n$ and $A_{2n,j} \coloneqq \nabla b_K(p)^\top \mathbb{S}_{K,j} \Psi_{nK}^{-1} \mathbf{R}_K^\top \mathbf{r}_K / n$.
	For $A_{2n,j}$, by Lemma \ref{lem:matLLN}(iv) and $\sqrt{n}K^{-\mu_0} \to 0$, we have
	\begin{align*}
	    |A_{2n,j}|
	    \le \| \nabla b_K(p) \| \cdot \| \mathbb{S}_{K,j} \Psi_{nK}^{-1} \mathbf{R}_K^\top \mathbf{r}_K / n \|
	    = \| \nabla b_K(p) \| \cdot O_P(K^{-\mu_0})
	    = \| \nabla b_K(p) \| \cdot o_P(n^{-1/2}).
	\end{align*}
	Define $A'_{1n,j} \coloneqq \nabla b_K(p)^\top  \mathbb{S}_{K,j} \Psi_K^{-1} \mathbf{R}_K^\top \bm{\xi}_K / n$.
	It is easy to see that 
	\begin{align*}
	    | A_{1n,j} - A'_{1n,j} | 
	    & \le \| \nabla b_K(p) \| \cdot \| \mathbb{S}_{K,j} ( \Psi_{nK}^{-1} - \Psi_K^{-1} ) \mathbf{R}_K^\top \bm{\xi}_K / n \| \\
	    & = \| \nabla b_K(p) \| \cdot O_P(\zeta_0(K) \sqrt{(K \log K)} / n) = \| \nabla b_K(p) \| \cdot  o_P(n^{-1/2}),
	\end{align*}
	by Lemma \ref{lem:matLLN}(ii), Markov's inequality, and Assumption \ref{as:series}(ii).
	Thus, by \eqref{eq:sigmalbound}, we obtain
	\begin{align}\label{eq:MTRlinear}
		\frac{\sqrt{n} \left( \tilde m^{(1)}_j(x, p) - m^{(1)}_j (x, p) \right)}{\sigma_{K,j}^{(1)}(p)} 
		=  \frac{\sqrt{n} ( A_{1n,j} + A_{2n,j})}{\sigma_{K,j}^{(1)}(p)} + o_P(1)
		=  \frac{\sqrt{n} A'_{1n,j}}{\sigma_{K,j}^{(1)}(p)} + o_P(1),
	\end{align}
	since we have assumed $\| \nabla b_K(p) \| \to \infty$, $\sqrt{n} K^{-\mu_0} \to 0$, and $\sqrt{n} K^{-\mu_1} / \| \nabla b_K(p) \| \to 0$.
	
	We now show the asymptotic normality of $\sqrt{n} A'_{1n,j} / \sigma_{K,j}^{(1)}(p)$.
	Let $\phi_{ji} \coloneqq \Pi_{K,j}(p) R_{i,K}^{(1)} \xi_{i,K}^{(1)} / \sqrt{n}$, where $\Pi_{K,j}(p) \coloneqq \nabla b_K(p)^\top  \mathbb{S}_{K,j} \Psi_K^{-1}  / \sigma_{K,j}^{(1)}(p)$, so that $\sum_{i = 1}^n \phi_{ji} =  \sqrt{n} A'_{1n,j} / \sigma_{K,j}^{(1)}(p)$.
	Since $\bE [B_K^{(1)}|X, \mathbf Z] = 0$ and $\bE [e^{(1)} | X, \mathbf Z] = 0$ as shown in \eqref{eq:errormean1}, we have $\bE [\phi_{ji}] = 0$ and $Var[\phi_{ji}] = n^{-1}$.
	Moreover, note that $\bE [(\xi_{i,K}^{(1)})^4 | X_i, \mathbf{Z}_i ] = O(1)$ holds by the $c_r$-inequality with Assumption \ref{as:error} and the uniform boundedness of $B_{i, K}^{(1)}$.
	Then, by the same argument as in the proof of Theorem 4.2 in \citet{hoshino2021treatment}, we obtain $\sum_{i = 1}^n \bE [\phi_{ji}^4] = O(\zeta_0^2(K) K / n ) = o(1)$ under Assumption \ref{as:series}(ii).
	Hence, result (i) follows from Lyapunov's central limit theorem.	

	\bigskip
	
	(ii) By Lemmas \ref{lem:estmatLLN}(iii) -- (vi), Assumption \ref{as:series}(ii), and $\sqrt{n} K^{-\mu_0} \to 0$, we have
	\begin{align*}
		\left\| \hat \beta_{n,j}^{(1)} - \beta_j^{(1)} \right\|
		& \le \left\| \mathbb{S}_{X,j} \left[ \hat{\Psi}_{nK}^{(1)} \right]^{-1} \hat{\mathbf{R}}_K^{(1)\top}\hat{\bm{\Delta}}_K^{(1)} / n \right\| 
		+ \left\| \mathbb{S}_{X,j} \left[ \hat{\Psi}_{nK}^{(1)} \right]^{-1} \hat{\mathbf{R}}_K^{(1)\top}\mathbf{r}_K^{(1)} / n \right\| 
		+ \left\| \mathbb{S}_{X,j} \left[ \hat{\Psi}_{nK}^{(1)} \right]^{-1} \hat{\mathbf{R}}_K^{(1)\top}\bm{\xi}_K^{(1)} / n \right\|  \\
	    & = O_P(n^{-1/2}) + \underbrace{ O_P(\zeta_1(K)\sqrt{K}/n) + O_P(\zeta_2(K)/n) + O_P(K^{-\mu_0}) }_{= \: o_P(n^{-1/2})}.
	\end{align*}
	Thus, by the definition of the feasible estimator $\hat m_j^{(1)}(x, p)$ and Assumption \ref{as:series}(i),
	\begin{align*}
		\hat m_j^{(1)}(x, p) - m_j^{(1)}(x, p)
		&= x^\top (\hat \beta_{n,j}^{(1)} - \beta_j^{(1)}) + \nabla b_K(p)^\top \hat \alpha_{n,j}^{(1)} - \nabla g_j^{(1)}(p)\\
		&= \nabla b_K(p)^\top (\hat \alpha_{n,j}^{(1)} - \alpha_j^{(1)}) + O_P(n^{-1/2}) + O_P(K^{-\mu_1})\\
		&= \mathfrak{A}_{1n, j} + \mathfrak{A}_{2n, j} + \mathfrak{A}_{3n, j} + O_P(n^{-1/2}) + O_P(K^{-\mu_1}),
	\end{align*}
	where
	\begin{align*}
		& \mathfrak{A}_{1n, j} \coloneqq \nabla b_K(p)^\top \mathbb{S}_{K,j} \hat{\Psi}_{nK}^{-1} \hat{\mathbf{R}}_K^\top \bm{\xi}_K / n,
		\qquad 
		\mathfrak{A}_{2n, j} \coloneqq \nabla b_K(p)^\top \mathbb{S}_{K,j} \hat{\Psi}_{nK}^{-1} \hat{\mathbf{R}}_K^\top \mathbf{r}_K / n,\\
		& \mathfrak{A}_{3n, j} \coloneqq \nabla b_K(p)^\top \mathbb{S}_{K,j} \hat{\Psi}_{nK}^{-1} \hat{\mathbf{R}}_K^\top \hat{\bm{\Delta}}_K / n.
	\end{align*}
	For $\mathfrak{A}_{1n, j}$, observe that
	\begin{align*}
		\mathfrak{A}_{1n, j}
		& = A_{1n, j} + \nabla b_K(p)^\top \mathbb{S}_{K,j} \left( \hat{\Psi}_{nK}^{-1} - \Psi_{nK}^{-1} \right) \mathbf{R}_K^\top \bm{\xi}_K / n + \nabla b_K(p)^\top \mathbb{S}_{K,j} \hat{\Psi}_{nK}^{-1} \left( \hat{\mathbf{R}}_K - \mathbf{R}_K \right)^\top \mathbf{B}_K / n \\
		& \quad +  \nabla b_K(p)^\top \mathbb{S}_{K,j} \hat{\Psi}_{nK}^{-1} \left( \hat{\mathbf{R}}_K - \mathbf{R}_K \right)^\top \mathbf{e} / n.
	\end{align*}
	By the same argument as in \eqref{eq:differenceRB1}, the second term on the right-hand side satisfies
	\begin{align*}
	    \left\| \nabla b_K(p)^\top \mathbb{S}_{K,j} \left( \hat{\Psi}_{nK}^{-1} - \Psi_{nK}^{-1} \right) \mathbf{R}_K^\top \bm{\xi}_K / n \right\| 
	    & \le \| \nabla b_K(p) \| \cdot \left\| \mathbb{S}_{K,j} \left( \hat{\Psi}_{nK}^{-1} - \Psi_{nK}^{-1} \right) \mathbf{R}_K^\top \bm{\xi}_K / n \right\| \\
	    & = \| \nabla b_K(p) \| \cdot O_P( \zeta_1(K) \sqrt{K} / n).
	\end{align*}
	Similarly, we can show that the third term is of order $\| \nabla b_K(p) \|\cdot\left\{ O_P( \zeta_1(K) \sqrt{\log K} / n) + O_P(\zeta_2(K) / n)\right\}$ by \eqref{eq:differenceRB2}.
	For the fourth term, recalling that $\bE [e^{(1)} | D, X, \mathbf{Z}]  = 0$, we have
	\begin{align*}
	& \bE \left[ \left\| \mathbb{S}_{K,j} \hat{\Psi}_{nK}^{-1} \left( \hat{\mathbf{R}}_K - \mathbf{R}_K \right)^\top \mathbf{e} / n \right\|^2 \Bigg| \{D_i , X_i, \mathbf Z_i \}_{i = 1}^n \right]\\
	& \quad = \text{tr} \left\{ \mathbb{S}_{K,j} \hat{\Psi}_{nK}^{-1} \left( \hat{\mathbf{R}}_K - \mathbf{R}_K \right)^\top  \bE [ \mathbf{e} \mathbf{e}^\top | \{D_i, X_i, \mathbf Z_i \}_{i = 1}^n ] \left( \hat{\mathbf{R}}_K - \mathbf{R}_K \right) \hat{\Psi}_{nK}^{-1} \mathbb{S}_{K,j}^\top \right\} / n^2 \\
	& \quad \le O(1/n^2) \cdot || \hat{\mathbf{R}}_K - \mathbf{R}_K ||^2 \cdot \text{tr} \{\mathbb{S}_{K,j} \hat{\Psi}_{nK}^{-1} \hat{\Psi}_{nK}^{-1} \mathbb{S}_{K,j}^\top \}
	= O_P( \zeta^2_1(K) K / n^2)
	\end{align*}
	by Assumption \ref{as:error} and \eqref{eq:differenceR}.
	Thus, by Markov's inequality,  we find that the fourth term is of order $ \| \nabla b_K(p) \| \cdot O_P( \zeta_1(K) \sqrt{K} / n)$.
	Combining these results yields 
	\begin{align*}
	    \mathfrak{A}_{1n, j} 
	    = A_{1n, j} + \| \nabla b_K(p) \| \cdot \left\{ O_P(\zeta_1(K) \sqrt{K} / n) + O_P( \zeta_2(K) / n)\right\}
	    = A_{1n,j} + \| \nabla b_K(p) \| \cdot o_P(n^{-1/2})
	\end{align*}
	under Assumption \ref{as:series}(ii).

	For $\mathfrak{A}_{2n, j}$, by Lemma \ref{lem:estmatLLN}(v) and $\sqrt{n} K^{-\mu_0} \to 0$, we have  $|\mathfrak{A}_{2n, j}| \le \| \nabla b_K(p) \| \cdot \| \mathbb{S}_{K,j} \hat{\Psi}_{nK}^{-1} \hat{\mathbf{R}}_K^\top \mathbf{r}_K / n \| = \| \nabla b_K(p) \| \cdot o_P(n^{-1/2})$.
	
	For $\mathfrak{A}_{3n, j}$, observe that $|\mathfrak{A}_{3n, j}| \le O(\sqrt{K}) \cdot \sup_{p \in [0,1]} |b_K(p)^\top \mathbb{S}_{K,j} \hat{\Psi}_{nK}^{-1} \hat{\mathbf{R}}_K^\top \hat{\bm{\Delta}}_K / n|$ by Assumption \ref{as:deriv}.
	Further, noting that 
	\begin{align*}
	\hat \Delta_K = (R_K - \hat R_K)^\top \theta^{(1)} 
	&= \sum_{h=1}^S \left[ \pi_h (P_h - \hat P_h) X^\top \beta_h^{(1)} + ( \pi_h - \hat \pi_{n, h} ) \hat P_h X^\top \beta_h^{(1)}\right.\\
	&\qquad\quad + \left. \pi_h (b_K(P_h) - b_K(\hat P_h))^\top \alpha_h^{(1)} + (\pi_h - \hat \pi_{n, h}) b_K(\hat P_h)^\top \alpha_h^{(1)}\right],    
	\end{align*}
	write
	\begin{align*}
		b_K(p)^\top \mathbb{S}_{K,j} \hat{\Psi}_{nK}^{-1} \hat{\mathbf{R}}_K^\top \hat{\bm{\Delta}}_K / n
		&= b_K(p)^\top \mathbb{S}_{K,j} \hat{\Psi}_{nK}^{-1} \left[ \frac{1}{n} \sum_{i = 1}^n \hat R_{i, K} \left( \sum_{h=1}^S \pi_h (P_{hi} - \hat P_{hi}) X_i^\top \beta_h^{(1)} \right) \right]\\
		& \qquad + b_K(p)^\top \mathbb{S}_{K,j} \hat{\Psi}_{nK}^{-1} \left[ \frac{1}{n} \sum_{i = 1}^n \hat R_{i, K} \left( \sum_{h=1}^S (\pi_h - \hat \pi_{n, h}) \hat P_{hi} X_i^\top \beta_h^{(1)} \right) \right]\\
		& \qquad + b_K(p)^\top \mathbb{S}_{K,j} \hat{\Psi}_{nK}^{-1} \left[ \frac{1}{n} \sum_{i = 1}^n \hat R_{i, K} \left( \sum_{h=1}^S \pi_h \left( b_K(P_{hi}) - b_K(\hat P_{hi}) \right)^\top \alpha_h^{(1)} \right) \right]\\
		& \qquad + b_K(p)^\top \mathbb{S}_{K,j} \hat{\Psi}_{nK}^{-1} \left[ \frac{1}{n} \sum_{i = 1}^n \hat R_{i, K} \left( \sum_{h=1}^S (\pi_h - \hat \pi_{n, h}) b_K(\hat P_{hi})^\top \alpha_h^{(1)} \right) \right]\\
		&\eqqcolon \mathfrak{B}_{1n,j}(p) + \mathfrak{B}_{2n,j}(p) + \mathfrak{B}_{3n,j}(p) + \mathfrak{B}_{4n,j}(p), \text{ say}.
	\end{align*}
	By Lemma \ref{lem:opnorm}, if $\zeta_0(K) \zeta_1(K) / \sqrt{n} = O(1)$,
	\begin{align*}
	    |\mathfrak{B}_{1n, j}(p)|
	    & = \left| b_K(p)^\top \mathbb{S}_{K,j} \hat{\Psi}_{nK}^{-1} \frac{1}{n} \sum_{i = 1}^n \hat R_{i, K}q_n(\mathbf{Z}_i,X_i) \right|\\
	    & = \left| \left( \hat{\mathcal{P}}_{n,K,j}^{(1)} q_n \right)(p) \right| \le \| \hat{\mathcal{P}}_{n,K,j}^{(1)}\|_{\infty} \cdot O_P(n^{-1/2}) = \left(\| \mathcal{P}_{n,K,j}^{(1)}\|_{\infty} + 1\right) \cdot O_P(n^{-1/2})
	\end{align*}
	for any $p \in [0,1]$, where the definition of $q_n(\mathbf{Z}_i,X_i)$ should be clear from the context. 
	Similarly, we can easily show that $|\mathfrak{B}_{2n, j}(p)|$, $|\mathfrak{B}_{3n, j}(p)|$, and $|\mathfrak{B}_{4n, j}(p)|$ are also of order $\left(\| \mathcal{P}_{n,K,j}^{(1)}\|_{\infty} + 1\right) \cdot O_P(n^{-1/2})$ uniformly in $p \in [0,1]$.
	Consequently, we have 
	\begin{align*}
	    \mathfrak{A}_{3n, j} = \left( \| \mathcal{P}_{n,K,j}^{(1)}\|_{\infty} + 1 \right) \cdot O_P(\sqrt{K / n}) = \| \nabla b_K(p) \| \cdot o_P(n^{-1/2})
	\end{align*}
	since we have assumed $(\| \mathcal{P}_{n,K,j}^{(1)}\|_{\infty} + 1) \sqrt{K} / \| \nabla b_K(p) \| \to 0$.
	Therefore, we have 
	\begin{align}\label{eq:MTRlinear2}
	\frac{\sqrt{n}  \left( \hat m_j^{(1)}(x, p) - m_j^{(1)}(x, p) \right) }{ \sigma_{K,j}^{(1)}(p)} 
	= \frac{\sqrt{n} A_{1n, j}}{\sigma_{K,j}^{(1)}(p)} + o_P(1),
	\end{align}
	and the result follows from the proof of result (i).
\end{proof}

\section{Appendix: Supplementary Technical Results}\label{sec:suppident}

\subsection{Identification of the Finite Mixture Probit Models}\label{subsec:mixture}

\subsubsection{Exogenous membership with constant membership probability}\label{subsubsec:exomember}

We provide an identification result for the finite mixture Probit model in which group membership is assumed to be exogenous.

\begin{assumption}\label{as:exomember}\hfil
	\begin{enumerate}[(i)]
		\item The treatment choice is generated by $D = \mathbf{1} \{ Z^\top \gamma_{zj} + \zeta_j \gamma_{\zeta j} \ge \epsilon_j^D \}$ if $s = j$, where $Z \in \mathbb{R}^{\dim(Z)}$ is a vector of common covariates among all groups, $\zeta_j \in \mathbb{R}$ is a group-specific continuous IV, $\epsilon_j^D \sim N(0, 1)$ independently of $(Z, \zeta)$ with denoting $\zeta = (\zeta_1, \dots, \zeta_S)^\top$, and $\gamma_{\zeta j} \neq 0$ for all $j \in \{ 1, \dots, S \}$.
		\item $s$ is independent of $(Z, \zeta, \epsilon_j^D)$ and $\pi_j = \Pr(s = j) > 0$ for all $j$.
		\item For each $z \in \text{supp}[Z]$, there exist $\mathrm{x}, \mathrm{x}', \mathrm{x}'', \mathrm{x}''' \in \text{supp}[\zeta \mid Z = z]$ such that $|\mathrm{x}_j| \neq |\mathrm{x}_j'|$, $|\mathrm{x}_j''| \neq |\mathrm{x}_j'''|$, and $\mathrm{x}_j + \mathrm{x}_j' \neq \mathrm{x}_j'' + \mathrm{x}_j'''$ for all $j$.
		\item $Z$ includes a constant and does not lie in a proper linear subspace of $\mathbb{R}^{\dim(Z)-1}$ a.s.
	\end{enumerate}
\end{assumption}

In Assumption \ref{as:exomember}(i), we formalize the finite mixture Probit model and require the existence of a continuous IV specific to each group.
Assumption \ref{as:exomember}(ii) restricts group membership to be exogenous. 
Assumption \ref{as:exomember}(iii) is a mild restriction on the support of the group-specific IV, and Assumption \ref{as:exomember}(iv) is a standard rank condition.

\begin{theorem}\label{thm:exomember}
	Under Assumption \ref{as:exomember}, the coefficients $\gamma_z = (\gamma_{z1}^\top, \dots, \gamma_{zS}^\top)^\top$ and $\gamma_{\zeta} = (\gamma_{\zeta 1}, \dots, \gamma_{\zeta S})^\top$ and the membership probabilities $\pi = (\pi_1, \dots, \pi_S)^\top$ are identified.
\end{theorem}

\begin{proof}
	%Denote a realization of $\zeta$ as $\mathrm{x} = (\mathrm{x}_1, \dots, \mathrm{x}_S)^\top$.
	By Assumptions \ref{as:exomember}(i)--(ii), it holds that $\Pr(D = 1 \mid Z = z, \zeta = \mathrm{x}) = \sum_{j = 1}^S \pi_j \Phi(z^\top \gamma_{zj} + \mathrm{x}_j \gamma_{\zeta j})$ .
	Then, for each $j$,
	\begin{align}\label{eq:derivativeProbit}
		\frac{\partial}{\partial \mathrm{x}_j} \Pr(D = 1 \mid Z = z, \zeta = \mathrm{x}) = \pi_j \phi(z^\top \gamma_{zj} + \mathrm{x}_j \gamma_{\zeta j}) \gamma_{\zeta j},
	\end{align}
	where $\phi$ denotes the standard normal density. 
	Hence, for another realization $\mathrm{x}' = (\mathrm{x}_1', \dots, \mathrm{x}_S')^\top$ of $\zeta$ such that $|\mathrm{x}_j| \neq |\mathrm{x}_j'|$, we have
	\begin{align*}
		\frac{\partial \Pr(D = 1 \mid Z = z, \zeta = \mathrm{x}) / \partial \mathrm{x}_j }{\partial \Pr(D = 1 \mid Z = z, \zeta = \mathrm{x}') / \partial \mathrm{x}_j' } 
		&= \frac{\phi(z^\top \gamma_{zj} + \mathrm{x}_j \gamma_{\zeta j})}{\phi(z^\top \gamma_{zj} + \mathrm{x}_j' \gamma_{\zeta j})}\\
		&= \exp \left( \frac{1}{2} \left[ (z^\top \gamma_{zj} + \mathrm{x}_j' \gamma_{\zeta j})^2 - (z^\top \gamma_{zj} + \mathrm{x}_j \gamma_{\zeta j})^2 \right]\right)\\
		&= \exp \left( \frac{1}{2} \left[ \left[ (\mathrm{x}_j')^2 - \mathrm{x}_j^2 \right]  \gamma_{\zeta j}^2 + 2(\mathrm{x}_j' - \mathrm{x}_j) z^\top \gamma_{zj} \gamma_{\zeta j} \right] \right).
	\end{align*}
	This implies that we can obtain the following linear equations with the parameters $\gamma_{\zeta j}^2$ and $z^\top \gamma_{zj} \gamma_{\zeta j}$:
	\begin{align*}
		2 \log \left[ \frac{\partial \Pr(D = 1 \mid Z = z, \zeta = \mathrm{x}) / \partial \mathrm{x}_j }{\partial \Pr(D = 1 \mid Z = z, \zeta = \mathrm{x}') / \partial \mathrm{x}_j' } \right]  & = \left[ (\mathrm{x}_j')^2 - \mathrm{x}_j^2 \right]  \gamma_{\zeta j}^2 + 2(\mathrm{x}_j' - \mathrm{x}_j) z^\top \gamma_{zj} \gamma_{\zeta j}, \\
		2 \log \left[ \frac{\partial \Pr(D = 1 \mid Z = z, \zeta = \mathrm{x}'') / \partial \mathrm{x}_j'' }{\partial \Pr(D = 1 \mid Z = z, \zeta = \mathrm{x}''') / \partial \mathrm{x}_j''' } \right]  & = \left[ (\mathrm{x}_j''')^2 - \mathrm{x}_j''^2 \right]  \gamma_{\zeta j}^2 + 2(\mathrm{x}_j''' - \mathrm{x}_j'') z^\top \gamma_{zj} \gamma_{\zeta j}.
	\end{align*}
	Note that the left-hand side terms can be identified from data.
	Thus, $\gamma_{\zeta j}^2$ and $z^\top \gamma_{zj} \gamma_{\zeta j}$ for each $z \in \text{supp}[Z]$ are identified by solving the system of linear equations under Assumption \ref{as:exomember}(iii).
	Note that $\gamma_{\zeta j}$ is identified from this result since the sign of $\gamma_{\zeta j}$ is known from \eqref{eq:derivativeProbit} (for this, notice that $\pi_j \phi$ is positive).
	Further, by Assumption \ref{as:exomember}(iv), $\gamma_{zj} \gamma_{\zeta j}$ is identified from the identification of $z^\top \gamma_{zj} \gamma_{\zeta j}$, and so is $\gamma_{zj}$.
	Finally, $\pi_j$ is also identified from \eqref{eq:derivativeProbit}.
	The above argument holds for any $j$, implying the identification of all $\gamma_{\zeta}$, $\gamma_{z}$, and $\pi$.
\end{proof}

\subsubsection{Endogenous membership with covariate-dependent membership probability}\label{subsubsec:endmember}

In line with the setup in Subsection \ref{subsec:hetero}, we consider the following finite mixture model with potentially endogenous group membership:

\begin{assumption}\label{as:endmember}\hfil
	\begin{enumerate}[(i)]
		\item The group membership and the treatment choice are determined by 
		\begin{align*}
			s &= 2 - \mathbf{1}\{ Z^\top \alpha_z + W_1 \alpha_w \ge \epsilon^s \},\\
			D &= \mathbf{1} \{ Z^\top \gamma_{zj} + \zeta_j \gamma_{\zeta j} \ge \epsilon_j^D \} \quad \text{if $s = j$,}
		\end{align*}
		where $Z \in \mathbb{R}^{\dim(Z)}$ is a vector of common covariates among the groups, $\zeta = (\zeta_1, \zeta_2) \in \mathbb{R}^2$ are group-specific continuous IVs, $W_1 \in \mathbb{R}$ is a continuous IV which affects group membership only, the error terms $(\epsilon^s, \epsilon_j^D)$ follow the standard bivariate normal distribution with correlation parameter $\rho_j$ independently of $(Z, \zeta, W_1)$, and $\gamma_{\zeta j} \neq 0$ for all $j \in \{1, 2\}$.
		\item Conditional on $(Z, \zeta)$, $W_1$ is distributed on the whole $\mathbb{R}$.
		The sign of $\alpha_w$ is known to be positive.
		\item For each $z \in \text{supp}[Z]$ and any sufficiently large $w_1 \in \mathbb{R}$, there exist $\mathrm{x}, \mathrm{x}', \mathrm{x}'', \mathrm{x}''' \in \text{supp}[\zeta \mid Z = z, W_1 = w_1]$ such that $|\mathrm{x}_j| \neq |\mathrm{x}_j'|$, $|\mathrm{x}_j''| \neq |\mathrm{x}_j'''|$, and $\mathrm{x}_j + \mathrm{x}_j' \neq \mathrm{x}_j'' + \mathrm{x}_j'''$ for all $j$.
		\item $(Z, W_1)$ include a constant and do not lie in a proper linear subspace of $\mathbb{R}^{\dim(Z)}$ a.s.
	\end{enumerate}
\end{assumption}

Assumption \ref{as:endmember}(i) allows for the dependence between $\epsilon_s$ and $\epsilon_j^D$ so that the group membership can be endogenous, although it requires the existence of a continuous IV that affects only the group membership for handling the endogeneity.
We use Assumption \ref{as:endmember}(ii) to separately identify the correlation parameter $\rho_j$ and the coefficients $\gamma_{zj}$ and $\gamma_{\zeta j}$ based on an identification-at-infinity argument, similar to the identification strategy often employed in the game econometrics literature (cf. \citealp{tamer2003incomplete}).
Assumptions \ref{as:endmember}(iii)--(iv) are analogous to Assumptions \ref{as:exomember}(iii)--(iv).

\begin{theorem}\label{thm:endmember}
	Under Assumption \ref{as:endmember}, the coefficients $\alpha_z$, $\alpha_w$, $(\gamma_{z1}, \gamma_{z2})$, and $(\gamma_{\zeta 1}, \gamma_{\zeta 2})$ and the correlation parameters $(\rho_1, \rho_2)$ are identified.
\end{theorem}

\begin{proof}
	We first show that $\gamma_{z1}$ and $\gamma_{\zeta 1}$ are identified (the identification of $\gamma_{z 2}$ and $\gamma_{\zeta 2}$ is symmetric and thus omitted). 
	Let $p(z, \mathrm{x}, w_1) \coloneqq \Pr(D = 1 \mid Z = z, \zeta = \mathrm{x}, W_1 = w_1)$ and $t(z,w_1) \coloneqq  z^\top \alpha_z + w_1 \alpha_w$.
	Under Assumption \ref{as:endmember}(i), observe that
	\begin{align*}
		p(z, \mathrm{x}, w_1)
		& = \Pr(D = 1 \mid s = 1, Z = z, \zeta = \mathrm{x}, W_1 = w_1) \Pr(s = 1 \mid Z = z, \zeta = \mathrm{x}, W_1 = w_1) \\
		& \quad + \Pr(D = 1 \mid s = 2, Z = z, \zeta = \mathrm{x}, W_1 = w_1) \Pr(s = 2 \mid Z = z, \zeta = \mathrm{x}, W_1 = w_1) \\
		& = \Pr(\epsilon_1^D \le z^\top \gamma_{z1} + \mathrm{x}_1 \gamma_{\zeta 1} \mid \epsilon^s \le t(z,w_1) ) \Pr( \epsilon^s \le t(z,w_1) ) \\
		& \quad + \Pr(\epsilon_2^D \le z^\top \gamma_{z2} + \mathrm{x}_2 \gamma_{\zeta 2} \mid \epsilon^s > t(z,w_1)) \Pr( \epsilon^s > t(z,w_1)).
	\end{align*}
	It is easy to see that the conditional density of $\epsilon_1^D$ given $\epsilon^s \le t(z,w_1)$ is obtained by
	\begin{align*}
		f_{\epsilon_1^D}(e_1 \mid \epsilon^s \le t(z,w_1)) 
		= \frac{f_{\epsilon_1^D}(e_1)}{\Pr(\epsilon^s \le t(z,w_1))} \Pr(\epsilon^s \le t(z,w_1) \mid \epsilon_1^D = e_1)  = \frac{\phi(e_1)}{\Phi(t(z,w_1))} \Phi\left( \frac{t(z,w_1) - \rho_1 e_1}{ \sqrt{1 - \rho_1^2}} \right).
	\end{align*}
	Similarly, the conditional density of $\epsilon_2^D$ given $\epsilon^s > t(z,w_1)$ is obtained by
	\begin{align*}
		f_{\epsilon_2^D}(e_2 \mid \epsilon^s > t(z,w_1)) 
		= \frac{\phi(e_2)}{1 - \Phi(t(z,w_1))} \left( 1 - \Phi\left( \frac{t(z,w_1) - \rho_2 e_2}{ \sqrt{1 - \rho_2^2} } \right) \right).
	\end{align*}
	Thus, we have
	\begin{align*}
		p(z, \mathrm{x}, w_1)
		= \int_{-\infty}^{z^\top \gamma_{z1} + \mathrm{x}_1 \gamma_{\zeta 1}} \phi(e_1) \Phi\left( \frac{t(z,w_1) - \rho_1 e_1}{\sqrt{1 - \rho_1^2} } \right) \mathrm{d}e_1
		+ \int_{-\infty}^{z^\top \gamma_{z2} + \mathrm{x}_2 \gamma_{\zeta 2}} \phi(e_2) \left( 1 - \Phi\left( \frac{t(z,w_1) - \rho_2 e_2}{\sqrt{1 - \rho_2^2} } \right) \right) \mathrm{d}e_2.
	\end{align*}
	Taking the partial derivative with respect to $\mathrm{x}_1$ leads to
	\begin{align}\label{eq:derivativeProbit2}
		\frac{\partial}{\partial \mathrm{x}_1} p(z, \mathrm{x}, w_1)
		= \gamma_{\zeta 1} \cdot \phi(z^\top \gamma_{z1} + \mathrm{x}_1 \gamma_{\zeta 1}) \Phi\left( \frac{t(z,w_1) - \rho_1 [z^\top \gamma_{z1} + \mathrm{x}_1 \gamma_{\zeta 1}]}{\sqrt{1 - \rho_1^2}} \right).
	\end{align}
	Then, under Assumption \ref{as:endmember}(ii), we have
	\begin{align*}
		\lim_{w_1 \to \infty} \frac{\partial}{\partial \mathrm{x}_1} p(z, \mathrm{x}, w_1) = \gamma_{\zeta 1} \cdot \phi(z^\top \gamma_{z1} + \mathrm{x}_1 \gamma_{\zeta 1}).
	\end{align*}
	For another realization $\mathrm{x}'$ of $\zeta$ such that $|\mathrm{x}_1| \neq |\mathrm{x}_1'|$, we have
	\begin{align*}
		&\frac{\lim_{w_1 \to \infty} \partial p(z, \mathrm{x}, w_1) / \partial \mathrm{x}_1}{\lim_{w_1 \to \infty} \partial p(z, \mathrm{x}', w_1) / \partial \mathrm{x}_1'} 
		= \exp \left( \frac{1}{2} \left[ \left[ (\mathrm{x}_1')^2 - \mathrm{x}_1^2 \right]  \gamma_{\zeta 1}^2 + 2(\mathrm{x}_1' - \mathrm{x}_1) z^\top \gamma_{z1} \gamma_{\zeta 1} \right] \right)\\
		& \Longrightarrow 2 \log \left[ \frac{\lim_{w_1 \to \infty} \partial p(z, \mathrm{x}, w_1) / \partial \mathrm{x}_1 }{\lim_{w_1 \to \infty} \partial p(z, \mathrm{x}', w_1) / \partial \mathrm{x}_1' } \right]  = \left[ (\mathrm{x}_1')^2 - \mathrm{x}_1^2 \right]  \gamma_{\zeta 1}^2 + 2(\mathrm{x}_1' - \mathrm{x}_1) z^\top \gamma_{z1} \gamma_{\zeta 1}.
	\end{align*}
	Thus, under Assumptions \ref{as:endmember}(iii)--(iv), the same argument as in the proof of Theorem \ref{thm:exomember} gives the identification of $\gamma_{z 1}$ and $\gamma_{\zeta 1}$.
	
	To examine identification of $\rho_1$, we rearrange \eqref{eq:derivativeProbit2} as follows:
	\begin{align} \label{eq:defL}
		\begin{split}
			& \frac{1}{\gamma_{\zeta 1} \cdot \phi(z^\top \gamma_{z1} + \mathrm{x}_1 \gamma_{\zeta 1})} \left[ \frac{\partial}{\partial \mathrm{x}_1} p(z, \mathrm{x}, w_1) \right] = \Phi\left( \frac{t(z,w_1) - \rho_1 [z^\top \gamma_{z1} + \mathrm{x}_1 \gamma_{\zeta 1}]}{\sqrt{1 - \rho_1^2}} \right) \\
			& \Longrightarrow \underbrace{\Phi^{-1} \left( \frac{1}{\gamma_{\zeta 1} \cdot \phi(z^\top \gamma_{z1} + \mathrm{x}_1 \gamma_{\zeta 1})} \left[ \frac{\partial}{\partial \mathrm{x}_1} p(z, \mathrm{x}, w_1) \right] \right)}_{\eqqcolon L(z, \mathrm{x}, w_1)} = \frac{t(z,w_1) - \rho_1 [z^\top \gamma_{z1} + \mathrm{x}_1 \gamma_{\zeta 1}]}{\sqrt{1 - \rho_1^2}}.
		\end{split}
	\end{align}
	Then, under Assumption \ref{as:endmember}(ii), we have
	\begin{align*}
		L(z, \mathrm{x}, w_1) - L(z, \mathrm{x}', w_1) = \frac{\rho_1(\mathrm{x}_1' - \mathrm{x}_1) \gamma_{\zeta 1}}{\sqrt{1 - \rho_1^2}} 
		\Longrightarrow \frac{L(z, \mathrm{x}, w_1) - L(z, \mathrm{x}', w_1)}{(\mathrm{x}_1' - \mathrm{x}_1) \gamma_{\zeta 1}}  = \frac{\rho_1}{\sqrt{1 - \rho_1^2}}.
	\end{align*}
	Noting that the left-hand side is already identified, by solving the above equation, we can identify $\rho_1$.
	The identification of $\rho_2$ can be established in the same manner.
	
	Finally, to identify $\alpha_w$ and $\alpha_z$, we further rearrange \eqref{eq:defL} as follows:
	\begin{align*}
		\sqrt{1 - \rho_1^2} L(z, \mathrm{x}, w_1) + \rho_1[z^\top \gamma_{z1} + \mathrm{x}_1 \gamma_{\zeta 1}] = z^\top \alpha_z + w_1 \alpha_w.
	\end{align*}
	Noting that the left-hand side is an identified quantity, the identification of $\alpha_z$ and $\alpha_w$ is achieved by Assumption \ref{as:endmember}(iv).
\end{proof}

\subsection{Identification of PRTE}\label{subsec:prte}

Consider a counterfactual policy that changes $\mathbf{P}$ but does not affect $Y^{(d)}$, $X$, $\epsilon_j^D$, and $s$.
Let $\mathbf{P}^\star = (P_1^\star, \dots, P_S^\star)$ be a counterfactual version of $\mathbf{P}$ whose distribution is known and $D^\star$ be the treatment status under $\mathbf{P}^\star$.
As in Assumption \ref{as:IV}(i), we assume that $\mathbf{P}^\star$ is independent of $(\epsilon^{(d)}, \epsilon_j^D, s)$ given $X$.
Denote the outcome after the policy as $Y^\star$.
The group-wise PRTE is defined as 
\begin{align*}
	\text{PRTE}_j \coloneqq \bE [Y^\star | X = x, s = j] - \bE [Y | X = x, s = j].
\end{align*}
We first focus on the identification of $\bE [Y^\star | X = x,  s = j]$.
By Assumptions \ref{as:IV}(i) and \ref{as:membership}(i), 
\begin{align*}
	\bE [D^\star Y^{(1)} | X = x, \mathbf{P}^\star = \mathbf{p}^\star, s = j ] = \int_0^1 \mathbf{1}\{p_j^\star \ge v_j\} m_j^{(1)}(x, v_j) \mathrm{d}v_j.
\end{align*}
Similarly, we can show that 
\begin{align*}
	\bE [(1 - D^\star) Y^{(0)} | X = x, \mathbf{P}^\star = \mathbf{p}^\star, s = j ] = \int_0^1 \mathbf{1}\{p_j^\star < v_j\} m_j^{(0)}(x, v_j) \mathrm{d}v_j.
\end{align*}
As a result, by the law of iterated expectations, we obtain
\begin{align*}
	\bE [Y^\star | X = x,  s = j] 
	&=  \bE \left[ \bE [D^\star Y^{(1)} + (1 - D^\star)Y^{(0)} | X = x, \mathbf{P}^\star, s = j] \Big|  X = x, s = j \right]\\
	&= \int_0^1 \Big( \Pr(P_j^\star \ge v_j | X = x) m_j^{(1)}(x, v_j) + \Pr(P_j^\star < v_j | X = x) m_j^{(0)}(x, v_j) \Big) \mathrm{d}v_j.
	%&= \sum_{j = 1}^S \pi_j \int_0^1 \Big( m_j^{(1)}(x, v_j) - \Pr(P_j^\star < v_j) [ m_j^{(1)}(x, v_j) - m_j^{(0)}(x, v_j)] \Big) \mathrm{d}v_j\\
\end{align*}
Further, $\bE [Y | X = x,  s = j]$ can be identified in an exactly analogous manner.
Thus, we can identify the PRTE through the MTR functions. 

\section{Appendix: Additional Mote Carlo Experiments} \label{sec:appendix:simulation}

\subsection{Other specifications of the finite mixture model}

In this experiment, in addition to the finite mixture Probit specification in Section \ref{sec:simulation}, we consider two alternative specifications: (i) the finite mixture Logit model, where $\epsilon_j^D \sim \Lambda(0, 1)$ for both $j \in \{ 1, 2 \}$, where $\Lambda(0, 1)$ denotes the standard logistic distribution, and (ii) the Probit-Logit mixture model, where $\epsilon_1^D \sim N(0, 1)$ and $\epsilon_2^D \sim \Lambda(0, 1)$. 

The simulation results are presented in Tables \ref{table:mc2} and \ref{table:mc3}.
The performances of the MTE estimation and ML estimation are satisfactory, particularly when $n = 4000$.
These results might imply that the normal distribution assumption is not necessary for identifying the parameters of the finite mixture model.

\subsection{Violation of Assumption \ref{as:IV}(ii)}

We consider a situation in which there are no group-specific IVs such that Assumption \ref{as:IV}(ii) is actually violated.
Specifically, we generate the treatment variable $D = \mathbf{1}\{ Z^\top \gamma_j \ge \epsilon_j^D \}$ for $s = j$, where $Z = (1, X_1, \zeta_1, \zeta_2)^\top$, $X_1 \sim N(0, 1)$, $\zeta_j \sim N(0, 1)$, and $\epsilon_j^D \sim N(0, 1)$ for both $j \in \{ 1, 2 \}$.
Notice that $Z$ is common to both groups.
We set $\gamma_1  = (0, -0.5, 0.5, -0.5)^\top$ and $\gamma_2 = (0, 0.5, -0.5, 0.5)^\top$.
We conduct the same estimation procedure as in Section \ref{sec:simulation}, in which we treat $\zeta_j$ improperly as a group-$j$-specific IV.

Table \ref{table:mc4} presents the simulation results.
We focus on feasible MTE estimation because the infeasible one is of no interest in this case.
Because $Z$ does not possess identification power for the group-wise MTE, the MTE estimation exhibits poor performance in this situation, which corroborates our theory.

\subsection{Misspecified number of unobserved groups}\label{subsec:misspec}

Considering the same two-group model as in Section \ref{sec:simulation}, we estimate the MTE and the finite mixture Probit model with improperly setting $S = 1$ or $3$.
Specifically, we estimate the misspecified treatment choice equation $D = \mathbf{1} \{ Z_j^\top \gamma_j \ge \epsilon_j^D \}$ for $s = j$, where $Z_j = (1, X_1, \zeta_j)^\top$, $\zeta_j \sim N(0, 1)$, and $\epsilon_j^D \sim N(0, 1)$ for $j \in \{ 1, \dots, S \}$ with $S \in \{1, 3\}$.
In addition, for each estimated model (including the correct model), we compute the values of AIC and BIC to examine whether we can correctly select the model with $S = 2$ based on these information criteria.

The estimation results are presented in Tables \ref{table:mc5} and \ref{table:mc6}.
As the true value of each estimator cannot be computed in this situation, we report the mean and standard deviation (SD) of each estimator, instead of its bias and RMSE.
For comparison, Table \ref{table:true1} provides each parameter value for the correctly specified model under $S = 2$.
As the infeasible MTE estimation is of no interest in this case, we provide the simulation results only for the feasible one.
When setting $S = 1$, the results are clearly uninformative for the true model parameters (except for the common intercept in the treatment choice model $\gamma_{11}$).
In contrast, interestingly, the estimation with setting $S = 3$ seems to work for groups 1 and 2.
However, the estimates for ``group 3'' can be unreasonably large in magnitude with a large SD.
This result reflects the fact that this group is fictitious.

The result of model selection is summarized in Table \ref{table.AIC.BIC}.
From this table, interestingly, we can observe that the optimal models based on BIC are often too parsimonious when $n$ is not large.
However, if $n$ is sufficiently large, the case with $S =2$ is chosen as the best model with a sufficiently high probability in terms of both AIC and BIC.

\begin{table}[!p]
	\caption{Supplementary simulation results: the finite mixture Logit model} \label{table:mc2}
	\begin{subtable}{\textwidth}
		\caption{MTE estimation}
		{\footnotesize
			\begin{center}
				\begin{tabular}{rrrrcrrrrcrrrr}
					\hline\hline
					& \multicolumn{3}{c}{\bfseries }&\multicolumn{1}{c}{\bfseries }&\multicolumn{4}{c}{\bfseries Group 1}&\multicolumn{1}{c}{\bfseries }&\multicolumn{4}{c}{\bfseries Group 2}\tabularnewline
					\cline{6-9} \cline{11-14}
					& \multicolumn{1}{c}{$n$}&\multicolumn{1}{c}{$\tilde K$}&\multicolumn{1}{c}{ridge}&\multicolumn{1}{c}{}&\multicolumn{1}{c}{MTE1.1}&\multicolumn{1}{c}{MTE1.2}&\multicolumn{1}{c}{MTE1.3}&\multicolumn{1}{c}{MTE1.4}&\multicolumn{1}{c}{}&\multicolumn{1}{c}{MTE2.1}&\multicolumn{1}{c}{MTE2.2}&\multicolumn{1}{c}{MTE2.3}&\multicolumn{1}{c}{MTE2.4}\tabularnewline
					\hline
					\multicolumn{14}{l}{\bfseries Bias for the feasible estimator}\tabularnewline
					&$1000$&$1$&$0$&&$ 0.285$&$-0.042$&$ 0.028$&$ 0.656$&&$-0.269$&$-0.127$&$ 0.212$&$ 0.379$\tabularnewline
					&$1000$&$1$&$1$&&$-0.150$&$-0.058$&$-0.126$&$-0.223$&&$-0.340$&$-0.153$&$-0.089$&$-0.190$\tabularnewline
					&$4000$&$1$&$0$&&$ 0.088$&$ 0.014$&$-0.052$&$ 0.175$&&$ 0.081$&$ 0.077$&$ 0.073$&$ 0.192$\tabularnewline
					&$4000$&$1$&$1$&&$-0.122$&$ 0.003$&$-0.067$&$-0.119$&&$-0.269$&$-0.055$&$ 0.003$&$-0.111$\tabularnewline
					&$4000$&$2$&$0$&&$ 0.153$&$-0.016$&$-0.017$&$ 0.138$&&$ 0.130$&$ 0.135$&$ 0.092$&$ 0.390$\tabularnewline
					&$4000$&$2$&$1$&&$-0.152$&$-0.007$&$-0.048$&$-0.155$&&$-0.261$&$-0.042$&$-0.004$&$-0.145$\tabularnewline
					\hline
					\multicolumn{14}{l}{\bfseries Bias for the infeasible estimator}\tabularnewline
					&$1000$&$1$&$0$&&$ 0.060$&$ 0.001$&$-0.009$&$ 0.040$&&$-0.221$&$-0.031$&$ 0.121$&$ 0.236$\tabularnewline
					&$1000$&$1$&$1$&&$-0.371$&$-0.136$&$-0.130$&$-0.352$&&$-0.371$&$-0.154$&$-0.133$&$-0.309$\tabularnewline
					&$4000$&$1$&$0$&&$-0.006$&$ 0.011$&$-0.010$&$-0.071$&&$ 0.018$&$ 0.002$&$ 0.001$&$ 0.018$\tabularnewline
					&$4000$&$1$&$1$&&$-0.265$&$-0.036$&$-0.031$&$-0.249$&&$-0.279$&$-0.069$&$-0.056$&$-0.238$\tabularnewline
					&$4000$&$2$&$0$&&$-0.027$&$ 0.027$&$-0.035$&$-0.026$&&$-0.047$&$ 0.033$&$-0.055$&$ 0.096$\tabularnewline
					&$4000$&$2$&$1$&&$-0.323$&$ 0.022$&$-0.060$&$-0.268$&&$-0.316$&$-0.018$&$-0.071$&$-0.271$\tabularnewline
					\hline
					\multicolumn{14}{l}{\bfseries RMSE for the feasible estimator}\tabularnewline
					&$1000$&$1$&$0$&&$ 3.533$&$ 1.549$&$ 1.758$&$ 4.657$&&$13.014$&$ 5.375$&$ 6.792$&$14.376$\tabularnewline
					&$1000$&$1$&$1$&&$ 0.785$&$ 0.708$&$ 0.735$&$ 0.908$&&$ 1.105$&$ 1.046$&$ 1.027$&$ 1.090$\tabularnewline
					&$4000$&$1$&$0$&&$ 1.707$&$ 0.624$&$ 0.724$&$ 2.072$&&$ 2.600$&$ 1.144$&$ 1.354$&$ 3.542$\tabularnewline
					&$4000$&$1$&$1$&&$ 0.652$&$ 0.395$&$ 0.440$&$ 0.864$&&$ 0.798$&$ 0.607$&$ 0.570$&$ 0.874$\tabularnewline
					&$4000$&$2$&$0$&&$ 2.474$&$ 0.842$&$ 0.980$&$ 2.921$&&$ 3.481$&$ 1.441$&$ 1.667$&$ 4.962$\tabularnewline
					&$4000$&$2$&$1$&&$ 0.651$&$ 0.562$&$ 0.699$&$ 0.877$&&$ 0.845$&$ 0.851$&$ 0.768$&$ 0.946$\tabularnewline
					\hline
					\multicolumn{14}{l}{\bfseries RMSE for the infeasible estimator}\tabularnewline
					&$1000$&$1$&$0$&&$ 3.305$&$ 0.942$&$ 1.204$&$ 4.328$&&$ 4.868$&$ 1.618$&$ 1.397$&$ 5.686$\tabularnewline
					&$1000$&$1$&$1$&&$ 0.805$&$ 0.613$&$ 0.660$&$ 1.082$&&$ 1.052$&$ 0.902$&$ 0.891$&$ 1.065$\tabularnewline
					&$4000$&$1$&$0$&&$ 1.581$&$ 0.448$&$ 0.574$&$ 1.964$&&$ 2.558$&$ 0.836$&$ 0.691$&$ 2.787$\tabularnewline
					&$4000$&$1$&$1$&&$ 0.714$&$ 0.342$&$ 0.388$&$ 0.986$&&$ 0.909$&$ 0.533$&$ 0.502$&$ 0.985$\tabularnewline
					&$4000$&$2$&$0$&&$ 2.143$&$ 0.944$&$ 1.033$&$ 2.827$&&$ 3.643$&$ 1.535$&$ 1.432$&$ 3.592$\tabularnewline
					&$4000$&$2$&$1$&&$ 0.713$&$ 0.653$&$ 0.749$&$ 1.059$&&$ 0.917$&$ 0.999$&$ 0.893$&$ 1.016$\tabularnewline
					\hline
				\end{tabular}
			\end{center}
			Note: The column labeled ``ridge'' indicates whether the ridge regression is used (1 for ``yes'' and 0 for ``no'').}
	\end{subtable}
	
	\bigskip \bigskip
	
	\begin{subtable}{\textwidth}
		\caption{ML estimation of the finite mixture Logit model}
		{\footnotesize
			\begin{center}
				\begin{tabular}{rrcrrrcrrrcrr}
					\hline\hline
					\multicolumn{2}{c}{\bfseries }&&\multicolumn{3}{c}{\bfseries Group 1}&&\multicolumn{3}{c}{\bfseries Group 2}&&\multicolumn{2}{c}{\bfseries Membership}\tabularnewline
					\cline{4-6} \cline{8-10} \cline{12-13}
					&\multicolumn{1}{c}{$n$}&&\multicolumn{1}{c}{$\gamma_{11}$}&\multicolumn{1}{c}{$\gamma_{12}$}&\multicolumn{1}{c}{$\gamma_{13}$}&&\multicolumn{1}{c}{$\gamma_{21}$}&\multicolumn{1}{c}{$\gamma_{22}$}&\multicolumn{1}{c}{$\gamma_{23}$}&&\multicolumn{1}{c}{$\pi_1$}&\multicolumn{1}{c}{$\pi_2$}\tabularnewline
					\hline
					\multicolumn{13}{l}{\textbf{Bias}}\tabularnewline
					&$1000$&&$-0.005$&$-0.119$&$0.195$&&$-0.003$&$0.118$&$-0.246$&&$-0.025$&$0.025$\tabularnewline
					&$4000$&&$-0.003$&$-0.027$&$0.056$&&$-0.002$&$0.037$&$-0.103$&&$-0.008$&$0.008$\tabularnewline
					\hline
					\multicolumn{13}{l}{\textbf{RMSE}}\tabularnewline
					&$1000$&&$ 0.589$&$ 0.594$&$0.425$&&$ 0.778$&$0.849$&$ 0.567$&&$ 0.147$&$0.147$\tabularnewline
					&$4000$&&$ 0.306$&$ 0.306$&$0.178$&&$ 0.444$&$0.440$&$ 0.281$&&$ 0.115$&$0.115$\tabularnewline
					\hline
				\end{tabular}
			\end{center}
		}
	\end{subtable}
\end{table}

\begin{table}[!p]
	\caption{Supplementary simulation results: the Probit-Logit mixture model} \label{table:mc3}
	\begin{subtable}{\textwidth}
		\caption{MTE estimation}
		{\footnotesize
			\begin{center}
				\begin{tabular}{rrrrcrrrrcrrrr}
					\hline\hline
					& \multicolumn{3}{c}{\bfseries }&\multicolumn{1}{c}{\bfseries }&\multicolumn{4}{c}{\bfseries Group 1}&\multicolumn{1}{c}{\bfseries }&\multicolumn{4}{c}{\bfseries Group 2}\tabularnewline
					\cline{6-9} \cline{11-14}
					& \multicolumn{1}{c}{$n$}&\multicolumn{1}{c}{$\tilde K$}&\multicolumn{1}{c}{ridge}&\multicolumn{1}{c}{}&\multicolumn{1}{c}{MTE1.1}&\multicolumn{1}{c}{MTE1.2}&\multicolumn{1}{c}{MTE1.3}&\multicolumn{1}{c}{MTE1.4}&\multicolumn{1}{c}{}&\multicolumn{1}{c}{MTE2.1}&\multicolumn{1}{c}{MTE2.2}&\multicolumn{1}{c}{MTE2.3}&\multicolumn{1}{c}{MTE2.4}\tabularnewline
					\hline
					\multicolumn{14}{l}{\bfseries Bias for the feasible estimator}\tabularnewline
					&$1000$&$1$&$0$&&$ 0.330$&$-0.032$&$-0.222$&$ 0.150$&&$-0.298$&$ 0.068$&$ 0.617$&$ 1.455$\tabularnewline
					&$1000$&$1$&$1$&&$ 0.030$&$-0.018$&$-0.128$&$-0.097$&&$-0.245$&$-0.116$&$-0.070$&$-0.104$\tabularnewline
					&$4000$&$1$&$0$&&$ 0.064$&$-0.002$&$-0.090$&$ 0.029$&&$ 0.000$&$ 0.025$&$ 0.078$&$ 0.269$\tabularnewline
					&$4000$&$1$&$1$&&$-0.035$&$-0.005$&$-0.067$&$-0.036$&&$-0.278$&$-0.063$&$-0.006$&$-0.102$\tabularnewline
					&$4000$&$2$&$0$&&$ 0.107$&$ 0.012$&$-0.104$&$-0.015$&&$-0.027$&$ 0.045$&$ 0.082$&$ 0.530$\tabularnewline
					&$4000$&$2$&$1$&&$-0.043$&$ 0.023$&$-0.093$&$-0.064$&&$-0.272$&$-0.061$&$-0.014$&$-0.124$\tabularnewline
					\hline
					\multicolumn{14}{l}{\bfseries Bias for the infeasible estimator}\tabularnewline
					&$1000$&$1$&$0$&&$ 0.051$&$ 0.002$&$-0.018$&$ 0.003$&&$-0.204$&$-0.013$&$ 0.115$&$ 0.175$\tabularnewline
					&$1000$&$1$&$1$&&$-0.194$&$-0.049$&$-0.041$&$-0.165$&&$-0.442$&$-0.195$&$-0.171$&$-0.371$\tabularnewline
					&$4000$&$1$&$0$&&$ 0.016$&$ 0.010$&$ 0.001$&$-0.013$&&$-0.012$&$-0.006$&$ 0.001$&$ 0.012$\tabularnewline
					&$4000$&$1$&$1$&&$-0.093$&$ 0.003$&$ 0.009$&$-0.073$&&$-0.325$&$-0.088$&$-0.069$&$-0.268$\tabularnewline
					&$4000$&$2$&$0$&&$ 0.006$&$ 0.023$&$-0.035$&$ 0.010$&&$-0.076$&$ 0.019$&$-0.045$&$ 0.081$\tabularnewline
					&$4000$&$2$&$1$&&$-0.124$&$ 0.051$&$-0.037$&$-0.066$&&$-0.377$&$-0.044$&$-0.074$&$-0.315$\tabularnewline
					\hline
					\multicolumn{14}{l}{\bfseries RMSE for the feasible estimator}\tabularnewline
					&$1000$&$1$&$0$&&$ 1.664$&$ 0.925$&$ 0.993$&$ 1.945$&&$ 8.714$&$ 4.825$&$ 8.553$&$16.465$\tabularnewline
					&$1000$&$1$&$1$&&$ 0.724$&$ 0.601$&$ 0.612$&$ 0.854$&&$ 1.015$&$ 0.992$&$ 0.980$&$ 1.022$\tabularnewline
					&$4000$&$1$&$0$&&$ 0.774$&$ 0.399$&$ 0.457$&$ 0.889$&&$ 2.598$&$ 1.083$&$ 1.202$&$ 3.642$\tabularnewline
					&$4000$&$1$&$1$&&$ 0.526$&$ 0.324$&$ 0.359$&$ 0.670$&&$ 0.712$&$ 0.576$&$ 0.550$&$ 0.859$\tabularnewline
					&$4000$&$2$&$0$&&$ 0.957$&$ 0.625$&$ 0.704$&$ 1.061$&&$ 3.663$&$ 1.433$&$ 1.523$&$ 5.105$\tabularnewline
					&$4000$&$2$&$1$&&$ 0.547$&$ 0.462$&$ 0.567$&$ 0.713$&&$ 0.744$&$ 0.860$&$ 0.773$&$ 0.894$\tabularnewline
					\hline
					\multicolumn{14}{l}{\bfseries RMSE for the infeasible estimator}\tabularnewline
					&$1000$&$1$&$0$&&$ 1.631$&$ 0.692$&$ 0.818$&$ 1.910$&&$ 3.903$&$ 1.456$&$ 1.359$&$ 5.447$\tabularnewline
					&$1000$&$1$&$1$&&$ 0.696$&$ 0.491$&$ 0.516$&$ 0.947$&&$ 0.979$&$ 0.816$&$ 0.805$&$ 0.962$\tabularnewline
					&$4000$&$1$&$0$&&$ 0.769$&$ 0.333$&$ 0.394$&$ 0.924$&&$ 1.967$&$ 0.751$&$ 0.665$&$ 2.659$\tabularnewline
					&$4000$&$1$&$1$&&$ 0.529$&$ 0.290$&$ 0.316$&$ 0.707$&&$ 0.827$&$ 0.503$&$ 0.489$&$ 0.932$\tabularnewline
					&$4000$&$2$&$0$&&$ 0.904$&$ 0.778$&$ 0.847$&$ 1.142$&&$ 2.761$&$ 1.448$&$ 1.388$&$ 3.402$\tabularnewline
					&$4000$&$2$&$1$&&$ 0.544$&$ 0.565$&$ 0.662$&$ 0.800$&&$ 0.856$&$ 0.939$&$ 0.876$&$ 0.938$\tabularnewline
					\hline
				\end{tabular}
			\end{center}
			Note: The column labeled ``ridge'' indicates whether the ridge regression is used (1 for ``yes'' and 0 for ``no'').}
	\end{subtable}
	
	\bigskip \bigskip
	
	\begin{subtable}{\textwidth}
		\caption{ML estimation of the Probit-Logit mixture model}
		{\footnotesize
			\begin{center}
				\begin{tabular}{rrcrrrcrrrcrr}
					\hline\hline
					\multicolumn{2}{c}{\bfseries }&&\multicolumn{3}{c}{\bfseries Group 1}&&\multicolumn{3}{c}{\bfseries Group 2}&&\multicolumn{2}{c}{\bfseries Membership}\tabularnewline
					\cline{4-6} \cline{8-10} \cline{12-13}
					&\multicolumn{1}{c}{$n$}&&\multicolumn{1}{c}{$\gamma_{11}$}&\multicolumn{1}{c}{$\gamma_{12}$}&\multicolumn{1}{c}{$\gamma_{13}$}&&\multicolumn{1}{c}{$\gamma_{21}$}&\multicolumn{1}{c}{$\gamma_{22}$}&\multicolumn{1}{c}{$\gamma_{23}$}&&\multicolumn{1}{c}{$\pi_1$}&\multicolumn{1}{c}{$\pi_2$}\tabularnewline
					\hline
					\multicolumn{13}{l}{\textbf{Bias}}\tabularnewline
					&$1000$&&$ 0.000$&$-0.038$&$0.192$&&$-0.016$&$0.048$&$-0.231$&&$-0.027$&$0.027$\tabularnewline
					&$4000$&&$-0.004$&$-0.008$&$0.057$&&$ 0.009$&$0.036$&$-0.066$&&$-0.016$&$0.016$\tabularnewline
					\hline
					\multicolumn{13}{l}{\textbf{RMSE}}\tabularnewline
					&$1000$&&$ 0.372$&$ 0.389$&$0.419$&&$ 0.751$&$0.870$&$ 0.575$&&$ 0.147$&$0.147$\tabularnewline
					&$4000$&&$ 0.169$&$ 0.170$&$0.169$&&$ 0.376$&$0.445$&$ 0.240$&&$ 0.105$&$0.105$\tabularnewline
					\hline
				\end{tabular}
			\end{center}
		}
	\end{subtable}
\end{table}

\begin{table}[!p]
	\caption{Supplementary simulation results: violation of Assumption \ref{as:IV}(ii)} \label{table:mc4}
	\begin{subtable}{\textwidth}
		\caption{MTE estimation}
		{\footnotesize
			\begin{center}
				\begin{tabular}{rrrrcrrrrcrrrr}
					\hline\hline
					& \multicolumn{3}{c}{\bfseries }&\multicolumn{1}{c}{\bfseries }&\multicolumn{4}{c}{\bfseries Group 1}&\multicolumn{1}{c}{\bfseries }&\multicolumn{4}{c}{\bfseries Group 2}\tabularnewline
					\cline{6-9} \cline{11-14}
					& \multicolumn{1}{c}{$n$}&\multicolumn{1}{c}{$\tilde K$}&\multicolumn{1}{c}{ridge}&\multicolumn{1}{c}{}&\multicolumn{1}{c}{MTE1.1}&\multicolumn{1}{c}{MTE1.2}&\multicolumn{1}{c}{MTE1.3}&\multicolumn{1}{c}{MTE1.4}&\multicolumn{1}{c}{}&\multicolumn{1}{c}{MTE2.1}&\multicolumn{1}{c}{MTE2.2}&\multicolumn{1}{c}{MTE2.3}&\multicolumn{1}{c}{MTE2.4}\tabularnewline
					\hline
					\multicolumn{14}{l}{\bfseries Bias for the feasible estimator}\tabularnewline
					&$1000$&$1$&$0$&&$ 5.628$&$ 3.411$&$ 3.776$&$ 6.724$&&$ 4.482$&$ 2.624$&$ 1.950$&$ 0.797$\tabularnewline
					&$1000$&$1$&$1$&&$-0.038$&$-0.146$&$-0.534$&$-1.112$&&$ 2.292$&$ 1.967$&$ 1.657$&$ 1.125$\tabularnewline
					&$4000$&$1$&$0$&&$ 4.736$&$ 2.695$&$ 2.191$&$ 4.662$&&$ 7.362$&$ 3.747$&$ 1.213$&$-1.658$\tabularnewline
					&$4000$&$1$&$1$&&$ 0.709$&$ 0.458$&$-0.485$&$-1.435$&&$ 4.746$&$ 3.405$&$ 2.277$&$ 0.492$\tabularnewline
					&$4000$&$2$&$0$&&$ 4.285$&$ 1.904$&$ 1.776$&$ 2.358$&&$ 7.704$&$ 3.883$&$ 1.225$&$-1.859$\tabularnewline
					&$4000$&$2$&$1$&&$ 0.796$&$ 0.250$&$-0.130$&$-1.211$&&$ 4.663$&$ 3.453$&$ 1.587$&$ 0.015$\tabularnewline
					\hline
					\multicolumn{14}{l}{\bfseries RMSE for the feasible estimator}\tabularnewline
					&$1000$&$1$&$0$&&$34.018$&$21.140$&$36.779$&$76.664$&&$16.030$&$14.000$&$23.558$&$38.531$\tabularnewline
					&$1000$&$1$&$1$&&$ 2.000$&$ 1.892$&$ 1.997$&$ 2.334$&&$ 3.481$&$ 3.259$&$ 3.102$&$ 2.924$\tabularnewline
					&$4000$&$1$&$0$&&$10.576$&$ 4.941$&$ 6.140$&$13.946$&&$11.333$&$ 5.810$&$ 4.683$&$ 9.351$\tabularnewline
					&$4000$&$1$&$1$&&$ 1.797$&$ 1.411$&$ 1.594$&$ 2.530$&&$ 5.301$&$ 4.099$&$ 3.279$&$ 2.501$\tabularnewline
					&$4000$&$2$&$0$&&$10.468$&$ 4.142$&$ 4.919$&$11.869$&&$11.418$&$ 5.665$&$ 5.038$&$10.252$\tabularnewline
					&$4000$&$2$&$1$&&$ 1.946$&$ 1.586$&$ 1.887$&$ 2.700$&&$ 5.220$&$ 4.530$&$ 3.388$&$ 2.491$\tabularnewline
					\hline
				\end{tabular}
			\end{center}
			Note: The column labeled ``ridge'' indicates whether the ridge regression is used (1 for ``yes'' and 0 for ``no'').}
	\end{subtable}
	
	\bigskip \bigskip
	
	\begin{subtable}{\textwidth}
		\caption{ML estimation of the finite mixture Probit model}
		{\footnotesize
			\begin{center}
				\begin{tabular}{rrcrrrcrrrcrr}
					\hline\hline
					\multicolumn{2}{c}{\bfseries }&&\multicolumn{3}{c}{\bfseries Group 1}&&\multicolumn{3}{c}{\bfseries Group 2}&&\multicolumn{2}{c}{\bfseries Membership}\tabularnewline
					\cline{4-6} \cline{8-10} \cline{12-13}
					&\multicolumn{1}{c}{$n$}&&\multicolumn{1}{c}{$\gamma_{11}$}&\multicolumn{1}{c}{$\gamma_{12}$}&\multicolumn{1}{c}{$\gamma_{13}$}&&\multicolumn{1}{c}{$\gamma_{21}$}&\multicolumn{1}{c}{$\gamma_{22}$}&\multicolumn{1}{c}{$\gamma_{23}$}&&\multicolumn{1}{c}{$\pi_1$}&\multicolumn{1}{c}{$\pi_2$}\tabularnewline
					\hline
					\multicolumn{13}{l}{\textbf{Bias}}\tabularnewline
					&$1000$&&$0.007$&$-0.193$&$-0.266$&&$ 0.012$&$ 0.272$&$0.059$&&$-0.024$&$0.024$\tabularnewline
					&$4000$&&$0.015$&$-0.008$&$-0.326$&&$-0.021$&$-0.018$&$0.228$&&$-0.027$&$0.027$\tabularnewline
					\hline
					\multicolumn{13}{l}{\textbf{RMSE}}\tabularnewline
					&$1000$&&$0.623$&$ 0.576$&$ 0.352$&&$ 0.855$&$ 0.899$&$0.480$&&$ 0.155$&$0.155$\tabularnewline
					&$4000$&&$0.344$&$ 0.302$&$ 0.337$&&$ 0.482$&$ 0.451$&$0.287$&&$ 0.129$&$0.129$\tabularnewline
					\hline
				\end{tabular}
			\end{center}
		}
	\end{subtable}
\end{table}

\begin{table}[!p]
	\caption{Supplementary simulation results: $S = 1$} \label{table:mc5}
	\begin{subtable}{\textwidth}
		\caption{MTE estimation}
		{\footnotesize  
			\begin{center}
				\begin{tabular}{rrrrcrrrr}
					\hline\hline
					&\multicolumn{3}{c}{\bfseries }&\multicolumn{1}{c}{\bfseries }&\multicolumn{4}{c}{\bfseries Group 1}\tabularnewline
					\cline{6-9}
					&\multicolumn{1}{c}{$n$}&\multicolumn{1}{c}{$\tilde K$}&\multicolumn{1}{c}{ridge}&\multicolumn{1}{c}{}&\multicolumn{1}{c}{MTE1.1}&\multicolumn{1}{c}{MTE1.2}&\multicolumn{1}{c}{MTE1.3}&\multicolumn{1}{c}{MTE1.4}\tabularnewline
					\hline
					\multicolumn{9}{l}{\bfseries Mean of the feasible estimator}\tabularnewline
					&$1000$&$1$&$0$&&$3.506$&$1.465$&$1.064$&$ 2.369$\tabularnewline
					&$1000$&$1$&$1$&&$1.714$&$1.221$&$0.918$&$ 0.813$\tabularnewline
					&$4000$&$1$&$0$&&$3.308$&$1.469$&$1.038$&$ 2.006$\tabularnewline
					&$4000$&$1$&$1$&&$2.263$&$1.334$&$0.986$&$ 1.218$\tabularnewline
					&$4000$&$2$&$0$&&$3.437$&$0.747$&$1.365$&$-1.151$\tabularnewline
					&$4000$&$2$&$1$&&$1.812$&$0.935$&$1.239$&$ 0.714$\tabularnewline
					\hline
					\multicolumn{9}{l}{\bfseries SD of the feasible estimator}\tabularnewline
					&$1000$&$1$&$0$&&$3.286$&$0.847$&$1.165$&$ 4.838$\tabularnewline
					&$1000$&$1$&$1$&&$0.428$&$0.375$&$0.384$&$ 0.686$\tabularnewline
					&$4000$&$1$&$0$&&$1.587$&$0.395$&$0.529$&$ 2.296$\tabularnewline
					&$4000$&$1$&$1$&&$0.516$&$0.236$&$0.285$&$ 0.899$\tabularnewline
					&$4000$&$2$&$0$&&$2.417$&$0.580$&$0.589$&$ 3.335$\tabularnewline
					&$4000$&$2$&$1$&&$0.384$&$0.491$&$0.509$&$ 0.690$\tabularnewline
					\hline
			\end{tabular}\end{center}
			Note: The column labeled ``ridge'' indicates whether the ridge regression is used (1 for ``yes'' and 0 for ``no'').}
	\end{subtable}
	
	\bigskip \bigskip 
	
	\begin{subtable}{\textwidth}
		\caption{ML estimation of the Probit model} 
		{\footnotesize
			\begin{center}
				\begin{tabular}{rrcrrr}
					\hline\hline
					&\multicolumn{1}{c}{\bfseries }&\multicolumn{1}{c}{\bfseries }&\multicolumn{3}{c}{\bfseries Group 1}\tabularnewline
					\cline{4-6}
					&\multicolumn{1}{c}{$n$}&\multicolumn{1}{c}{}&\multicolumn{1}{c}{$\gamma_{11}$}&\multicolumn{1}{c}{$\gamma_{12}$}&\multicolumn{1}{c}{$\gamma_{13}$}\tabularnewline
					\hline
					\multicolumn{6}{l}{\bfseries Mean}\tabularnewline
					&$1000$&&$0.000$&$-0.086$&$0.255$\tabularnewline
					&$4000$&&$0.000$&$-0.084$&$0.252$\tabularnewline
					\hline
					\multicolumn{6}{l}{\bfseries SD}\tabularnewline
					&$1000$&&$0.040$&$ 0.040$&$0.042$\tabularnewline
					&$4000$&&$0.020$&$ 0.021$&$0.021$\tabularnewline
					\hline
		\end{tabular}\end{center}}
	\end{subtable}
\end{table}

\begin{landscape}
	\begin{table}[!p]
		\caption{Supplementary simulation results: $S = 3$} \label{table:mc6}
		\begin{subtable}{\textwidth}
			\caption{MTE estimation}
			{\footnotesize 
				\begin{center}
					\begin{tabular}{rrrrcrrrrcrrrrcrrrr}
						\hline\hline
						&\multicolumn{3}{c}{\bfseries }&\multicolumn{1}{c}{\bfseries }&\multicolumn{4}{c}{\bfseries Group 1}&\multicolumn{1}{c}{\bfseries }&\multicolumn{4}{c}{\bfseries Group 2}&\multicolumn{1}{c}{\bfseries }&\multicolumn{4}{c}{\bfseries Group 3}\tabularnewline
						\cline{6-9} \cline{11-14} \cline{16-19}
						&\multicolumn{1}{c}{$n$}&\multicolumn{1}{c}{$\tilde K$}&\multicolumn{1}{c}{ridge}&\multicolumn{1}{c}{}&\multicolumn{1}{c}{MTE1.1}&\multicolumn{1}{c}{MTE1.2}&\multicolumn{1}{c}{MTE1.3}&\multicolumn{1}{c}{MTE1.4}&\multicolumn{1}{c}{}&\multicolumn{1}{c}{MTE2.1}&\multicolumn{1}{c}{MTE2.2}&\multicolumn{1}{c}{MTE2.3}&\multicolumn{1}{c}{MTE2.4}&\multicolumn{1}{c}{}&\multicolumn{1}{c}{MTE3.1}&\multicolumn{1}{c}{MTE3.2}&\multicolumn{1}{c}{MTE3.3}&\multicolumn{1}{c}{MTE3.4}\tabularnewline
						\hline
						\multicolumn{19}{l}{\bfseries Mean of the feasible estimator}\tabularnewline
						&$1000$&$1$&$0$&&$1.184$&$0.800$&$0.671$&$0.933$&&$-0.082$&$0.428$&$0.618$&$0.531$&&$-0.268$&$-0.091$&$ 0.515$&$  2.025$\tabularnewline
						&$1000$&$1$&$1$&&$0.957$&$0.839$&$0.776$&$0.844$&&$ 0.211$&$0.311$&$0.363$&$0.371$&&$-0.232$&$-0.229$&$-0.247$&$ -0.250$\tabularnewline
						&$4000$&$1$&$0$&&$0.928$&$0.779$&$0.735$&$0.939$&&$ 0.176$&$0.348$&$0.445$&$0.381$&&$ 0.621$&$ 0.526$&$ 1.945$&$  5.591$\tabularnewline
						&$4000$&$1$&$1$&&$0.850$&$0.783$&$0.750$&$0.865$&&$ 0.207$&$0.316$&$0.381$&$0.329$&&$-0.284$&$-0.260$&$-0.299$&$ -0.326$\tabularnewline
						&$4000$&$2$&$0$&&$0.969$&$0.865$&$0.641$&$0.959$&&$ 0.261$&$0.394$&$0.393$&$0.435$&&$ 0.728$&$ 0.375$&$ 0.462$&$ -0.581$\tabularnewline
						&$4000$&$2$&$1$&&$0.863$&$0.859$&$0.667$&$0.858$&&$ 0.217$&$0.362$&$0.366$&$0.314$&&$-0.222$&$-0.266$&$-0.286$&$ -0.315$\tabularnewline
						\hline
						\multicolumn{19}{l}{\bfseries SD of the feasible estimator}\tabularnewline
						&$1000$&$1$&$0$&&$1.609$&$1.084$&$1.037$&$1.660$&&$ 2.349$&$1.721$&$1.816$&$2.576$&&$37.150$&$22.749$&$23.658$&$ 45.045$\tabularnewline
						&$1000$&$1$&$1$&&$0.805$&$0.666$&$0.647$&$0.781$&&$ 0.934$&$0.891$&$0.882$&$0.926$&&$ 4.840$&$ 4.841$&$ 4.842$&$  4.843$\tabularnewline
						&$4000$&$1$&$0$&&$0.818$&$0.464$&$0.503$&$0.820$&&$ 1.208$&$0.797$&$0.712$&$1.186$&&$27.763$&$23.265$&$44.625$&$115.593$\tabularnewline
						&$4000$&$1$&$1$&&$0.595$&$0.378$&$0.382$&$0.611$&&$ 0.699$&$0.534$&$0.477$&$0.684$&&$ 2.894$&$ 2.889$&$ 2.899$&$  2.909$\tabularnewline
						&$4000$&$2$&$0$&&$0.928$&$0.692$&$0.765$&$0.968$&&$ 1.498$&$1.012$&$1.032$&$1.488$&&$32.591$&$24.520$&$39.174$&$ 73.612$\tabularnewline
						&$4000$&$2$&$1$&&$0.613$&$0.506$&$0.594$&$0.660$&&$ 0.772$&$0.654$&$0.627$&$0.778$&&$ 2.927$&$ 2.934$&$ 2.956$&$  2.999$\tabularnewline
						\hline
				\end{tabular}\end{center}
				Note: The column labeled ``ridge'' indicates whether the ridge regression is used (1 for ``yes'' and 0 for ``no'').}
		\end{subtable}
		
		\bigskip \bigskip
		
		\begin{subtable}{\textwidth}
			\caption{ML estimation of the finite mixture Probit model} 
			{\footnotesize
				\begin{center}
					\begin{tabular}{rrcrrrcrrrcrrrcrrr}
						\hline\hline
						&\multicolumn{1}{c}{\bfseries }&\multicolumn{1}{c}{\bfseries }&\multicolumn{3}{c}{\bfseries Group 1}&\multicolumn{1}{c}{\bfseries }&\multicolumn{3}{c}{\bfseries Group 2}&\multicolumn{1}{c}{\bfseries }&\multicolumn{3}{c}{\bfseries Group 3}&\multicolumn{1}{c}{\bfseries }&\multicolumn{3}{c}{\bfseries Membership }\tabularnewline
						\cline{4-6} \cline{8-10} \cline{12-14} \cline{16-18}
						&\multicolumn{1}{c}{$n$}&\multicolumn{1}{c}{}&\multicolumn{1}{c}{$\gamma_{11}$}&\multicolumn{1}{c}{$\gamma_{12}$}&\multicolumn{1}{c}{$\gamma_{13}$}&\multicolumn{1}{c}{}&\multicolumn{1}{c}{$\gamma_{21}$}&\multicolumn{1}{c}{$\gamma_{22}$}&\multicolumn{1}{c}{$\gamma_{23}$}&\multicolumn{1}{c}{}&\multicolumn{1}{c}{$\gamma_{31}$}&\multicolumn{1}{c}{$\gamma_{32}$}&\multicolumn{1}{c}{$\gamma_{33}$}&\multicolumn{1}{c}{}&\multicolumn{1}{c}{$\pi_1$}&\multicolumn{1}{c}{$\pi_2$}&\multicolumn{1}{c}{$\pi_3$}\tabularnewline
						\hline
						\multicolumn{18}{l}{\bfseries Mean}\tabularnewline
						&$1000$&&$ 0.005$&$-0.393$&$0.948$&&$-0.050$&$0.282$&$-1.111$&&$-0.005$&$-0.090$&$ 0.000$&&$0.451$&$0.323$&$0.226$\tabularnewline
						&$4000$&&$-0.007$&$-0.474$&$0.670$&&$ 0.007$&$0.426$&$-0.751$&&$-0.017$&$-0.041$&$-0.010$&&$0.500$&$0.331$&$0.169$\tabularnewline
						\hline
						\multicolumn{18}{l}{\bfseries SD}\tabularnewline
						&$1000$&&$ 0.478$&$ 0.592$&$0.474$&&$ 0.766$&$0.957$&$ 0.633$&&$ 1.365$&$ 1.511$&$ 1.008$&&$0.110$&$0.100$&$0.107$\tabularnewline
						&$4000$&&$ 0.189$&$ 0.232$&$0.187$&&$ 0.319$&$0.415$&$ 0.294$&&$ 0.849$&$ 1.046$&$ 0.531$&&$0.085$&$0.075$&$0.087$\tabularnewline
						\hline
			\end{tabular}\end{center}}
		\end{subtable}
	\end{table}
\end{landscape}

\begin{table}
	\caption{True parameter values} \label{table:true1}
	\begin{subtable}{\textwidth}
		\caption{MTE}
		{\footnotesize
			\begin{center}
				\begin{tabular}{rrrrcrrrr}
					\hline \hline
					\multicolumn{4}{c}{\bfseries Group 1} & & \multicolumn{4}{c}{\bfseries Group 2} \tabularnewline \cline{1-4} \cline{6-9}
					MTE1.1 & MTE1.2 & MTE1.3 & MTE1.4 & & MTE2.1 & MTE2.2 & MTE2.3 & MTE2.4 \tabularnewline \hline 
					$0.84$ & $0.76$ & $0.76$ & $0.84$ & & $0.34$ & $0.26$ & $0.26$ & $0.34$ \tabularnewline \hline
				\end{tabular}
		\end{center}}
	\end{subtable}
	
	\bigskip \bigskip
	
	\begin{subtable}{\textwidth}
		\caption{The finite mixture Probit model}
		{\footnotesize
			\begin{center}
				\begin{tabular}{rrrcrrrcrr}
					\hline\hline
					\multicolumn{3}{c}{\bfseries Group 1}&&\multicolumn{3}{c}{\bfseries Group 2}&&\multicolumn{2}{c}{\bfseries Membership}\tabularnewline
					\cline{1-3} \cline{5-7} \cline{9-10}
					\multicolumn{1}{c}{$\gamma_{11}$}&\multicolumn{1}{c}{$\gamma_{12}$}&\multicolumn{1}{c}{$\gamma_{13}$}&&\multicolumn{1}{c}{$\gamma_{21}$}&\multicolumn{1}{c}{$\gamma_{22}$}&\multicolumn{1}{c}{$\gamma_{23}$}&&\multicolumn{1}{c}{$\pi_1$}&\multicolumn{1}{c}{$\pi_2$}\tabularnewline
					\hline
					$0$ & $-0.5$ & $0.5$ && $0$ & $0.5$ & $-0.5$ && $0.6$ & $0.4$ \tabularnewline 
					\hline
				\end{tabular}
		\end{center}}
	\end{subtable}
\end{table}

\begin{table}[!tbp]
	{\footnotesize
		\caption{Supplementary simulation results: model selection} \label{table.AIC.BIC} 
		\begin{center}
			\begin{tabular}{rcrrrcrrr}
				\hline\hline
				\multicolumn{1}{c}{\bfseries }&\multicolumn{1}{c}{\bfseries }&\multicolumn{3}{c}{\bfseries AIC}&\multicolumn{1}{c}{\bfseries }&\multicolumn{3}{c}{\bfseries BIC}\tabularnewline
				\cline{3-5} \cline{7-9}
				\multicolumn{1}{c}{$n$}&\multicolumn{1}{c}{}&\multicolumn{1}{c}{$S=1$}&\multicolumn{1}{c}{$S=2$}&\multicolumn{1}{c}{$S=3$}&\multicolumn{1}{c}{}&\multicolumn{1}{c}{$S=1$}&\multicolumn{1}{c}{$S=2$}&\multicolumn{1}{c}{$S=3$}\tabularnewline
				\hline
				$1000$&&$0.5\%$&$85.8\%$&$13.7\%$&&$59.9\%$&$40.1\%$&$0\%$ \tabularnewline
				$4000$&&$0\%$&$88.5\%$&$11.5\%$&&$0.1\%$&$99.9\%$&$0\%$ \tabularnewline
				\hline
	\end{tabular}\end{center}
	%Note: This table reports the frequency that the number of unobserved groups is selected by each information criterion.
	}
\end{table}

\end{document}